\documentclass{amsart}
\usepackage{amsthm}
\usepackage{amsfonts, amsmath, amssymb}
\usepackage{tabularx}
\usepackage{graphicx}
\usepackage[english]{babel}

\theoremstyle{definition}
\newtheorem{theorem}{Theorem}[section]
\newtheorem{lemma}[theorem]{Lemma}
\newtheorem{proposition}[theorem]{Proposition}
\newtheorem{corollary}[theorem]{Corollary}
\newtheorem{definition}[theorem]{Definition}
\newtheorem{remark}[theorem]{Remark}

\usepackage{amsmath}
\usepackage{amstext}
\usepackage{amssymb}
\usepackage{amsfonts}
\usepackage{enumerate}
\usepackage{mathrsfs}
\usepackage{braket}
\usepackage{bbm}
\usepackage{color}
\usepackage[all]{xy}
\usepackage{ragged2e}
\usepackage{geometry}

\newcommand{\pair}[2]{{\langle #1, #2 \rangle}}

\newcommand{\n}[1]{ \left\|#1\right\| }

\newcommand{\mbb}[1]{\mathbb{#1}}

\newcommand{\mcl}[1]{\mathcal{#1}}

\DeclareMathOperator{\Di}{Diag}
\DeclareMathOperator{\range}{range}
\DeclareMathOperator{\spa}{span}
\DeclareMathOperator{\rank}{rank}
\DeclareMathOperator{\Tr}{Tr}

\DeclareMathOperator{\U}{U}

\DeclareMathOperator{\D}{D}
\DeclareMathOperator{\E}{E}

\DeclareMathOperator{\Li}{L}
\DeclareMathOperator{\B}{B}
\renewcommand{\L}{\Li}

\DeclareMathOperator{\Pos}{Pos}

\DeclareMathOperator{\Herm}{Herm}

\DeclareMathOperator{\Rea}{Re}

\renewcommand{\Re}{\Rea}

\numberwithin{equation}{section}

\title{\textbf{Phase Retrievability of Super Operators and Measurements}} 

\author[J.A. Ch\'avez-Dom\'inguez]{Javier Alejandro Ch\'avez-Dom\'inguez}
\address{Department of Mathematics, University of Oklahoma, Norman, OK 73019-3103,
USA} \email{jachavezd@ou.edu}

\author[J.I. Isu]{Jeremiah Ifeanyi Isu}
\email{isujeremiah@ou.edu}

\thanks{The first-named author was partially supported by NSF grants DMS-2247374 and DMS-2555053.
The second-named author was partially supported by a Dissertation Research Fellowship from the Dodge Family College of Arts and Sciences at the University of Oklahoma.
This paper forms part of the second-named author’s
PhD research at the University of Oklahoma, conducted under the supervision of the first-named author.
}


\date{}

\begin{document}

\begin{abstract}
We continue the study of Phase Retrieval in the Quantum Information setting in the style of \cite{liu2023phase}, in a manner which is more general in two primary ways: (a) Instead of only pure states we consider also states (and positive semidefinite operators) of bounded rank, as is done in Quantum Tomography; (b) We consider general super operators instead of only quantum channels.

We show that in order to have phase retrieval with respect to positive semidefinite operators of bounded rank, it suffices to discriminate between perfectly distinguishable pairs (i.e. those with orthogonal supports). From this we prove that phase retrievability is always Lipschitz stable with respect to the trace norm, and then we use the Fuchs-van de Graaf inequalities to deduce stability with respect to the Bures--Wasserstein distance (thus generalizing the known stability results for phase retrieval in Frame Theory).
For super operators that commute with the adjoint and take values in block-diagonal matrices, we characterize phase retrievability in terms of conditions inspired by the classical complement property from Frame Theory.

\end{abstract}

\maketitle

\section{Introduction}

Quantum Information Theory studies how quantum systems --- the tiny particles that make up our universe --- store, process, and transmit information. Because quantum systems evolve over time, a central question is whether one can infer an earlier state of a system from its final state. This paper examines two related problems. First, given only the final state of a quantum system, can the initial state be recovered? Second, if recovery is possible, can it be carried out stably, so that small errors in the observed data lead only to small errors in the recovered state?  These questions are fundamental for developing reliable quantum technologies, where recovering altered or degraded information is essential. Quantum channels provide a natural framework for studying such state changes. Quantum measurements like quantum channels provide a mechanism for extracting information from a physical system and the two are closely related. We can therefore pose analogous questions for quantum measurements as well. More broadly, we ask when a super operator is guaranteed to be injective on the set of positive semidefinite operators, and how stable such injectivity condition is. Positive semidefinite operators form the basic mathematical representation of nonnegative quantum observables. They include quantum effects, unnormalized quantum states, and measurement outcomes, each characterized by having nonnegative expectation values on all physical states. This nonnegativity is what allows positive semidefinite operators to encode probabilities, likelihoods, and post measurement states in quantum mechanics. In the same way that quantum channels are used to recover quantum states, quantum measurements are used to retrieve information from quantum systems. We show that this relationship extends naturally to a correspondence between super operators that are injective on positive semidefinite operators, and operator valued functions (that is, general measurements) which are used to retrieve information from physical systems.

The point of view of the present work borrows heavily from classical signal processing, which we recall now.
A \emph{signal} is an unknown vector in a linear space, corresponding to real world data.
Often we can only gain information about it through scalar-valued linear measurements, which correspond to measuring devices, and the goal is to reconstruct the signal.
Mathematically, this is often modeled by considering the problem of identifying an element $x$ in a Hilbert space $\mcl{X}$ (over the field $\mathbb{F}$) from a collection of inner products $\{ \pair{y_a}{x} \}_{a \in \Sigma}$.
In a wide range of situations of great practical interest --- such as crystallography, astronomy, medical imaging, computer generated holography, or optical computing --- the physical limitations of real-world measuring instruments imply that we do not have access to the full measurement $\pair{y_a}{x}$ but only to its magnitude $|\pair{y_a}{x}|$, meaning that the phase has been lost.
Recovering the signal from these hampered (\emph{phaseless}) measurements is known as the Phase Retrieval problem \cite{MR2224902,MR3526436,MR4094471,MR4189292}.
It should be noted that one can only hope to recover the signal \emph{up to a phase}: if $x=\lambda z$ for some scalar $\lambda$ with $|\lambda|=1$, the phaseless measurements $|\pair{y_a}{x}|$ and $|\pair{y_a}{z}|$ will always agree. Thus, a collection $\{y_a\}_{a \in \Sigma}$ in $\mcl{X}$ is said to have the \emph{phase retrieval property} if the map $\mcl{X} \to \mathbb{F}^\Sigma$ given by  $x \mapsto \big( |\pair{y_a}{x}| \big)_{a \in \Sigma}$ is injective up to a phase.

In the present work we build upon recent developments \cite{liu2023phase} relating the Phase Retrieval problem to a central one in Quantum Information Theory: telling quantum states apart.
The connection comes from the fact that if $x$ and $z$ are two unit vectors in $\mcl{X}$ that differ by a phase, then the corresponding associated pure (i.e. rank one) states are the same. 
To explain this we will be using the following notation from \cite{watrous2018theory}, which may be unfamiliar to the reader: the canonical identification $\mcl{X} = \L(\mathbb{F}, \mcl{X})$ (where a vector $x \in \mcl{X}$ is identified with the mapping $\lambda \mapsto \lambda x$) allows one to write $x^* \in \L(\mcl{X},\mathbb{F})$ for the adjoint of $x\in \mcl{X}$: this is simply the mapping $\pair{x}{\cdot} : \mcl{X} \to \mathbb{F}$. With this notation, if $\n{x}=1$ then $xx^* : \mcl{X} \to \mcl{X}$ is the orthogonal projection onto $\mathbb{F} x$, that is, the pure state associated to $x$.
Moreover, $x$ and $z$ differ by a unimodular scalar if and only if $xx^*=zz^*$. By rewriting $|\pair{y_a}{x}|^2 = \Tr( y_ay_a^*xx^* ) = \pair{y_ay_a^*}{xx^*}$ (where the latter is the Hilbert-Schmidt inner product on $\L(\mcl{X})$), we see that the collection $\{y_a\}_{a \in \Sigma}$ has the phase retrieval property if and only if the linear map $\L(\mcl{X}) \to \mathbb{F}^\Sigma$ given by $\rho \mapsto \big( \pair{y_ay_a^*}{\rho} \big)_{a\in \Sigma}$ restricts to an injective map on the subset of positive semidefinite operators of rank at most one.
The paper \cite{liu2023phase} defines analogous phase-retrieval-like properties where one replaces the rank-one operators $y_ay_a^*$ by more general operators $Y_a \in \L(\mcl{X})$, with the additional assumption that the collection $\{Y_a\}_{a\in \Sigma}$ is a quantum measurement or positive operator valued measure (i.e. each $Y_a$ is positive semidefinite and they add up to the identity of $\mcl{X}$).
A companion phase-retrieval-like property was defined in \cite{liu2023phase} for special types of linear maps $\Phi : \L(\mcl{X}) \to \L(\mcl{Y})$ --- the \emph{quantum channels} ---  and shown to be equivalent to injectivity on pure states.
The main purpose of the present paper is to build upon \cite{liu2023phase} in three different ways:
\begin{enumerate}[(a)]
    \item\label{new ways:other sets} Consider phase-retrieval-like properties with respect to other subsets of states beyond the pure ones, which is closely related to Quantum Tomography \cite{heinosaari2013quantum}. 
    \item\label{new ways:other maps} Show that the restriction to quantum channels in \cite{liu2023phase} is unnecessary, and the phase retrieval circle of ideas can be developed for general linear maps $\L(\mcl{X}) \to \L(\mcl{Y})$.
    \item\label{new ways:other stability} Study the stability of phase retrieval in the quantum context, with respect to various distances of interest.
\end{enumerate}

We must emphasize the fact that exploiting connections between phase retrievability (in Frame Theory) and Quantum Tomography predates \cite{liu2023phase} by at least a decade, see e.g. the title of \cite{kech2015quantum}.
In this regard the main novelty of our approach lies in \eqref{new ways:other maps} above, as the literature so far has (understandably) concentrated on quantum channels/measurements.
Additionally, we point out that looking at stability questions in the quantum setting also has precedents such as \cite{MR3366275,MR3671475}.
Here a further difference between the present paper and those previous works is that we concentrate on the subset of positive semidefinite operators of rank at most $k$, which does not appear to have received much attention in said previous works. 
\\
\\
To avoid over-complicating the introduction, the reader is referred to Section \ref{Sec: Notation and Prelim} for detailed definitions.
Below, a \emph{measurement} means a collection  $\{Y_a\}_{a\in \Sigma} \subseteq \L(\mcl{X})$. The measurement is said to be Hermitian if each $Y_a$ is a Hermitian operator.
A \emph{super operator} is a linear map $\Phi : \L(\mcl{X})\to\L(\mcl{Y})$.
Given a subset $\Omega \subseteq \L(\mcl{X})$, a measurement $\{Y_a\}_{a\in \Sigma} \subseteq \L(\mcl{X})$ is said to be \emph{$\Omega$-phase retrievable} when the aforementioned map $\L(\mcl{X}) \to \mathbb{F}^\Sigma$ given by $\rho \mapsto \big( \pair{Y_a}{\rho} \big)_{a\in \Sigma}$ restricts to an injective map on the subset $\Omega$.
On the other hand, a super operator $\Phi : \L(\mcl{X})\to\L(\mcl{Y})$ is said to be \emph{$\Omega$-phase retrievable} if there exists a measurement $\{Y_a\}_{a\in \Sigma} \subseteq \L(\mcl{Y})$ such that $\{\Phi^* Y_a\}_{a\in \Sigma} \subseteq \L(\mcl{X})$ is $\Omega$-phase retrievable.
These definitions of phase retrievability directly generalize those of \cite{liu2023phase}, simply with fewer restrictions on $\{Y_a\}_{a\in \Sigma}$ or $\Phi$, and also the notion of informationally complete quantum measurements from \cite{heinosaari2013quantum}.
Our most frequent choice of $\Omega$ will be the set of positive semidefinite operators of rank at most $k$, denoted by $\mathcal{P}_k$.

We now summarize the main contributions of the paper:

\begin{enumerate}[(a)]
    \item We show some equivalent conditions for $\Omega$-phase retrievability of a super operator inspired by a combination of both the phase retrievability of \cite{liu2023phase} and the Quantum Tomography of \cite{heinosaari2013quantum}  (Theorem \ref{thm: injectivity and phase retrievability of a quantum channel}). In particular, we show that in the definition of a phase retrievable super operator the associated measurement can be assumed to satisfy an extra condition.

    \item We demonstrate a correspondence between the $\Omega$-phase retrievability of a super operator and that of an associated measurement (Theorem \ref{thm: equivalence of P-phase retrievable measurements and channels}), and we show how this relationship specializes to quantum channels and their corresponding quantum measurements (Corollary \ref{cor: equivalence of k-phase retrievable measurements and channels}).

    \item In the case where $\Omega$ consists of positive semidefinite operators, we identify  necessary and sufficient conditions under which a quantum measurement is $\Omega$-phase retrievable (Theorem \ref{thm: k-phase retrivable Q-measurements}), and highlight the connection between these conditions and the Bhattacharyya coefficient (Corollary \ref{cor:Bhattacharyya xterization}).

    \item For the special class of super operators called quantum channels, we extend the result in \cite[Thm. 3]{liu2026enhancing} by characterizing their $\Omega$-phase retrievability through properties of their complementary channels, reinterpreting the condition from \cite[Thm. 3]{liu2026enhancing} in terms of the quantum confusability graph.
    
    \item For the specific case of $\Omega=\mcl{P}_k$, we present more equivalent criteria of phase retrievability for super operators (Theorem \ref{thm: prso p_k}) and, in parallel, for super operators that commute with the adjoint (Theorem \ref{thm: Herm-pres}).
    
    \item We establish alternative formulations under which super operators that commute with the adjoint and take values in block-diagonal matrices are $\mcl{P}_k$-phase retrievable, phrased in terms of either the orthogonal complement or the dimension of certain spaces (Theorems \ref{thm: prso and span complex} and \ref{thm: prso and span}).
    This class of super operators in particular includes the quantum channels associated to quantum instruments \cite{heinosaari2011mathematical,jordan2024quantum}.
    These results should be understood as generalizations of the characterization in the real case of phase retrievability for frames in terms of the complement property \cite{MR2224902}, or more precisely to the complex version from \cite[Thm. 4]{MR3202304}.
    Building on these (and previous results), we obtain equivalent conditions for $\mcl{P}_k$-phase retrievability of measurements (Corollary \ref{cor: KPRM}) and, in parallel, for Hermitian measurements (Corollary \ref{cor: KPR Herm M}).
    
    \item We show the existence of quantum channels that are injective on $\mcl{P}_k$ (hence, $\mcl{P}_k$-phase retrievable), but not injective (Theorem \ref{thm: Non-injective PKPR QC}). Our approach is fundamentally different from the case $k=1$ proved in \cite{liu2023phase} which was done using Frame Theory tools, whereas we use results from Quantum Tomography \cite{heinosaari2013quantum}.
    
    \item By refining those Quantum Tomography arguments, we also prove the existence of quantum measurements that are $\mcl{P}_k$-phase retrievable but not $\mcl{P}_{k+1}$-phase retrievable. From this, the existence of quantum channels that are $\mcl{P}_k$-phase retrievable, but not $\mcl{P}_{k+1}$-phase retrievable naturally follows (Theorem \ref{thm: PKPR not PK+1PR QC}).

    \item \label{PKSO othorgonal} We show that in order for a super operator to be $\mcl{P}_k$-phase retrievable, it suffices for it to distinguish pairs of elements of $\mcl{P}_k$ with orthogonal images (Theorem \ref{thm: stability of Pk-PRSO}). This corresponds to the analogous fact that in Frame Theory, distinguishing orthogonal vectors is sufficient to have phase retrievability \cite{MR4727764}.

    \item \label{PKSO minimize} We produce a minimization characterization of the $\mcl{P}_k$-phase retrievability of a super operator (Lemma \ref{lem: Pk-PRSO equiv}), and use this to prove that $\mcl{P}_k$-phase retrievability of super operators is stable under small perturbations (Theorem \ref{thm: stability for PkPR s-operators}), generalizing the analogous result in Frame Theory. It follows that the set of $\mcl{P}_k$-phase retrievable super operators is open (Corollary \ref{cor: set PkPRSO is open}).

    \item \label{PKSO stability} We provide a different stability result showing that a super operator is $\mcl{P}_k$-phase retrievable if and only if it is bi-Lipschitz with respect to the trace norm when restricted to $\mcl{P}_k$ (Theorem \ref{thm: stability Pk-PRSO-input}). In addition, we show in Theorem \ref{thm: PK-PRSO group rep} a more general stability result for a version of $\mcl{P}_k$-phase retrievability that ``respects symmetries''. 

    \item In Section \ref{Sec: Stability of PKPRM}, we develop measurement analogues of the results established in Section \ref{Sec: Stability of PKPRMSO}, specifically those stated in \eqref{PKSO othorgonal} through \eqref{PKSO stability}. This establishes a structural correspondence between $\mcl{P}_k$-phase retrievable super operators and $\mcl{P}_k$-phase retrievable measurements. In particular, we establish a close correspondence between $\mcl{P}_k$-phase retrievable quantum channels and $\mcl{P}_k$-phase retrievable quantum measurements.

    \item Section \ref{Sec: Stability BW Distance} is devoted to establishing how our stability results with respect to norms for $\mcl{P}_k$-phase retrievability of super operators relate to classical stability results for phase retrievability in Frame Theory. The key is to recognize that the usual way to quantify the stability in Frame Theory corresponds to the Bures--Wasserstein distance on the quantum side. In particular, Theorem \ref{thm: BW Fvd Graaf} and Corollary \ref{cor: norm stability plus F-vdG} make use of the famous Fuchs--van de Graaf inequalities to cleanly make a connection between our stability results and the classical ones. In the case $k=1$ we are able to generalize the classical results to obtain Lipschitz stability of $\mcl{P}_1$-phase retrieval for positive super operators (Theorem \ref{thm: stability BW}), whereas for $k>1$ we only get H\"older stability with respect to the Bures--Wasserstein distance on bounded sets  (Corollary \ref{cor: norm stability plus F-vdG}). We remark that in the case $k>1$ it is known that Lipschitz stability with respect to the Bures--Wasserstein distance can never hold for a certain class of super operators by results from \cite{MR4477807}.
We remark that our Theorem \ref{thm: stability BW} 
generalizes the fundamental result in Frame Theory
about stability of phase retrieval in the finite-dimensional setting, for which a multitude of proofs have appeared in the literature
\cite{MR3202304,balan2015invertibility,MR3554699,MR3656501,MR4727764} (some of them including the case of continuous frames).
Our arguments are overall closest in spirit to those from \cite{MR3554699}, although the reduction to orthogonal pairs in Theorem \ref{thm: stability of Pk-PRSO} was inspired by \cite{MR4727764}.
\end{enumerate}

\section{Notation and preliminaries} \label{Sec: Notation and Prelim}

Our notation partly follows that of \cite{watrous2018theory}. For a field $\mbb{F}$, where $\mbb{F}$ is $\mbb{C}$ or $\mbb{R}$, we represent an \emph{$\mbb{F}$-Euclidean space}\footnote{An $\mbb{F}$-Euclidean space is simply a finite dimensional real or complex Hilbert space, representing a physical system.} by $\mcl{X}, \mcl{Y}$, or $\mcl{Z}$.
 Following the Physics convention, our inner products will be conjugate-linear in the first variable. 
In Quantum Information Theory, traditionally one only considers complex scalars. As our approach is inspired by phase retrieval in Frame Theory, where the real case is often simpler, we also consider real scalars whenever possible.
Sometimes the results we reference are originally stated in the setting of complex Euclidean spaces, but we use the formulations provided here because the same arguments also apply in the real case.
One must, however, be mindful of the fact that the real theory is not 100\% the same as the complex one. See \cite{MR4659878} for a discussion of some fundamental differences and their implications in Quantum Information Theory.
We note one crucial difference: in the complex case, the complex span of the Hermitian operators recovers all complex linear operators. However, the real span of the symmetric operators is just the symmetric operators.
This implies that results such as \cite[Lemma 1]{heinosaari2013quantum} do not hold in the real case.

The space of linear operators from $\mcl{X}$ to $\mcl{Y}$ is denoted $\Li(\mcl{X}, \mcl{Y})$ and we write $\Li(\mcl{X})$ if $\mcl{X} = \mcl{Y}$. We denote the elements of $\Li(\mcl{X}, \mcl{Y})$ by block letters $A, U, V, X, \text{ and }Y$. The identity operator on $\mcl{X}$ is denoted $\mathbbm{1}_\mcl{X}$ and the trace operator on $\Li(\mcl{X})$ is denoted $\Tr$.
Unless indicated otherwise, the \emph{Frobenius inner product}\footnote{For $A,B \in \Li(\mcl{X}, \mcl{Y})$, the Frobenius inner product of $A$ and $B$ is defined as $\pair{A}{B} = \Tr(A ^ {*} B)$} is the one used for operators, while the $\mbb{F}$-Euclidean inner product is used for vectors. For any operator $A \in \Li(\mcl{X}, \mcl{Y)}$, the Frobenius norm of $A$ is denoted $\n{A}_{F} = \Tr(A^*A)^{1/2} = \pair{A}{A}^{1/2}$.\\
For any two $\mbb{F}$-Euclidean spaces $\mcl{X}$ and $\mcl{Y}$, their \emph{tensor product}\footnote{The tensor product of two $\mbb{F}$-Euclidean spaces provides the mathematical framework for treating two subsystems as a single unified system.} is denoted $\mcl{X} \otimes \mcl{Y}$. The \emph{partial trace}\footnote{In Quantum Information Theory, the partial trace allows us to analyze the state of a subsystem within a larger composite system.} operator on the space $\Li(\mcl{X} \otimes \mcl{Y})$, with respect to $\mcl{X}$ is denoted $\Tr_{\mcl{X}}$.
If $A$ is a linear operator, we denote its adjoint by $A^{*}$. In the real case the adjoint coincides with the transpose, denoted $A^T$.
We denote the set of Hermitian operators on $\mcl{X}$, i.e. those satisfying $A=A^*$, by $\Herm(\mcl{X})$. In the real case, $\Herm(\mcl{X})$ is commonly called the set of all symmetric operators on $\mcl{X}$.
For an $\mbb{F}$-Euclidean space $\mcl{X}$, and a linear operator $A \in \Li(\mcl{X})$, we denote the \emph{Hermitian part} of $A$ by $\Herm(A)\footnote{\text{In the complex case, some authors call it the real part of a linear operator $A$ and denote it $\Re(A)$}.}$. Therefore, $\Herm(A) = \tfrac{A + A^*}{2}$.
We say $A \in \L(\mcl{X})$ is positive semidefinite if there exists $B \in \L(\mcl{X})$ such that $A = B^*B$. The set of positive semidefinite operators on $\mcl{X}$ is denoted by $\Pos(\mcl{X})$.
We remind the reader that in the real case the condition $\pair{Ax}{x} \ge 0$ for all $x \in \mcl{X}$ does not by itself imply being positive semidefinite, as the rotation in $\mbb{R}^2$ by $90^\circ$ shows. One additionally needs the assumption $A=A^T$.
A positive semidefinite linear operator on $\mcl{X}$ with trace one is called a \emph{state}. The set of states on $\mcl{X}$ is denoted by $D(\mcl{X})$. The set of unitary operators on an $\mbb{F}$-Euclidean space $\mcl{X}$ is denoted $\U(\mcl{X})$. A state is called \emph{classical} (in a given basis) if and only if its  matrix representation is diagonal in that basis. The diagonal entries form a probability distribution. We represent finite sets by  $\Sigma$ and $\Gamma$.\\

In Quantum Information Theory, it is often convenient to work with $\mbb{F}$-Euclidean spaces whose orthonormal bases are indexed by a finite set $\Sigma$. This convention reflects the fact that such a basis can be used to store classical data — for example, a message to be processed, transmitted, encrypted, or the outcomes of a measurement. Given a finite set $\Sigma$, we denote the canonically associated $\mbb{F}$-Euclidean space by $\mathbb{F}^{\Sigma}$, equipped with the standard orthonormal basis $\{e_a : a \in \Sigma\}$. Basic linear algebra tells us that  if $|\Sigma| = d$, we can  identify $\mathbb{F}^{\Sigma}$ with $\mathbb{F}^{d}$. Therefore, for a $d$-dimensional $\mbb{F}$-Euclidean space $\mcl{X}$, every $v \in \mcl{X}$ is represented by a column vector in $\mbb{F}^{d}$, and corresponds to a linear map from $\mbb{F}$ to $\mcl{X}$. Thus, its adjoint, denoted by $v ^ {*}$, is a linear map going in the opposite direction. In this setting, we represent the standard basis of the space $\Li(\mbb{F} ^ \Sigma)$ by $\{\E_{a,b}\}_{a,b \in \Sigma}$, where $\E_{a,b} = e_a e_b ^ {*}$.\\

A linear map $\Phi$ of the form $\Phi : \Li(\mcl{X}) \to \Li(\mcl{Y})$ is called a \emph{super operator}.
Such a $\Phi$ is said to be \emph{positivity-preserving} if $\Phi(X) \in \Pos(\mcl{Y})$ whenever $X \in \Pos(\mcl{X})$, and we say it \emph{commutes with the adjoint} if $\Phi(X^*) = \Phi(X)^*$ for all $X \in \L(\mcl{X})$.
In the complex case, commutation with the adjoint is equivalent
to the condition that $\Phi$ maps Hermitian matrices to Hermitian matrices, and such super operators are often called \emph{Hermitian-preserving} {\cite{watrous2018theory}}.
As pointed out in \cite{MR4659878}, in the real case commutation with the adjoint
is a strictly stronger property than being Hermitian-preserving.
As in that work, we adjust the usual definitions so that things work properly in the real case.
We say $\Phi$ is  is \emph{positive} if it is positivity-preserving and commutes with the adjoint.
In the complex case the set of positive maps coincides with
the set of positivity-preserving maps, so our use of the term \emph{positive} agrees with the standard one for the complex setting.
We call $\Phi$ \emph{completely positive} if $\Phi \otimes \mathbbm{1}_{\Li(\mcl{Z})}$ is positive for every  $\mbb{F}$-Euclidean space $\mcl{Z}$.
By \cite{choi1975completely}, the map $\Phi$ is completely positive if and only if there exist an index set $\Sigma$ and operators $\{A_j\}_{j \in \Sigma} \subseteq \Li(\mcl{X}, \mcl{Y})$ such that for every $X \in \L(\mcl{X})$
\begin{equation} \label{eqn: Kraus rep}
    \Phi(X) = \sum_{j \in \Sigma} A_j X A_j^{*}.
\end{equation}
 Equivalently (see \cite[Thm. 2.22]{watrous2018theory}), the map $\Phi$ is completely positive if there is an operator $A \in \Li(\mcl{X}, \mcl{Y} \otimes \mcl{Z})$  for some $\mbb{F}$-Euclidean space $\mcl{Z}$ such that for every $X \in \Li(\mcl{X})$
\begin{equation} \label{eqn: Stinespring rep}
    \Phi(X) = \Tr_{\mcl{Z}}(A X A ^ {*}).
\end{equation}
The set $\mcl{A} = \{A_a\}_{a \in \Sigma}$ is called a set of Kraus operators for $\Phi$, and the operator $A$ is called a Stinespring operator of $\Phi$.
The expressions in \eqref{eqn: Kraus rep} and \eqref{eqn: Stinespring rep} are not unique, and are referred to as Kraus and Stinespring representations of $\Phi$. From \cite[Cor. 2.23]{watrous2018theory} it follows that two sets of Kraus operators for $\Phi$ are related by a unitary map. A similar result holds in \cite[Cor. 2.24]{watrous2018theory} for Stinespring operators of $\Phi$.
\\

A super operator $\Phi : \Li(\mcl{X}) \to \Li(\mcl{Y})$ is called \emph{trace preserving} if $\Tr(\Phi(X)) = \Tr(X)$ for all $X \in \Li(\mcl{X})$. Well known is the fact that a completely positive $\Phi$ with Kraus representation as in \eqref{eqn: Kraus rep} is trace preserving if and only if $\sum_{a \in \Sigma} A_a^* A_a = \mathbbm{1}_\mcl{X}$ \cite[Thm. 2.26]{watrous2018theory}. We remark that super operators that commute with the adjoint are characterized as differences of two completely positive ones. While not all the equivalences of \cite[Thm. 2.25]{watrous2018theory} hold in the real case, items 2,3, and 4 are still equivalent (see \cite[Lemma 3.1]{MR4659878}).
\\

Given $\mbb{F}$-Euclidean spaces $\mcl{X}$ and $\mcl{Y}$, we can make $\Li(\mcl{X}, \mcl{Y})$ a normed space by defining a norm on it. Although there are many ways to do this,  we only consider the class of norms called the \emph{Schatten $p$-norms}, where $p \in [1 , \infty]$. Given $A \in \Li(\mcl{X}, \mcl{Y})$ and $p \in [1, \infty)$, the Schatten $p$-norm of $A$, denoted $\n{A}_{p}$, is defined as \begin{equation} \label{eqn: Schatten-p-norm}
    \n{A}_{p} = \bigl(\Tr((A^*A)^{p/2})\bigr)^{1/p}.
\end{equation}
 In the case $p= \infty$, the Schatten $p$-norm of $A \in \Li(\mcl{X}, \mcl{Y})$ is its operator norm, that is, 
\[
\n{A}_{\infty} = \max \{\n{Ax} : x \in \mcl{X}, \n{x} \le 1\}.
\]

It is easy to see that if $p=2$, $\Li(\mcl{X},\mcl{Y})$ is an inner product space and the inner product generates the Frobenius norm. If $\rho \in \Pos(\mcl{X})$, and $p = 1$, it follows from  \eqref{eqn: Schatten-p-norm} that $\n{\rho}_1 = \Tr(\rho)$. Since the Schatten-1 norm of every positive semidefinite operator is its trace, $\n{A}_{1}$ is referred to as the \emph{trace norm} of $A$, for every $A \in \Li(\mcl{X}, \mcl{Y})$.

If the spaces $\Li(\mcl{X})$ and $\Li(\mcl{Y})$ are endowed with the Schatten $p$-norm and Schatten $q$-norm, respectively, and $\Phi : \Li(\mcl{X}) \to \Li(\mcl{Y})$ is a super operator, we denote the operator norm of $\Phi$ by $\n{\Phi}_{p \rightarrow q}$. Therefore,
\[
\n{\Phi}_{p \rightarrow q} = \max \{\n{\Phi(A)}_{q} : A \in \Li(\mcl{X}), \n{A}_{p} \le 1\}.
\]
For a finite set $\Sigma$ and an $\mbb{F}$-Euclidean space $\mcl{X}$, a function $\mu : \Sigma \to \Li(\mcl{X})$ is called a \emph{measurement}.
We stress that our terminology differs slightly from that commonly in use in Quantum Information Theory, e.g. in \cite{watrous2018theory}.
If the range of $\mu$ is contained in $\Herm(\mcl{X})$ (resp. $\Pos(\mcl{X})$), we call $\mu$ a \emph{Hermitian measurement} (resp. \emph{positive measurement}). We call it a \emph{quantum measurement} if it is positive and furthermore satisfies

\[
\sum_{x \in \Sigma} \mu(x) = \mathbbm{1}_{\mcl{X}}.
\]
A quantum measurement is often also called a \emph{positive operator valued measure (POVM)}, as seen in \cite{liu2023phase} and \cite{heinosaari2013quantum}. We are able to extract information from a quantum system using quantum measurements. Given a finite set $\Sigma$, and $\mbb{F}$-Euclidean spaces $\mcl{X}$ and $\mbb{F} ^ \Sigma$, \cite[Thm. 2.37]{watrous2018theory} states that there is a one to one correspondence between the set of quantum measurements of the form $\mu : \Sigma \to \Pos(\mcl{X})$ and quantum-to-classical channels of the form $\Psi_\mu : \Li(\mcl{X}) \to \Li(\mbb{F}^\Sigma)$ defined by 
 \begin{equation} \label{eqn 2.4}
    \Psi_\mu(X) = \sum_{a \in \Sigma} \pair{\mu(a)}{X}\E_{a,a}
\end{equation}

for all $X \in \Li(\mcl{X})$. In the general case where $\mu : \Sigma \to \Li(\mcl{X})$ is just a measurement, we take the super operator $\Psi_\mu$ defined in \eqref{eqn 2.4} to be the super operator corresponding to $\mu$.\\

Phase retrievability of quantum channels, introduced in \cite[Def. 1.1]{liu2023phase}, provides a framework for identifying rank‑one states using ideas from Frame Theory. A consequence of \cite[Prop. 1.4]{liu2023phase} is that if two pure states, $\rho, \sigma$ produce the same output over a phase retrievable quantum channel, then the states must coincide. Writing $\rho = uu^*$ and $ \sigma = vv^*$ for some unit vectors $u$ and $v$, the equality $uu^* = vv^*$ implies the existence of $\lambda \in \mbb{F}$ such that $|\lambda|=1$, and $u = \lambda v$. Thus, the underlying vectors defining the pure states are uniquely determined up to a phase.  However, extending these results to higher‑rank states proves  more challenging. This paper develops a generalization rooted in Quantum Information Theory, drawing in particular on tools from Quantum Tomography to extend the analysis beyond the rank‑one regime.  More generally, we give conditions for the phase-retrievability of super operators and show how these may be interpreted in terms of measurements. With this background in place, we introduce the following definitions.

\begin{definition} \label{def: Omega-phase retrievable measurement}
Let $\Sigma$ be a finite set, $\mcl{X}$ an $\mbb{F}$-Euclidean space, $\Omega \subseteq \Li(\mcl{X})$. A measurement $\mu : \Sigma \to \Li(\mcl{X})$ is called \emph{$\Omega$-phase retrievable} if for $\rho,\sigma \in \Omega$, the condition
\[
    \langle \mu(a), \rho \rangle = \langle \mu(a), \sigma \rangle \quad \text{for all } a \in \Sigma
\]
implies $\rho = \sigma$.
\end{definition}

If $\mcl{X}$ is $\mbb{C}$-Euclidean, $\Omega$ a subset of states and $\mu$ a quantum measurement, then Definition \ref{def: Omega-phase retrievable measurement} is the same as \cite[Sec. 2]{heinosaari2013quantum} where it is called an \emph{informationally complete measurement} with respect to $\Omega$, since for every $X \in \Li(\mcl{X})$ and every $a \in \Sigma$, $\pair{\mu(a)}{X} = \Tr(\mu(a)X)$.

\begin{definition} \label{def: Omega-PRs-operator}
Let $\mcl{X}$ be a $\mbb{F}$-Euclidean space, $\Omega \subseteq \Li(\mcl{X})$ a subset, and $\Phi : \L(\mcl{X}) \rightarrow \L(\mcl{Y})$ a super operator. The super operator $\Phi$ is \emph{$\Omega$-phase retrievable} if there exists a finite set $\Sigma$, and a measurement $\mu : \Sigma \to \Li(\mcl{Y})$ such that the composed measurement
\[
    \Phi^{*} \circ \mu : \Sigma \to \Li(\mcl{X})
\]
is $\Omega$-phase retrievable.
\end{definition}

The adjoint in Definition \ref{def: Omega-PRs-operator} is taken with respect to the Frobenius inner product, and we say that $\Phi$ is $\Omega$-phase retrievable with respect to the measurement $\mu$. If $\mu$ is a quantum measurement and $\Phi ^ *$ is positive and unital, then $\Phi ^ * \circ \mu$ is a quantum measurement. In particular, it is so when $\Phi$ is a quantum channel.
Note that the case when $\Phi$ is a channel, $\mu$ is a quantum measurement, and $\Omega$ is the set of pure states is the notion of a phase retrievable channel of \cite[Def. 1.1]{liu2023phase}.

\begin{definition}  \label{def: Omega-injective-s-operator}
Let $\Sigma$ be a finite set, let $\mcl{X}$ be an $\mbb{F}$-Euclidean space, let $\Omega \subseteq \Li(\mcl{X})$, and let $\Phi : \L(\mcl{X}) \to \L(\mcl{Y})$ be a super operator. The super operator $\Phi : \L(\mcl{X}) \to \L(\mcl{Y})$ is \emph{$\Omega$-injective} (or injective on $\Omega$) if for all $\rho, \sigma \in \Omega$, we have $\rho = \sigma$ whenever $\Phi(\rho) = \Phi(\sigma)$.
\end{definition}

We introduce notations for some special subsets used in this paper: \(\Omega_k = \{X \in \Li(\mcl{X}): \rank(X) \le k\}\), $\mcl{P}_k = \Omega_k \cap \Pos({\mcl{X}})$, and $\D_k = \Omega_k \cap \D(\mcl{X})$. Obviously, $\D_k \subset \mcl{P}_k \subset \Omega_k$. We represent a generic subset of $\Pos(\mcl{X})$ by $\mcl{P}$.\\

Before ending this section, we state some useful results about measurements and super operators.

\begin{proposition} \label{prop: Herm meas and so}
Let $\mcl{X}$ be an $\mbb{F}$-Euclidean space and $\Sigma$ a finite set. A measurement $\mu : \Sigma \to \Li(\mcl{X})$ is Hermitian if and only if the corresponding super operator $\Psi_\mu : \Li(\mcl{X}) \to \Li(\mbb{F}^\Sigma)$ defined by
\[
\Psi_\mu(X) = \sum_{a \in \Sigma} \pair{\mu(a)}{X}\E_{a,a}
\]
for all $X \in \Li(\mcl{X})$ commutes with the adjoint.
\end{proposition}
\begin{proof}
    First assume that $\mu$ is Hermitian. Then for every $X \in \Li(\mcl{X})$ we have
    \[
    \begin{aligned}
    \Psi_\mu(X^*) & = \sum_{a \in \Sigma} \pair{\mu(a)}{X^*}\E_{a,a} =\sum_{a \in \Sigma} \overline{\Tr(\mu(a)X)}\E_{a,a} = \sum_{a \in \Sigma}\overline{\pair{\mu(a)}{X}}\E_{a,a} = \Psi_\mu(X)^*.
    \end{aligned}
    \]
    The second equality holds because the trace operator is cyclic, taking adjoints conjugates the trace, and $\mu(a)$ is Hermitian. The third equality holds because $\mu(a)$ is Hermitian.\\
    Now, we assume that $\Psi_\mu$ commutes with the adjoint. For all $X \in \Li(\mcl{X})$ and $a \in \Sigma$ we have
    \[
    \pair{\mu(a)}{X^*} = \overline{\pair{\mu(a)}{X}}.
    \]
    From this it follows that $\pair{\mu(a)^* - \mu(a)}{X} = 0$ for all $X \in \Li(\mcl{X})$ and $a \in \Sigma$. Therefore $\mu(a)^* = \mu(a)$ for all $a \in \Sigma$, so that $\mu$ is Hermitian.
\end{proof}

\begin{proposition} \label{prop: positive meas and so}
Let $\mcl{X}$ be an $\mbb{F}$-Euclidean space and $\Sigma$ a finite set. A measurement $\mu : \Sigma \to \Li(\mcl{X})$ is positive if and only if the corresponding super operator $\Psi_\mu : \Li(\mcl{X}) \to \Li(\mbb{F}^\Sigma)$ defined by
\[
\Psi_\mu(X) = \sum_{a \in \Sigma} \pair{\mu(a)}{X}\E_{a,a}
\]
for all $X \in \Li(\mcl{X})$ is positive.
\end{proposition}
\begin{proof}
    Assume $\mu$ is positive. Let $\rho \in \Pos(\mcl{X})$, then for every $a \in \Sigma, \pair{\mu(a)}{\rho} \ge 0$. Therefore, $\Psi_\mu(\rho) \in \Li(\mbb{F}^\Sigma)$ is a Hermitian, diagonal operator with non-negative diagonal entries.  Hence, $\Psi_\mu(\rho) \in \Pos(\mbb{F}^\Sigma)$.\\
    Conversely, assume that $\Psi_\mu$ is positive, then by definition, it commutes with the adjoint. From Proposition \ref{prop: Herm meas and so}, $\mu$ is a Hermitian measurement so that for every $a \in \Sigma$ we have $\mu(a)^* = \mu(a)$. Let $\rho \in \Pos(\mcl{X})$, then $\Psi_\mu(\rho) \in \Pos(\mbb{F}^\Sigma)$, hence, $\pair{\mu(a)}{\rho} \ge 0$ for all $a \in \Sigma$. Therefore $\mu(a) \in \Pos(\mcl{X})$ for all $a \in \Sigma$. Since we additionally have that for every $a \in \Sigma, \mu(a)$ is Hermitian, it follows that $\mu$ is a positive measurement.
\end{proof}
Since it is easy to check from the definition that the adjoint of a super operator that commutes with the adjoint also commutes with the adjoint, and the adjoint of a positive super operator is again positive \cite[Lemma 3.1]{MR4659878}, we may state the following propositions without providing their proofs.
\begin{proposition} \label{prop: composed meas Herm} Let $\Sigma$ be a finite set and let $\mcl{X}$ and $\mcl{Y}$ be $\mbb{F}$-Euclidean spaces. For a Hermitian measurement $\mu: \Sigma \to \Li(\mcl{Y})$ and a super operator $\Phi : \Li(\mcl{X}) \to \Li(\mcl{Y})$ that commutes with the adjoint, the composed measurement $\Phi^* \circ \mu : \Sigma \to \Li(\mcl{X})$ is Hermitian.
\end{proposition}

\begin{proposition} \label{prop: composed meas Pos}
 Let $\Sigma$ be a finite set and let $\mcl{X}$ and $\mcl{Y}$ be $\mbb{F}$-Euclidean spaces. For a positive measurement $\mu: \Sigma \to \Li(\mcl{Y})$ and a positive super operator $\Phi : \Li(\mcl{X}) \to \Li(\mcl{Y})$, the composed measurement $\Phi^* \circ \mu : \Sigma \to \Li(\mcl{X})$ is positive.
\end{proposition}

\section{Characterizations of $\Omega$-phase retrievable super operators and $\Omega$-phase retrievable quantum Measurements}\label{Sec: Xter PRMSO & PRQM}

Using the definitions from the previous section, we now state some basic results. The first is related to and inspired by \cite[Prop. 2.2]{heinosaari2013quantum}.

\begin{lemma} \label{lem: injectivity}
    Let $\mcl{X}$ and $\mcl{Y}$ be $\mbb{F}$-Euclidean spaces, $\Phi : \Li(\mathcal{X}) \to \Li(\mathcal{Y})$ a super operator, and $\Omega \subseteq \Li(\mathcal{X})$ a subset. The following are equivalent:
\begin{enumerate}[(a)]
    \item $\Phi$ is injective on $\Omega$.
    \item $(\Omega - \Omega) \cap \range(\Phi^*)^\perp = \{0\}$.
\end{enumerate}
\end{lemma}
\begin{proof} A standard result in linear algebra states that $\ker(\Phi) = \range(\Phi ^ {*}) ^ \perp$. Hence, we find that
$\Phi$ is injective on $\Omega$ if and only if
$(\Omega - \Omega) \cap \ker(\Phi) = \{0\} $
if and only if
$(\Omega - \Omega)\, \cap\, \range(\Phi^*)^\perp = \{0\} $.
\end{proof}
The next result generalizes and strengthens \cite[Prop. 1.4]{liu2023phase}. The generalization is provided by parts \eqref{thm: injectivity and phase retrievability of a quantum channel a} and \eqref{thm: injectivity and phase retrievability of a quantum channel b}, which extend the conclusions to a broader class of sets and super operators. The improvement appear in parts \eqref{thm: injectivity and phase retrievability of a quantum channel c} and \eqref{thm: injectivity and phase retrievability of a quantum channel d}, which establishes an additional equivalence that complements and sharpens the statements in parts \eqref{thm: injectivity and phase retrievability of a quantum channel a} and \eqref{thm: injectivity and phase retrievability of a quantum channel b}. The statement used in part (c) was motivated by the statement of \cite[Prop. 2]{heinosaari2013quantum}. This reinforces the usefulness of tools from Quantum Information Theory in extending some of the results of \cite{liu2023phase} to more general settings.

\begin{theorem} \label{thm: injectivity and phase retrievability of a quantum channel}
     Let $\mcl{X}$ and $\mcl{Y}$ be $\mbb{F}$-Euclidean spaces, $\Phi : \Li(\mathcal{X}) \to \Li(\mathcal{Y})$ a super operator, and $\Omega \subseteq \Li(\mathcal{X})$ a subset. The following are equivalent:
    \begin{enumerate}[(a)]
        \item \label{thm: injectivity and phase retrievability of a quantum channel a} $\Phi$ is injective on $\Omega$.

        \item \label{thm: injectivity and phase retrievability of a quantum channel b} There exists a finite set $\Sigma$ and a measurement $\mu : \Sigma \to \Li(\mcl{Y})$ such that $\Phi ^ {*} \circ \mu$ is $\Omega$-phase retrievable. (That is, $\Phi$ is $\Omega$-phase retrievable). Furthermore, if $\mbb{F}=\mbb{C}$ then $\mu$ can always be taken to be a quantum measurement.

        \item \label{thm: injectivity and phase retrievability of a quantum channel c} If $\nu : \Sigma \to \Li(\mcl{Y})$ is a measurement that satisfies $\Phi ^ {*} (\spa\{\nu(a)\}_{a \in \Sigma}) = \range(\Phi ^ {*})$, then $\Phi ^ {*} \circ \nu$ is $\Omega$-phase retrievable.

        \item \label{thm: injectivity and phase retrievability of a quantum channel d} If $\nu: \Sigma\to\Li(\mcl{Y})$ is a measurement with $\spa\{\nu(a)\}_{a \in \Sigma} = \Li(\mcl{Y})$, then $\Phi^*\circ \nu$ is $\Omega$-phase retrievable.
        
    \end{enumerate}
\end{theorem}
\begin{proof}
    
    \eqref{thm: injectivity and phase retrievability of a quantum channel b}$\implies$\eqref{thm: injectivity and phase retrievability of a quantum channel a}: Assume (b) to be true and let $\rho$ and $\sigma$ in $\Omega$ be such that $\Phi(\rho) = \Phi(\sigma)$.  
    Let $\mu : \Sigma \to \Li(\mcl{Y})$ be a measurement with $\Phi^{*} \circ \mu : \Sigma \to \Li(\mcl{X})$ being $\Omega$-phase retrievable.  
    Then for any $a \in \Sigma$,
    \[
        \langle (\Phi^{*} \circ \mu)(a), \rho \rangle
        = \langle \mu(a), \Phi(\rho) \rangle
        = \langle \mu(a), \Phi(\sigma) \rangle
        = \langle (\Phi^{*} \circ \mu)(a), \sigma \rangle.
    \]
    Hence, $\rho = \sigma$ and this proves (a).
    \\
    
    \eqref{thm: injectivity and phase retrievability of a quantum channel a}$\implies$\eqref{thm: injectivity and phase retrievability of a quantum channel c}: Suppose (a) holds and let $\nu: \Sigma \to \Li(\mcl{Y})$ be a measurement satisfying $\Phi ^ {*} (\spa\{\nu(a)\}_{a \in \Sigma}) \allowbreak = \range(\Phi ^ {*})$. Let $\rho$, $\sigma \in \Omega$ satisfy: for every $a \in \Sigma$, $\pair{(\Phi^{*} \circ \nu)(a)}{\rho} = \pair{(\Phi^{*} \circ \nu)(a)}{\sigma}$. Then, $\pair{\Phi^{*}(X)}{\rho - \sigma} = 0$ for every $X \in \spa\{\nu(a)\}_{a \in \Sigma}$. Hence, $\rho - \sigma \in (\Omega - \Omega) \cap \range(\Phi^{*}) ^ \perp$. By (a), we conclude from Lemma \ref{lem: injectivity} that $\rho = \sigma$. 
    \\
    
    \eqref{thm: injectivity and phase retrievability of a quantum channel c}$\implies$\eqref{thm: injectivity and phase retrievability of a quantum channel d}: Assuming that the statement in \eqref{thm: injectivity and phase retrievability of a quantum channel c} is true, then for every measurement $\nu: \Sigma \to \Li(\mcl{Y})$ with $\spa\{\nu(a)\}_{a \in \Sigma} = \Li(\mcl{Y})$ it follows that $\Phi^*(\spa\{\nu(a)\}_{a \in \Sigma}) = \range(\Phi^*)$. Therefore, $\Phi^*\circ\nu$ is $\Omega$-phase retrievable.
    \\

    \eqref{thm: injectivity and phase retrievability of a quantum channel d}$\implies$\eqref{thm: injectivity and phase retrievability of a quantum channel b}: If we assume that \eqref{thm: injectivity and phase retrievability of a quantum channel d} is true, then \eqref{thm: injectivity and phase retrievability of a quantum channel b} follows immediately simply by taking any measurement $\nu : \Sigma \to \Li(\mcl{Y})$ such that  $\spa\{\nu(a)\}_{a \in \Sigma} = \Li(\mcl{Y})$, which always exists simply by considering a spanning set. In the case $\mbb{F}=\mbb{C}$, this can be done with $\mu$ being a quantum measurement by \cite[Examples 2.7 and 2.45]{watrous2018theory} (where it is called an informationally-complete measurement).
\end{proof}

\begin{remark} \label{rem: Special case of P-phase retrievable quantum channels}
Since by the previous theorem $\Omega$-phase retrievability is equivalent to injectivity on $\Omega$,
in the special case where \(\Omega = \mcl{P}_k\) note that $\mcl{P}_k$-phase retrievability extends the classical notion of phase retrieval for vectors. Recall the latter means that two vectors 
$x, y \in \mathcal{X}$ satisfy $xx^* = yy^*$ precisely when $x = \lambda y$ for some scalar $\lambda \in \mathbb{F}$ with $|\lambda| = 1$. 
On the other hand, for any $\mbb{F}$-Euclidean space $\mcl{X}$, every 
positive semidefinite operator $P \in \Li(\mathcal{X})$ of rank at most $k$ admits a factorization of the form $P = AA^*$ for some $A \in \Li(\mathbb{F}^k, \mathcal{X})$, and any two operators $A, B \in \Li(\mathbb{F}^k, \mathcal{X})$ satisfy $AA^* = BB^*$ if and only if $A = BU$ for some unitary operator $U \in \Li(\mathbb{F}^k)$ \cite[Thm. 2.2.10]{johnston2021advanced}.
Thus, $\mcl{P}_k$-phase retrievability captures a high-dimensional analogue of the usual phase ambiguity.
\\

The following theorem provides a foundation for characterizing  $\Omega$-phase retrievable measurements in terms of the $\Omega$-phase retrievability of their associated super operators. Specifically, it establishes a connection between the phase-retrieval properties of measurements and those of the corresponding super operators, enabling one to pass naturally between these two perspectives. This correspondence will play a crucial role in the study of stability. In particular, it will be used to establish a relationship between the stability properties of $\mcl{P}_k$-phase retrievable super operators and those of $\mcl{P}_k$-phase retrievable measurements (see Sections \ref{Sec: Stability of PKPRMSO} and \ref{Sec: Stability of PKPRM}.)

\end{remark}

\begin{theorem} \label{thm: equivalence of P-phase retrievable measurements and channels}
    Let $\Sigma$ be a finite set, and $\mcl{X}$ an $\mbb{F}$-Euclidean space. For a fixed set $\Omega \subseteq \Li(\mcl{X})$, a measurement $\mu : \Sigma \to \Li(\mcl{X})$ is $\Omega$-phase retrievable if and only if the super operator $\Psi_\mu : \Li(\mcl{X}) \to \Li(\mbb{F}^\Sigma)$   defined by
    
    \begin{equation}  \label{eqn: measurement SO}
    \Psi_\mu (X) = \sum_{a\in \Sigma} \pair{\mu(a)}{X}\E_{a,a}
    \end{equation}
    for all $X \in \Li(\mcl{X})$ is $\Omega$-phase retrievable.
\end{theorem}

\begin{proof}
    Assume $\mu$ is $\Omega$-phase retrievable. We show that $\Psi_{\mu}$ is injective on $\Omega$. For this, assume $\rho$, $\sigma \in \Omega$ with $\Psi_\mu(\rho) = \Psi_\mu(\sigma)$. Then
    \[
    \sum_{a \in \Sigma} \pair{\mu(a)}{\rho}\E_{a,a} = \sum_{a \in \Sigma}\pair{\mu(a)}{\sigma}\E_{a,a}.
    \]
    Since on both sides of the equality we are diagonal operators, it follows that for every $a \in \Sigma$, $\pair{\mu(a)}{\rho} = \pair{\mu(a)}{\sigma}$. Hence $\rho = \sigma$ since $\mu$ is $\Omega$-phase retrievable. This shows that $\Psi_\mu$ is injective on $\Omega$ so that it is $\Omega$-phase retrievable by Theorem \ref{thm: injectivity and phase retrievability of a quantum channel}.

    On the other hand, suppose $\Psi_\mu$ is $\Omega$-phase retrievable. From Theorem \ref{thm: injectivity and phase retrievability of a quantum channel}, it follows that $\Psi_\mu$ is injective on $\Omega$. Let $\rho$, $\sigma \in \Omega$ be such that for every $a \in \Sigma$, $\pair{\mu(a)}{\rho} = \pair{\mu(a)}{\sigma}$. Then we have
    \[
    \Psi_\mu (\rho) = \sum_{a \in \Sigma} \pair{\mu(a)}{\rho} \E_{a,a} = \sum_{a \in \Sigma} \pair{\mu(a)}{\sigma} \E_{a,a} = \Psi_\mu (\sigma).
    \]
    Hence, $\rho = \sigma$.
\end{proof}

As an immediate consequence of the preceding theorem, we obtain the following result for quantum measurements and their associated quantum-to-classical channels.

\begin{corollary} \label{cor: equivalence of k-phase retrievable measurements and channels}
    Let $\Sigma$ be a finite set, $\mcl{X}$ an $\mbb{F}$-Euclidean space, and $\mcl{P} \subseteq \Pos(\mcl{X})$. A quantum measurement $\mu : \Sigma \to \Pos(\mcl{X})$ is $\mcl{P}$-phase retrievable if and only if the quantum-to-classical channel $\Psi_\mu : \Li(\mcl{X}) \to \Li(\mbb{C} ^ \Sigma)$ defined as
    \[
     \Psi_\mu (X) = \sum_{a \in \Sigma} \pair{\mu(a)}{X} \E_{a,a}
    \]
    for all $X \in \Li(\mcl{X})$, is $\mcl{P}$-phase retrievable.
\end{corollary}

Observe that the adjoint of the super operator $\Psi_\mu$ defined in \eqref{eqn: measurement SO} satisfies $\range(\Psi_\mu^*) = \spa\{\mu(a)\}_{a\in \Sigma}$.  Combining Lemma \ref{lem: injectivity} with Theorem \ref{thm: equivalence of P-phase retrievable measurements and channels}, we obtain the following result, noting that $\spa\{\mu(a)\}_{a \in \Sigma}^\perp = \{\mu(a)\}_{a \in \Sigma}^\perp$.

\begin{corollary} \label{cor: Omega PR measurement}
    \label{thm: phase retrievable measurment and span}
    Suppose $\Sigma$ is a finite set and $\mcl{X}$ an $\mbb{F}$-Euclidean space. Let $\Omega \subseteq \Li(\mcl{X})$ be a subset and $\mu : \Sigma \longrightarrow \Li(\mcl{X})$ a measurement. The following are equivalent:
    \begin{enumerate}[(a)]
        \item $\mu$ is $\Omega$-phase retrievable.
        \item $(\Omega - \Omega) \cap \{\mu(a)\}_{a \in \Sigma} ^ \perp = \{0\}$. 
    \end{enumerate}

\end{corollary}
We state characterizations for $\mcl{P}$-phase retrievability of quantum measurements. The corollary that follows gives a characterization of $\mcl{P}$-phase retrievable, positive, trace preserving super operators.

\theorem \label{thm: k-phase retrivable Q-measurements}

Let $\mcl{X}$ be an $\mbb{F}$-Euclidean space, $\Sigma$ a finite set, $\mu : \Sigma \to \Pos(\mcl{X})$ a quantum measurement, and $\mcl{P} \subseteq \Pos(\mcl{X})$ be a subset closed under multiplication by positive scalars. The following are equivalent:
\begin{enumerate}[(a)]
\item\label{k-PRQM-a} The measurement $\mu$ is $\mcl{P}$-phase retrievable.

\item\label{k-PRQM-b} For any $\rho, \sigma \in \mcl{P}$, if the vectors $\big( \pair{\mu(a)}{\rho} \big)_{a\in\Sigma}$ and $\big( \pair{\mu(a)}{\sigma} \big)_{a\in\Sigma}$ are linearly dependent, then $\rho$ and $\sigma$ are also linearly dependent.

\item \label{k-PRQM-c} For any $\rho, \sigma \in \mcl{P}$,  
\[
\sum_{a\in \Sigma} \sqrt{\pair{\mu(a)}{\rho}} \sqrt{\pair{\mu(a)}{\sigma}}  = \sqrt{\n{\rho}_1\n{\sigma}_1}
\]
implies $\rho$ and $\sigma$ are linearly dependent.
\end{enumerate}

\begin{proof}
(\ref{k-PRQM-a})$\implies$(\ref{k-PRQM-b}): Assuming (\ref{k-PRQM-a}), suppose that $\rho,\sigma \in \mcl{P}$ and $\big( \pair{\mu(a)}{\rho} \big)_{a\in\Sigma}$ and $\big( \pair{\mu(a)}{\sigma} \big)_{a\in\Sigma}$ are linearly dependent. Then there exists $\lambda\in \mbb{F}$ such that for every $a\in\Sigma$, $\pair{\mu(a)}{\rho}=\lambda\pair{\mu(a)}{\sigma}$. We know $\pair{\mu(a)}{\rho},\pair{\mu(a)}{\sigma} \geq 0$. If for every $a \in \Sigma$, $\pair{\mu(a)}{\sigma} = 0$, then taking a sum over $a \in \Sigma$, we have
\[
0=\sum_{a\in\Sigma}\pair{\mu(a)}{\sigma}= \Big\langle \sum_{a\in\Sigma}\mu(a), \sigma \Big \rangle=\pair{\mathbbm{1}_{\mcl{X}}}{\sigma}=\Tr(\sigma)=\n{\sigma}_1.
\] 
Therefore, $\sigma=0$, which implies that $\rho$ and $\sigma$ are linearly dependent. Otherwise, there exists $a_0 \in \Sigma$ such that $\pair{\mu(a_0)}{\sigma}>0$. Therefore, $\lambda = \pair{\mu(a_0)}{\rho}/\pair{\mu(a_0)}{\sigma} \ge 0$, and we have $\pair{\mu(a)}{\rho} = \pair{\mu(a)}{\vartheta}$ for all $a \in \Sigma$, where $\vartheta = \lambda \sigma \in \mcl{P}$. If $\lambda = 0$, then for every $a \in \Sigma$ we have $\pair{\mu(a)}{\rho} = 0$, and summing over $a \in \Sigma$ implies $\Tr(\rho)=0$. Therefore, $\rho = 0$, hence $\rho$ and $\sigma$ are linearly dependent. If $\lambda > 0$, then $\vartheta \in \mcl{P}$. By (\ref{k-PRQM-a}), we have $\rho= \vartheta =\lambda \sigma$, which implies that $\rho$ and $\sigma$ are linearly dependent.
\\

(\ref{k-PRQM-b}) $\implies$ (\ref{k-PRQM-c}): Suppose (\ref{k-PRQM-b}) holds. Let $\rho, \sigma\in \mcl{P}$ be such that 
    \[
\sum_{a\in \Sigma} \sqrt{\pair{\mu(a)}{\rho}} \sqrt{\pair{\mu(a)}{\sigma}}  = \sqrt{\n{\rho}_1\n{\sigma}_1}
\] 
By the Cauchy-Schwarz inequality, 
    \[
    \sum_{a \in \Sigma}\sqrt{\pair{\mu(a)}{\rho}}\sqrt{\pair{\mu(a)}{\sigma}} \leq \Bigl( \sum_{a \in \Sigma}\pair{\mu(a)}{\rho} \Bigr)^\frac{1}{2}\Bigl( \sum_{a \in \Sigma}\pair{\mu(a)}{\sigma} \Bigr)^\frac{1}{2}=\sqrt{\n{\rho}_1\n{\sigma}_1}
    \]
    where the above equality  follows because 
    \[
    \sum_{a \in \Sigma}\pair{\mu(a)}{\rho}=\Big \langle \sum_{a \in \Sigma}\mu(a),\rho \Big \rangle = \pair{\mathbbm{1}_{\mcl{X}}}{\rho} = \Tr(\rho) = \n{\rho}_1.
    \]
    Similarly, 
    \[
    \sum_{a \in \Sigma}\pair{\mu(a)}{\sigma} =  \n{\sigma}_1.
    \]
    But we have the equality condition in the assumption, and it follows that:
    \[
    \sum_{a \in \Sigma}\sqrt{\pair{\mu(a)}{\rho}}\sqrt{\pair{\mu(a)}{\sigma}} = \Bigl( \sum_{a \in \Sigma}\pair{\mu(a)}{\rho} \Bigr)^\frac{1}{2}\Bigl( \sum_{a \in \Sigma}\pair{\mu(a)}{\sigma} \Bigr)^\frac{1}{2}
    \]
    which by the Cauchy-Schwarz inequality, equality holds if and only if $\sqrt{\pair{\mu(a)}{\rho}}$ and $\sqrt{\pair{\mu(a)}{\sigma}}$ are linearly dependent. Let $\alpha \in \mbb{F}$ be such that $\sqrt{\pair{\mu(a)}{\rho}} = \alpha \sqrt{\pair{\mu(a)}{\sigma}}$ for all $a\in \Sigma$. Then $\pair{\mu(a)}{\rho} = \alpha^2 \pair{\mu(a)}{\sigma}$ for every $a \in \Sigma$, so that $(\pair{\mu(a)}{\rho})_{a\in \Sigma}$ and $(\pair{\mu(a)}{\sigma})_{a\in \Sigma}$ are linearly dependent for all $a \in \Sigma$. Then it follows from (\ref{k-PRQM-b}) that $\rho$ and $\sigma$ are linearly dependent, as required.
    \\

    (\ref{k-PRQM-c})$\implies$(\ref{k-PRQM-a}): Assume that $\rho, \sigma \in  \mcl{P}$ are such that for all $a \in \Sigma$, $\pair{\mu(a)}{\rho} = \pair{\mu(a)}{\sigma}$. Then $\n{\rho}_1 = \n{\sigma}_1$. Therefore, $\rho = 0$ if and only if $\sigma = 0$ and in this case, $\rho = \sigma$. We consider the case where $\rho, \sigma \in \mcl{P} \setminus \{0\}$. We have the equalities
    \[
    \sum_{a \in \Sigma} \sqrt{\pair{\mu(a)}{\rho}} \sqrt{\pair{\mu(a)}{\sigma}} = \sum_{a \in \Sigma} \pair{\mu(a)}{\rho} = \Tr(\rho) = \n{\rho}_1= \sqrt{\n{\rho}_1 \n{\sigma}_1}.
    \]
    By (\ref{k-PRQM-c}), we conclude that $\rho = \lambda \sigma$ for some $\lambda \in \mbb{F}$. Then we have $\Tr(\rho) = \lambda \Tr(\sigma)$. Hence, $\n{\rho}_1  = \lambda \n{\sigma}_1$. Therefore, $\lambda = 1$ and this implies $\rho = \sigma$, as required. This completes the proof.
\end{proof}

If $\Phi : \L(\mcl{X}) \to \L(\mcl{Y})$ is positive and trace preserving, and $\mu : \Sigma \to \L(\mcl{Y})$ is a quantum measurement, then $\Phi ^ {*} \circ \mu$ is also a quantum measurement because $\Phi^*$ is positive and unital. 
Thus, Theorem \ref{thm: k-phase retrivable Q-measurements} immediately implies the following.

\begin{corollary}\label{cor:k-phase retrievable ptp operators} Let $k \in \mbb{N}$, $\mcl{X}$ and $\mcl{Y}$ $\mbb{F}$-Euclidean spaces, and $\Sigma$ a finite set. Let $\Phi : \Li(\mcl{X}) \to \Li(\mcl{Y})$ be a positive, trace preserving super operator, and $\mu:\Sigma \to \Pos(\mcl{Y})$ a quantum measurement.
Suppose $\mcl{P} \subseteq \Pos(\mcl{X})$ is closed under multiplication by positive scalars.
The following are equivalent:

\begin{enumerate}[(a)]
\item \label{k-PRptp-a} For every $\rho, \sigma \in \mcl{P}$, having $\pair{\mu(a)}{\Phi(\rho)} = \pair{\mu(a)}{\Phi(\sigma)}$ for all $a \in \Sigma$ implies $\rho = \sigma$.

\item \label{k-PRptp-b} For every $\rho, \sigma \in \mcl{P}$, if the vectors $\big( \pair{\mu(a)}{\Phi(\rho)} \big)_{a\in\Sigma}$ and $\big( \pair{\mu(a)}{\Phi(\sigma)} \big)_{a\in\Sigma}$ are linearly dependent, then $\rho$ and $\sigma$ are also linearly dependent.

\item \label{k- PRptp-c} For every $\rho, \sigma \in \mcl{P}$,  
\[
\sum_{a\in \Sigma} \sqrt{\pair{\mu(a)}{\Phi(\rho)}} \sqrt{\pair{\mu(a)}{\Phi(\sigma)}}  = \sqrt{\n{\rho}_1\n{\sigma}_1}
\]
implies that $\rho$ and $\sigma$ are linearly dependent.

\end{enumerate}
\end{corollary}

\remark When $\mcl{P} = \mcl{P}_1$ and $\rho, \sigma \in \mcl{P}$, the Spectral theorem provides that the operators $\rho$ and $\sigma$ that appear in Theorem \ref{thm: k-phase retrivable Q-measurements} and Corollary \ref{cor:k-phase retrievable ptp operators} can be expressed as $\rho = u u ^ *$ and $\sigma = v v ^ *$ for some $u, v \in \mcl{X}$. In this setting, the trace norms of $\rho$ and $\sigma$ that appear in parts (\ref{k-PRQM-c}) of Theorem \ref{thm: k-phase retrivable Q-measurements} and Corollary \ref{cor:k-phase retrievable ptp operators} reduce to the square of the Euclidean norms of $u$ and $v$, respectively.\\

The next characterization of quantum measurements in this section is a consequence of Theorem \ref{thm: k-phase retrivable Q-measurements}. It is related to a quantity called the \emph{Bhattacharyya coefficient}, which is a statistical measure of the overlap or similarity between two probability distributions and can be thought of as the classical version of \emph{quantum fidelity}\footnote{The fidelity of two quantum states is also a measure of their similarity.}. Following \cite{watrous2018theory}, given a finite set $\Sigma$, and a pair of vectors $u, v \in [0, \infty) ^ \Sigma$ with non-negative entries $u(a)$ and $v(a)$, the Bhattacharyya coefficient $\B(u,v)$ is defined as
\[
\B(u,v) = \sum_{a \in \Sigma}\sqrt{u(a)} \sqrt{v(a)}.
\]
Motivated by the role of the Bhattacharyya coefficient in comparing probability distributions and its established relationship to quantum fidelity \cite[Thm.3.24]{watrous2018theory}, we formulate a dedicated theorem capturing its connection to phase retrievability. We refer to this result as the Bhattacharyya coefficient characterization of a $\mcl{P}$-phase retrievable quantum measurement.

\begin{corollary}
    \label{cor:Bhattacharyya xterization}
Let $\Sigma$ be a finite set, $\mcl{X}$ an $\mbb{F}$-Euclidean space, and $\mcl{P} \subseteq \D(\mcl{X})$. A quantum measurement $\mu : \Sigma \to \Pos(\mcl{X})$ is $\mcl{P}$-phase retrievable if and only if for any $\rho$, $\sigma \in \mcl{P}$, having
\[
\sum_{a \in \Sigma} \sqrt{\pair{\mu(a)}{\rho}} \sqrt{\pair{\mu(a)}{\sigma}} = 1
\]
implies $\rho = \sigma$.
\end{corollary}

\section{Characterization of $\Omega$-phase Retrievability of Quantum Channels via Complementary Quantum Channels}

Given a quantum channel $\Phi: \Li(\mcl{X}) \to \Li(\mcl{Y})$, where $\mcl{X}$ and $\mcl{Y}$ are $\mbb{F}$-Euclidean spaces, there is a family $\mcl{F}_\Phi$ of quantum channels associated with $\Phi$, called \emph{complementary channels of $\Phi$}. Each member of $\mcl{F}_\Phi$ is a quantum channel that models the information lost to the environment when quantum information is transmitted from a source to a receiver through the quantum channel $\Phi$. The notion of a complementary channel of a quantum channel was first introduced in \cite{MR2160796}. For more about complementary channels, we refer the reader to \cite[Ch. 8]{watrous2018theory}, \cite[Sec. IV]{kretschmann2008information}, and \cite{holevo2007complementary}. Every set of Kraus operators $\mcl{A}=\{A_a\}_{a \in \Sigma}$ for $\Phi$ identifies an element $\Phi^c_\mcl{A} \in \mcl{F}_\Phi$, which is of the form $\Phi^c_\mcl{A}:\Li(\mcl{X}) \to \Li(\mbb{F}^\Sigma)$, defined by

\begin{equation} \label{eqn: complementary channel}
    \Phi_{\mcl{A}}^c(X) = \sum_{j\in \Gamma} R_j X R_j^*
\end{equation}

for all $X \in \Li(\mcl{X})$, where for the standard orthonormal basis $\{e_a\}_{a \in \Sigma}$ in $\mbb{F}^\Sigma$, and $\{f_j\}_{j\in \Gamma}$ a fixed orthonormal basis of $\mcl{Y}$, we define $R_j = \sum_{a \in \Sigma}e_a f_j^* A_a$. The quantum channel $\Phi^c_{\mcl{A}}$ is said to be the canonical complementary quantum channel associated with $\mcl{A}$  and the space $\mbb{F}^\Sigma$ is called the ancilla space or environment corresponding to $\mcl{A}$.
Furthermore, any two complementary channels of $\Phi$ are equivalent modulo conjugation by an isometry. That is, if $\Phi ^c:\Li(\mcl{X}) \to \Li(\mcl{Z}_1)$ and $\Phi^{c'}:\Li(\mcl{X}) \to \Li(\mcl{Z}_2)$ are complementary to $\Phi$ and without loss of generality we assume that $\dim(\mcl{Z}_1) \le \dim(\mcl{Z}_2)$ then there exists an isometry $W  : \mcl{Z}_1 \to \mcl{Z}_2$ such that $\Phi^{c'}(X) = W \Phi^{c}(X) W^*$ for all $X \in \Li(\mcl{X})$ (this follows easily from the usual equivalence between Kraus representations up to a unitary \cite[Cor. 2.23]{watrous2018theory}). In what follows, we give a characterization of $\Omega$-phase retrievability of a quantum channel in terms of its complementary channels. Before we proceed, we need the following definitions.

\begin{definition} \label{def: Omega admissibility}
A quantum channel $\Phi : \Li(\mcl{X}) \to \Li(\mcl{Y})$ is called \emph{$\Omega$-admissible} if it has a Kraus representation $\Phi(X) = \sum_{a\in \Sigma} A_a X A_a^*$ such that the measurement $\mu: \Sigma \times \Sigma \to \Li(\mcl{X})$ defined by $\mu(a,b) = A_a^*A_b$ is $\Omega$-phase retrievable.
\end{definition}

For example, for $\Omega = \Li(\mathcal{X})$, the trace channel $X \mapsto \Tr(X)$ is $\Omega$-admissible. Definition \ref{def: Omega admissibility} generalizes the notion of pure-state informational completeness used in \cite{liu2026enhancing}. We prefer not to adopt the term ``\emph{informationally complete}" here, since it would be inconsistent with our terminology where informational completeness is the same as phase retrievability. 
\\

Suppose $\Phi$ is $\Omega$-admissible with respect to the Kraus operators $\mcl{A} = \{A_a\}_{a \in \Sigma}$. By Corollary \ref{cor: Omega PR measurement} and the unitary equivalence between Kraus representations \cite[Cor. 2.23]{watrous2018theory}, this property is invariant under changes of Kraus representations. Indeed if $\mcl{A}' = \{B_{b }\}_{b\in \Gamma}$ is another Kraus representation of $\Phi$, then without loss of generality, we can assume $|\Gamma| \le |\Sigma|$ so that we could pad with zeros the set $\mcl{A}'$ so it has the same cardinality as $\mcl{A}$. From \cite[Cor. 2.23]{watrous2018theory} there exists a unitary operator $U \in \U(\mbb{F}^\Sigma)$ such that 
\[ 
B_a = \sum_{b}u_{ab}A_b,
\]

where $u_{ab}$ is the $(a,b)$ entry of $U$. Consequently, $\spa\{B_a^*B_b: a,b \in \Sigma\} \subseteq \spa\{A_a^*A_b: a,b \in \Sigma\}$. By reversing the roles of $\mcl{A}$ and $\mcl{A}'$, the opposite inclusion also holds. Hence the two spans coincide. By Corollary \ref{cor: Omega PR measurement} and the assumption that $\Phi$ is $\Omega$-admissible with respect to $\mcl{A}$, it follows that $\Phi$ is also $\Omega$-admissible with respect to $\mcl{A}'$. Hence, $\Phi$ is $\Omega$-admissible with respect to any Kraus representation. The definition that follows is from \cite{duan2012zero}.

\definition \label{def: confusability graph} For a quantum channel $\Phi:\Li(\mcl{X}) \to \Li(\mcl{Y})$ with Kraus operators $\mcl{A} = \{A_a\}_{a \in \Sigma}$, the subspace $\mcl{S}_\Phi = \spa\{A_a^*A_b: a,b\in \Sigma\}$ is called the \emph{confusability graph} of $\Phi$.
\\

\begin{remark} \label{rem: confusability graph not dependent on kraus rep}
    The reader might wonder why we do not specify the dependence on the Kraus representation. Applying \cite[Cor. 2.23]{watrous2018theory} it follows that $\mcl{S}_\Phi$ is not dependent on the Kraus representation of $\Phi$, as pointed out in \cite[Lemma 1]{duan2012zero}. 
    Similarly, if $\Phi:\Li(\mcl{X}) \to \Li(\mcl{Y})$ is a quantum channel and $W : \mcl{Y} \to \mcl{Z}$ is an isometry, the channel $\Psi : \Li(\mcl{X}) \to \Li(\mcl{Z})$ given by $\Psi(X) = W\Phi(X)W^*$ satisfies $\mcl{S}_\Phi = \mcl{S}_\Psi$; that is, the confusability graph stays the same if we place the space $\mcl{Y}$ inside a larger one $\mcl{Z}$ (formally, it is invariant under conjugation by an isometry).
\end{remark}

We now prove a result that characterizes the $\Omega$-phase retrievability of a quantum channel in terms of how the  orthogonal complement of the confusability graph of its complementary channels interacts with the set $\Omega - \Omega$.
It provides a generalization of \cite[Thm. 3]{liu2026enhancing}.
Note that it characterizes a property of a channel in terms of another property for a complementary channel, an approach with an illustrious history in Quantum Information e.g. the complementarity of private and correctable subsystems in quantum error correction \cite{kretschmann2008complementarity},
degradable and antidegradable channels
\cite{MR2160796,MR2464600},
and information–disturbance principles
such as the Schumacher--Nielsen coherent-information criterion
\cite{MR1411837}.

\begin{theorem} \label{thm: confusability graph of Phi}
    Let $\Phi : \Li(\mcl{X}) \to\Li( \mcl{Y})$ be a quantum channel. The following are equivalent:
    \begin{enumerate}[(a)]
        \item \label{conf_graph_of_Phi_a} $\Phi$ is $\Omega$-phase retrievable.
        
        \item \label{conf_graph_of_Phi_b} There exists one complementary channel $\Phi^c$ to $\Phi$  such that $(\Omega - \Omega) \, \cap \, 
        S_{\Phi^c}^\perp = \{0\}$.
        
        \item \label{conf_graph_of_Phi_c} For every complementary channel $\Phi^c$ to $\Phi$ we have $(\Omega - \Omega) \, \cap \, S_{\Phi^c}^\perp = \{0\}$.

        \item \label{conf_graph_of_Phi_d}  There exists one complementary channel $\Phi^c$ to $\Phi$ which is $\Omega$-admissible.

        \item \label{conf_graph_of_Phi_e} Every complementary channel $\Phi^c$ to $\Phi$  is $\Omega$-admissible.

    \end{enumerate}
\end{theorem}

\begin{proof} 

\eqref{conf_graph_of_Phi_a}$\iff$\eqref{conf_graph_of_Phi_b}: Let $\mcl{A} = \{A_a\}_{a \in \Sigma}$ be a set of Kraus operators of $\Phi$ and let $\Phi_{\mcl{A}}^c$ be the canonical complementary channel associated with $\mcl{A}$ as defined in \eqref{eqn: complementary channel}. If the Stinespring representation for $\Phi$ associated with $\mcl{A}$ (in the sense of \cite[Thm. 2.22]{watrous2018theory}) is
\[
 \Phi(X) = \Tr_{\mbb{F}^\Sigma}(A X A ^ {*}),
\]
then it is well-known (and it follows from a simple calculation) that the associated complementary channel $\Phi_{\mcl{A}}^c : \L(\mcl{X}) \to \L(\mbb{F}^\Sigma)$ is given by
\[
\Phi_{\mcl{A}}^c(X) = \Tr_{\mcl{Y}}(A X A ^ {*}).
\]

Applying \cite[Lem. 1]{duan2012zero} to the quantum channel $\Phi_{\mcl{A}}^c$ we have $\mcl{S}_{\Phi^c_\mcl{A}} = \range(\Phi^*)$.
\\

Let $R_\mcl{A} = \{R_j : j \in \Gamma\}$ be the corresponding set of Kraus operators of $\Phi^c_{\mcl{A}}$. Since in an $\mbb{F}$-Euclidean space a set and its linear span have the same orthogonal complement, Lemma \ref{lem: injectivity} and Theorem \ref{thm: injectivity and phase retrievability of a quantum channel} together imply that $\Phi$ is $\Omega$-phase retrievable if and only if $(\Omega - \Omega) \, \cap \, \{R_j^* R_k : j,k \in \Gamma\}^\perp = \{0\}$.
\\

\eqref{conf_graph_of_Phi_b}$\iff$\eqref{conf_graph_of_Phi_c}:
This holds because all complementary channels to $\Phi$ have the same confusability graph, since in any pair of them one of them is a conjugation of the other by an isometry, and this operation leaves the confusability graph the same (see Remark \ref{rem: confusability graph not dependent on kraus rep}).
\\

\eqref{conf_graph_of_Phi_b}$\iff$\eqref{conf_graph_of_Phi_d}: These are seen to be equivalent by combining the definition of $\Omega$-admissibility with the statement of Corollary \ref{cor: Omega PR measurement}.
\\

\eqref{conf_graph_of_Phi_d}$\iff$\eqref{conf_graph_of_Phi_e}: This follows directly from the fact that complementary channels are unique up to conjugation by an isometry.
\end{proof}
\begin{remark}
    As we conclude this section, we observe that Theorem \ref{thm: confusability graph of Phi} extends to the naturally to the more general class of completely positive maps, since the trace-preservation property of quantum channels was not used in any of the preceding arguments.
\end{remark}


\section{Characterizations of $\mcl{P}_k$-Phase Retrievability via spanning properties} \label{Sec: Xter PKPRM}

The task of determining when a Hermitian measurement is $\mcl{P}_k$-phase retrievable is frequently non‑trivial. To facilitate this analysis, we establish several equivalent criteria for $\mcl{P}_k$-phase retrievability. Later sections illustrate its usefulness through applications. The following well‑known characterization of positive semidefinite operators will be used repeatedly:
``\emph{Given an $\mbb{F}$-Euclidean space $\mcl{X}$, an operator $X \in \Pos(\mcl{X})$ has rank at most $k$ if and only if there exists $U \in \Li(\mbb{F}^{k}, \mcl{X})$ such that $X = UU^*$}.'' This follows immediately from the spectral theorem together with \cite[Prop. 1.7]{watrous2018theory}, applied in the setting of positive semidefinite operators. We begin by proving more general theorems, from which the characterization we require emerges as a direct consequence when applied in a specific setting. The first is an equivalence of the $\mcl{P}_k$-phase retrievability of a super operator in terms of the behavior of some linear operator in its kernel.

\theorem \label{thm: prso p_k} Let $\mcl{X}, \mcl{Y}$ be  $\mbb{F}$-Euclidean spaces, and $\Phi: \Li(\mcl{X}) \to \Li(\mcl{Y})$ a super operator. Then the following are equivalent:
\begin{enumerate}[(a)]
    \item \label{thm: prso p_k a}  The super operator $\Phi$ is $\mcl{P}_k$-phase retrievable.
    
    \item \label{thm: prso p_k b} For any $X,Y \in \Li(\mbb{F}^k, \mcl{X}), \Phi(XY^* +YX^*) = 0$ implies $XY^* + YX^* = 0$.
\end{enumerate}

\begin{proof}
    (\ref{thm: prso p_k a})$\implies$(\ref{thm: prso p_k b}): Assume (\ref{thm: prso p_k a}) is true and $X,Y \in \Li(\mbb{F}^k, \mcl{X})$ are such that $ \Phi(XY^* +YX^*) = 0$. Let $U = X+Y$ and $V = X-Y$. Then $UU^* - VV^* = 2(XY^* + YX^*)$. Now $UU^*, VV^* \in \mcl{P}_k$. Therefore, $\Phi(UU^*) = \Phi(VV^*)$. By (\ref{thm: prso p_k a}), $UU^* = VV^*$, and this implies $XY^* + YX^* = 0$.
    \\
    
    (\ref{thm: prso p_k b})$\implies$(\ref{thm: prso p_k a}): Assume (\ref{thm: prso p_k b}) holds. Let $\rho, \sigma \in \mcl{P}_k$ be such that $\Phi(\rho) = \Phi(\sigma)$. Let $U, V \in \Li(\mbb{F}^k, \mcl{X})$ satisfy $\rho = UU^*$ and $\sigma = VV^*$. Let $X = U+V$ and $Y =U-V$. Then $XY^*+YX^* = 2(\rho - \sigma)$ and $\Phi(XY^* + YX^*) =0$. By (\ref{thm: prso p_k b}), $\rho - \sigma = 0$.
\end{proof}

Theorem \ref{thm: prso p_k} provides a natural extension of the result in \cite[Lem. 2.7]{liu2023phase}, showing that the same conclusion holds for $\mcl{P}_k$ phase retrievable super operators in general. The second result identifies additional equivalent conditions of $\mcl{P}_k$-phase retrievability of a super operator under the assumption of commuting with the adjoint.

\theorem \label{thm: Herm-pres} Let $\mcl{X}$ and $\mcl{Y}$ be $\mbb{F}$-Euclidean spaces. Let $\Phi: \Li(\mcl{X}) \to \Li(\mcl{Y})$ be a super operator that commutes with the adjoint. The following are equivalent.

\begin{enumerate}[(a)] 
    \item \label{thm: Herm-pres a} The super operator $\Phi$ is $\mcl{P}_k$-phase retrievable.
    \item \label{thm: Herm-pres b} For any $X, Y \in \L(\mbb{F}^k, \mcl{X})$, $\Herm \Phi(XY^*) = 0$ implies $XY^* + YX^* = 0$.

    \item \label{thm: Herm-pres c} For any $U, V \in \L(\mbb{F}^k, \mcl{X})$, $\Herm \Phi[(U+V)(U-V)^*] = 0$ implies $UU^* = VV^*$.
    
\end{enumerate}

\begin{proof}
    (\ref{thm: Herm-pres a})$\implies$(\ref{thm: Herm-pres b}): Suppose that (\ref{thm: Herm-pres a}) is true. Let $X, Y \in \Li(\mbb{F}^k, \mcl{X})$ be such that $\Herm \Phi(XY^*)= 0$. Then 
    \[
    \Phi(XY^*+ YX^*) = \Phi(XY^*) + \Phi(YX^*) = \Phi(XY^*) + \Phi(XY^*)^* = 2 \Herm \Phi(XY^*) = 0,
    \]
    where the second equality follows clearly from commutation with the adjoint. From Theorem \ref{thm: prso p_k}, it follows that $XY^* + YX^* = 0$.
    \\
    
    (\ref{thm: Herm-pres b})$\implies$(\ref{thm: Herm-pres a}): Assume statement (\ref{thm: Herm-pres b}) is true and $X,Y \in \Li(\mbb{F}^k, \mcl{X})$ satisfy $\Phi(XY^* + YX^*) = 0$. Then 
    \[
    2 \Herm \Phi(XY^*) = \Phi(XY^*) + \Phi(XY^*)^* = \Phi(XY^*) + \Phi(YX^*) = \Phi(XY^* + YX^*) = 0.
    \]
    From (\ref{thm: Herm-pres b}), it holds that $XY^* + YX^* = 0$. By Theorem \ref{thm: prso p_k}, $\Phi$ is $\mcl{P}_k$-phase retrievable.
    \\
    
    (\ref{thm: Herm-pres b})$\iff$(\ref{thm: Herm-pres c}): Assume that (\ref{thm: Herm-pres b}) holds and $U,V \in \Li(\mbb{F}^k, \mcl{X})$ satisfy $\Herm \Phi[(U+V)(U-V)^*]=0$. Taking $X = U+V$ and $Y = U-V$, then $\Herm \Phi(XY^*) = 0$. Applying the assumption (\ref{thm: Herm-pres b}) it follows that  $XY^* + YX^* = 2(UU^* - VV^*) = 0$. Conversely, if the statement in (\ref{thm: Herm-pres c}) holds, then an analogous argument shows that (\ref{thm: Herm-pres b}) is true.
\end{proof}
Within the class of super operators that commute with the adjoint, Theorem \ref{thm: Herm-pres} furnishes two further equivalent characterizations of $\mcl{P}_k$-phase retrievability, given in parts \eqref{thm: Herm-pres b} and (\ref{thm: Herm-pres c}). In general, these additional criteria are more tractable to check.
We now use the first and second results we have to find conditions for the $\mcl{P}_k$-phase retrievability of a particular super operator that commutes with the adjoint in terms of the orthogonal complement of the span of a set, its dimension and the rank of some linear map. We note that given operators $X_1, X_2, \dots, X_n$, where for every $j$, $X_j \in \Li(\mcl{X})$, the real span of the operators is denoted $\spa_\mbb{R}\{X_j\}_{j=1}^n$. The real inner product in $\Li(\mcl{X}, \mcl{Y})$ is defined as $\pair{A}{B}_\mbb{R} = \Re \pair{A}{B} = \Re(\Tr(A^*B))$. For any subspace $V \subseteq \Li(\mcl{X}, \mcl{Y})$, the real dimension of $V$ is denoted $\dim_\mbb{R}(V)$. In this setting, we denote the orthogonal complement of $X \in \Li(\mcl{X}, \mcl{Y})$ by $X^{\perp_\mbb{R}}$. We state two results, one for the complex case, the other for the real case.
To help the reader make sense of the somewhat complicated statements, let us spend a moment relating them to previous work.

In the real case, a beautiful result from \cite{MR2224902} is that a collection $\{y_a\}_{a\in\Sigma}$ of vectors in the real Euclidean space $\mcl{X}$ does phase retrieval (in the Frame Theory sense that the map $x \mapsto \big( |\pair{y_a}{x}| \big)_{a \in \Sigma}$ is injective up to a phase) if and only if it satisfies the \emph{complement property}: for every partition of $\Sigma$ into two disjoint subsets $\Sigma_1,\Sigma_2$, either $\spa_{\mathbb{R}}\{y_a\}_{a\in\Sigma_1} = \mcl{X}$ or $\spa_{\mathbb{R}}\{y_a\}_{a\in\Sigma_2} = \mcl{X}$.
It is not too difficult to see that this property is equivalent to the following \emph{spanning property}: for every $x \in \mcl{X}\setminus\{0\}$, $\spa_\mathbb{R}\{y_ay_a^*x\}_{a\in\Sigma} = \mcl{X}$.
In the complex case, the conditions get more complicated but still depend on the same span.
According to \cite[Thm. 4]{MR3202304}, in the complex case the phase retrievability of $\{y_a\}_{a\in\Sigma}$ is equivalent to either of the following two equivalent conditions:
\begin{enumerate}[(a)]
\item For every $x \in \mcl{X}\setminus\{0\}$,  $\dim_\mathbb{R} \spa_\mathbb{R}\{y_ay_a^*x\}_{a\in\Sigma} \ge 2 \dim_\mathbb{C}\mcl{X} - 1=\dim_\mathbb{R}\mcl{X} - 1$.
\item For every $x \in \mcl{X}\setminus\{0\}$, $\spa_\mathbb{R}\{y_ay_a^*x\}_{a\in\Sigma}$ is the orthogonal complement of $ \spa_{\mathbb{R}}\{i x\}$ with respect to the real inner product $\Re \pair{\cdot}{\cdot}$ on $\mcl{X}$.
\end{enumerate}

In our language, the phase retrieval property from Frame Theory corresponds to the $\mcl{P}_1$-phase retrievability of the associated super operator $\Li(\mcl{X}) \to \Li(\mbb{F}^\Sigma)$ given by $X \mapsto \sum_{a \in \Sigma} \pair{y_ay_a^*}{X} \E_{a,a}$; note that this super operator takes values in diagonal matrices.
Analogous results in a more general setting are proved in 
\cite[Thms. 2.1 and 2.2]{MR3981280}, once again involving two equivalent properties of a judiciously chosen span: one involving its dimension, and one involving its orthogonal complement.
Translating to our language, these results from  \cite{MR3981280} correspond to $\mcl{P}_1$-phase retrievability of the super operator associated to a more general Hermitian measurement $\{Y_a\}_{a\in \Sigma}$, which again is a super operator taking values in diagonal matrices.

Below, our Theorems \ref{thm: prso and span complex} and \ref{thm: prso and span} more generally characterize the $\mcl{P}_k$-phase retrievability of super operators that commute with the adjoint and take values in block-diagonal matrices.
This is once again done in terms of either the orthogonal complement or the dimension of a certain span.

Before stating the main results we need some preparatory ones.

\begin{lemma} \label{lem: dimension_complex}
    Let $\mcl{X}$ be a $\mbb{C}$-Euclidean space of dimension $d$, and $X \in \Li(\mbb{C}^k,\mcl{X})$ a nonzero linear operator of rank $r$. For the subspace $V = \{Y: XY^*+YX^* = 0 \} \subseteq \Li( \mbb{C}^k,\mcl{X})$ we have $\dim_\mbb{R}(V) = r^2 + 2d(k-r)$.
\end{lemma}
 \begin{proof}

 From the assumptions, we may assume $X \in M_{d \times k}(\mbb{C})$. We choose a basis such that $X$ and $Y$ have block representations
\[X = \begin{bmatrix}
D & 0\\
0 & 0
\end{bmatrix}
\text{ and } Y = \begin{bmatrix}
    Y_1 & Y_2\\
    Y_3 & Y_4
\end{bmatrix},\]
where $D \in M_{r \times r}(\mbb{C})$ is invertible, $Y_1 \in M_{r \times r}(\mbb{C}), Y_2 \in M_{r \times (k-r)}(\mbb{C}), Y_3 \in M_{(d-r)\times r}(\mbb{C}) \text{ and } Y_4 \in M_{(d-r) \times (k-r)}(\mbb{C})$.\\

Therefore, it follows that 
\[XY^* + YX^*  
= \begin{bmatrix}
    DY_1^* + Y_1D^*& DY_3^* \\
    Y_3D^* & 0
\end{bmatrix}
=
\begin{bmatrix}
    0 & 0 \\
    0 & 0
\end{bmatrix}. 
\] 

Since $D$ is invertible, $Y_3 = 0$, and this means 
\[
Y = \begin{bmatrix}
    Y_1 & Y_2\\
    0 & Y_4\\
\end{bmatrix}.
\] 
We have $DY_1^* + Y_1D^* = 0 \iff DY_1^* = -(D Y_1^*)^*$. This means that $DY_1^* \in M_{r \times r}(\mbb{C})$ is skew Hermitian. Define a map $L : \{W \in M_{r \times r}(\mbb{C}): DW^* + WD^* = 0\} \longrightarrow \{A \in M_{r \times r}(\mbb{C}) : A \text{ skew Hermitian}\}$ by $L(W) = DW^*$. Then $L$ is an invertible linear map between finite dimensional real vector spaces, with inverse $L^{-1}(A) = (D^{-1}A)^*$ for all skew Hermitian operator $A \in M_{r \times r}(\mbb{C})$. Hence, it is an isomorphism. Taking $i = \sqrt{-1}$, every skew Hermitian matrix $A \in M_{r \times r}(\mbb{C})$ is of the form
\[
A = \begin{bmatrix}
    ia_{1,1} & a_{1,2} & a_{1,3} & \dots & a_{1,r}\\
    -\overline{a_{1,2}} & ia_{2,2} & a_{2,3} & \dots & a_{2,r}\\
    -\overline{a_{1,3}} & -\overline{a_{2,3}} & ia_{3,3} & \dots & a_{3,r}\\
    \vdots & \vdots & \vdots & \ddots & \vdots\\
    -\overline{a_{1,r}} & \dots & \dots & \dots & ia_{r,r}\\
\end{bmatrix},
\]
where $a_{j,j} \in \mbb{R}$. 
Hence, the real dimension of the space $\{A \in M_{r \times r}(\mbb{C}) : A \text{ is skew Hermitian}\}$ is $2\Big [r + (r-1)+ \dots + 1\Big ] - r = r^2$. Therefore, $\dim_{\mbb{R}}(Y_1) = r^2$. Now $Y_2$ and $Y_4$ are free matrices. Therefore, $\displaystyle \dim_{\mbb{R}}(V) = \sum_{j = 1,2,4} \dim_{\mbb{R}}(Y_j) =  r^2 + 2d(k-r)$.
\end{proof}

\begin{corollary} \label{cor: dimension_complex}
    Let $\mcl{X}$ be a $\mbb{C}$-Euclidean space of dimension $d$, and $X \in \Li(\mbb{C}^k, \mcl{X})$ a nonzero linear operator of rank $r$. For the subspace $V = \{Y: XY^*+YX^* = 0 \} \subseteq \Li(\mbb{C}^k, \mcl{X})$ we have $\dim_\mbb{R}(V^{\perp_\mbb{R}}) = r(2d-r)$.
\end{corollary}
\begin{proof}
    Since $\dim_\mbb{R}(\Li(\mbb{C}^k,\mcl{X})) = 2dk$, we have $\dim_\mbb{R}(V^{\perp_\mbb{R}}) =  2dk - (r^2 + 2d(k-r)) = r(2d-r)$.   
\end{proof}

We record the next lemma and its corresponding corollary without proofs, as their verification mirrors the arguments already established in the proofs of Lemma \ref{lem: dimension_complex} and Corollary \ref{cor: dimension_complex}.
\begin{lemma} \label{lem: dimension_real}
    Let $\mcl{X}$ be an $\mbb{R}$-Euclidean space of dimension $d$, and $X \in \Li(\mbb{R}^k,\mcl{X})$ a nonzero linear operator of rank $r$. For the subspace $V = \{Y: XY^*+YX^* = 0 \} \subseteq \Li( \mbb{R}^k,\mcl{X})$ we have $\dim(V) = \frac{r(r-1)}{2} + d(k-r)$.
\end{lemma}

\begin{corollary} \label{cor: dimension_real}
    Let $\mcl{X}$ be an $\mbb{R}$-Euclidean space of dimension $d$, and $X \in \Li( \mbb{R}^k,\mcl{X})$ a nonzero linear operator of rank $r$. For the subspace $V = \{Y: XY^*+YX^* = 0 \} \subseteq \Li( \mbb{R}^k,\mcl{X})$ we have $\dim(V^{\perp}) = r\Big(d- \frac{r-1}{2}\Big)$.
\end{corollary}

We are now ready for the promised characterization of $\mcl{P}_k$-phase retrievability for certain super operators taking values in block-diagonal matrices.
We start with the complex case, whose proof is more involved.

Before proceeding, recall that a quantum instrument \cite{heinosaari2011mathematical,jordan2024quantum} is a collection of completely positive maps $\{ \Phi_a : \L(\mcl{X}) \to \L(\mcl{Y}) \}_{a\in \Sigma}$ such that $\sum_{a\in\Sigma} \Phi_a$ is trace preserving, and its action is given by the channel $\Phi : \L(\mcl{X}) \to \L(\mcl{Y} \otimes \mbb{F}^\Sigma)$ defined by
\[
\Phi(X) = \sum_{a\in\Sigma} \Phi_a(X) \otimes E_{aa}.
\]
It follows from the Stinespring representation that each such channel admits a representation as in item \eqref{thm: prso and span complex a} of the theorem below, illustrating that the result applies to a class of channels of great interest.

\theorem \label{thm: prso and span complex}
Let $\mcl{X},\mcl{Z}$ be $\mbb{C}$-Euclidean spaces and assume that $\mcl{X}$ has dimension $d$. Let $\Sigma$ be a finite set. For every $a \in \Sigma$, let $A_a, B_a \in \L(\mcl{X}, \mcl{Z}\otimes \mbb{C}^s)$. Let $\mcl{P}_k = \{\rho \in \Pos(\mcl{X}): \rank(\rho) \le k\}$.  The following are equivalent:

\begin{enumerate}[(a)]
        \item \label{thm: prso and span complex a} The super operator $\Phi : \L(\mcl{X}) \to \L(\mathbb{C}^s \otimes \mathbb{C}^\Sigma)$  defined by
\[
\Phi(\rho) = \sum_{a\in \Sigma} \Tr_{\mcl{Z}}(A_a\rho A_a^* - B_a \rho B_a^*)  \otimes E_{aa}
\]
is $\mcl{P}_k$-phase retrievable.

        \item \label{thm: prso and span complex b} For every $X \in \L(\mathbb{C}^k,\mcl{X}) \setminus \{0\}$ with $\rank(X) = r$, we have
        \[
        \big(\spa_{\mbb{R}} \big\{(A_a^*(\mathbbm{1}_{\mcl{Z}}\otimes C)A_a - B_a^*(\mathbbm{1}_{\mcl{Z}}\otimes C)B_a) X\big\}_{a\in\Sigma, C \in \Herm(\mbb{C}^s)}\big)^{\perp_{\mbb{R}}} = \{ Y \in \L(\mathbb{C}^k,\mcl{X}) \; :\; XY^*+YX^*=0 \}.
        \]
        \item \label{thm: prso and span complex c} For every $X \in \L(\mathbb{C}^k,\mcl{X}) \setminus \{0\}$ with $\rank(X) = r$, we have 
        \[
        \dim_\mbb{R} (\spa_{\mbb{R}} \{(A_a^*(\mathbbm{1}_{\mcl{Z}}\otimes C)A_a - B_a^*(\mathbbm{1}_{\mcl{Z}}\otimes C)B_a) X\}_{a\in\Sigma, C \in \Herm(\mbb{C}^s)}) \ge r(2d-r).
        \]
    
        \item \label{thm: prso and span complex d} For every $X \in \L(\mathbb{C}^k,\mcl{X}) \setminus \{0\}$ with $\rank(X) = r$, and basis $\{C_1, C_2, \dots, C_m \}$ of $\Herm(\mbb{C}^s)$, where $m = s^2$, the $\mbb{R}$-linear map $T_{X} : \mathbb{R}^{\Sigma \times \{1,2, \dots, m\}} \to \L(\mathbb{C}^k,\mcl{X})$ given by 
        \[
        T_{X}(c_{aj})_{a\in\Sigma, j \in \{1,2,\dots, m\}} = \sum_{a\in\Sigma} \sum_{j=1}^m c_{aj}(A_a^* (\mathbbm{1}_{\mcl{Z}} \otimes C_j)A_a - B_a^* (\mathbbm{1}_{\mcl{Z}} \otimes C_j)B_a)X
        \]
        satisfies $\rank(T_{X}) \ge r(2d-r)$.
    \end{enumerate}
    
\begin{proof}
    (\ref{thm: prso and span complex a})$\implies$(\ref{thm: prso and span complex b}): Assume that the statement in (\ref{thm: prso and span complex a}) holds and let $X \in \Li(\mbb{C}^k, \mcl{X}) \setminus \{0\}$ with $\rank(X) = r$. If $Y \in \Li(\mbb{C}^k, \mcl{X})$ satisfies $XY^* + YX^* = 0$, then for every $a \in \Sigma$ and every $C \in \Herm(\mbb{C}^s)$, the definition of the Frobenius inner product together with the cyclicity of the trace implies that
    \[
    \begin{aligned}
    \pair{(A_a^*(\mathbbm{1}_{\mcl{Z}}\otimes C)A_a - B_a^*(\mathbbm{1}_{\mcl{Z}}\otimes C)B_a) X}{Y} & = \Tr((A_a^*(\mathbbm{1}_{\mcl{Z}}\otimes C)A_a - B_a^*(\mathbbm{1}_{\mcl{Z}}\otimes C)B_a)YX^*) \\
    & = \pair{(A_a^*(\mathbbm{1}_{\mcl{Z}}\otimes C)A_a - B_a^*(\mathbbm{1}_{\mcl{Z}}\otimes C)B_a)}{YX^*}.
    \end{aligned}
    \]
    Invoking cyclicity of the trace again, and the fact that taking adjoints conjugates the trace, we conclude that
    \[
    \pair{(A_a^*(\mathbbm{1}_{\mcl{Z}}\otimes C)A_a - B_a^*(\mathbbm{1}_{\mcl{Z}}\otimes C)B_a) X}{Y} = \overline{\pair{(A_a^*(\mathbbm{1}_{\mcl{Z}}\otimes C)A_a - B_a^*(\mathbbm{1}_{\mcl{Z}}\otimes C)B_a)}{XY^*}}.
    \]
    Therefore, we have
    \[
    \begin{aligned}
         0 = \pair{(A_a^*(\mathbbm{1}_{\mcl{Z}}\otimes C)A_a - B_a^*(\mathbbm{1}_{\mcl{Z}}\otimes C)B_a)}{XY^* + YX^*} & = \pair{(A_a^*(\mathbbm{1}_{\mcl{Z}}\otimes C)A_a - B_a^*(\mathbbm{1}_{\mcl{Z}}\otimes C)B_a)}{XY^*} \\
         & + \pair{(A_a^*(\mathbbm{1}_{\mcl{Z}}\otimes C)A_a - B_a^*(\mathbbm{1}_{\mcl{Z}}\otimes C)B_a)}{YX^*}\\
         & = 2 \Re \pair{(A_a^*(\mathbbm{1}_{\mcl{Z}}\otimes C)A_a - B_a^*(\mathbbm{1}_{\mcl{Z}}\otimes C)B_a)}{YX^*}.
    \end{aligned}  
    \]
    This means $\Re \pair{(A_a^*(\mathbbm{1}_{\mcl{Z}}\otimes C)A_a - B_a^*(\mathbbm{1}_{\mcl{Z}}\otimes C)B_a)}{YX^*} = 0$ and from this it follows that
    \[
    \Re \pair{(A_a^*(\mathbbm{1}_{\mcl{Z}}\otimes C)A_a - B_a^*(\mathbbm{1}_{\mcl{Z}}\otimes C)B_a)X}{Y} = 0.
    \]
    Hence, $Y \in (\spa_\mbb{R} \{(A_a^*(\mathbbm{1}_{\mcl{Z}}\otimes C)A_a - B_a^*(\mathbbm{1}_{\mcl{Z}}\otimes C)B_a) X\}_{a\in\Sigma, C \in \Herm(\mbb{C}^s)}\big)^{\perp_{\mbb{R}}}$. Note that this shows that the containment $\supseteq$ in statement \eqref{thm: prso and span complex b} is always true.\\
    
    For the other containment, we first show that for every $X \ne 0, Y \in \Li(\mbb{C}^k, \mcl{X})$, $\Herm \Phi(XY^*) = 0 \iff Y \in (\spa_\mbb{R} \{(A_a^*(\mathbbm{1}_{\mcl{Z}}\otimes C)A_a - B_a^*(\mathbbm{1}_{\mcl{Z}}\otimes C)B_a) X\}_{a\in\Sigma, C \in \Herm(\mbb{C}^s)}\big)^{\perp_{\mbb{R}}}$. Indeed, for every $X \ne 0, Y \in \Li(\mbb{C}^k, \mcl{X})$ it follows that

    \[
    \begin{aligned}
        \Herm \Phi(XY^*) = 0 & \iff \Herm \Tr_\mcl{Z}(A_a XY^* A_a^* - B_a XY^* B_a^*)=0 \text{ for all } a \in \Sigma\\
        & \iff \Re \pair{\Herm \Tr_\mcl{Z}(A_a XY^* A_a^* - B_a XY^* B_a^*)}{V}=0 \text{ for all } a \in \Sigma, V \in \Li(\mbb{C}^s)\\
        & \iff \Re \pair{A_a XY^* A_a^* - B_a XY^* B_a^*}{\mathbbm{1}_\mcl{Z} \otimes \Herm(V)} = 0 \text{ for all } a \in \Sigma, V \in \Li(\mbb{C}^s)\\
        & \iff \Re \pair{A_a XY^* A_a^* - B_a XY^* B_a^*}{\mathbbm{1}_\mcl{Z} \otimes C} = 0 \text{ for all } a \in \Sigma, C \in \Herm (\mbb{C}^s)\\
        & \iff \Re \pair{XY^*}{A_a^*(\mathbbm{1}_\mcl{Z}\otimes C) A_a - B_a^*(\mathbbm{1}_\mcl{Z} \otimes C)B_a)} = 0 \text{ for all } a \in \Sigma, C \in \Herm(\mbb{C}^s)\\
        & \iff \Re \pair{(A_a^*(\mathbbm{1}_\mcl{Z}\otimes C) A_a - B_a^*(\mathbbm{1}_\mcl{Z} \otimes C)B_a))X}{Y} = 0 \text{ for all } a \in \Sigma, C \in \Herm(\mbb{C}^s)\\
        & \iff Y \in (\spa_\mbb{R} \{(A_a^*(\mathbbm{1}_{\mcl{Z}}\otimes C)A_a - B_a^*(\mathbbm{1}_{\mcl{Z}}\otimes C)B_a) X\}_{a\in\Sigma, C \in \Herm(\mbb{C}^s)}\big)^{\perp_{\mbb{R}}}.
    \end{aligned}
    \]
    We note that the first equivalence holds because $\Phi$ is a block diagonal super operator. In the second equivalence, one direction is trivial and the other one follows from observing that if for every $a \in \Sigma$ and $V \in \Li(\mbb{C}^s)$ we have 
    \[\Re \pair{\Herm \Tr_\mcl{Z}(A_a XY^* A_a^* - B_a XY^* B_a^*)}{V}=0,
    \]
    then taking $V = \Herm \Tr_\mcl{Z}(A_a XY^* A_a^* - B_a XY^* B_a^*)$ it holds that
    \[
    \n{\Herm \Tr_\mcl{Z}(A_a XY^* A_a^* - B_a XY^* B_a^*)}_2^2 = \Re \pair{\Herm \Tr_\mcl{Z}(A_a XY^* A_a^* - B_a XY^* B_a^*)}{V} = 0.
    \]
    Therefore, for every $a \in \Sigma$, $\Herm \Tr_\mcl{Z}(A_a XY^* A_a^* - B_a XY^* B_a^*) = 0$.
    The third equivalence holds by the easily verified fact that $A \to \Herm(A)$ is self adjoint with respect to the real inner product $\pair{\cdot}{\cdot}_\mbb{R}$, and the well-known fact that the adjoint of $\Tr_\mcl{Z}$ is tensoring with $\mathbbm{1}_\mcl{Z}$.

    From Theorem \ref{thm: Herm-pres} we conclude that $Y \in (\spa_\mbb{R} \{(A_a^*(\mathbbm{1}_{\mcl{Z}}\otimes C)A_a - B_a^*(\mathbbm{1}_{\mcl{Z}}\otimes C)B_a) X\}_{a\in\Sigma, C \in \Herm(\mbb{C}^s)}\big)^{\perp_{\mbb{R}}}$ implies $XY^* + YX^* = 0$.
    \\
    
    (\ref{thm: prso and span complex b})$\implies$(\ref{thm: prso and span complex a}): We now assume that (\ref{thm: prso and span complex b}) holds. Let $X \ne 0, Y \in \Li(\mbb{C}^k, \mcl{X})$ satisfy $\Phi(XY^* + YX^*) = 0$. Then for every $a \in \Sigma$,
    \[
    \Tr_{\mcl{Z}}(A_a(XY^* + YX^*)A_a^* - B_a (XY^* + YX^*) B_a^*) = 0
    \]
    if and only if for every $a \in \Sigma$ and every $C \in \Herm(\mbb{C}^s)$,
    \[
    \pair{A_a(XY^* + YX^*) A_a^* - B_a (XY^* + YX^*) B_a^*}{\mathbbm{1}_\mcl{Z}\otimes C} = 0,
    \]
    if and only if for every $a \in \Sigma$ and every $C \in \Herm(\mbb{C}^s)$,
    \[
    \pair{XY^*+YX^*}{A_a^*(\mathbbm{1}_{\mcl{Z}}\otimes C)A_a - B_a^*(\mathbbm{1}_{\mcl{Z}}\otimes C)B_a } = 0.
    \]
    Using linearity, cyclicity of the trace and the fact that taking adjoints conjugates the trace, we obtain that for every $a \in \Sigma$ and every $C \in \Herm(\mbb{C}^s)$,
    \[
    \Re \pair{Y}{(A_a^*(\mathbbm{1}_\mcl{Z}\otimes C)A_a - B_a^*(\mathbbm{1}_\mcl{Z}\otimes C)B_a)X}= 0,
    \]
    if and only if $Y \in (\spa_\mbb{R} \{(A_a^*(\mathbbm{1}_{\mcl{Z}}\otimes C)A_a - B_a^*(\mathbbm{1}_{\mcl{Z}}\otimes C)B_a) X\}_{a\in\Sigma, C \in \Herm(\mbb{C}^s)})^{\perp_\mbb{R}}$. From (\ref{thm: prso and span complex b}) it follows that $XY^* + YX^* = 0$. By Theorem \ref{thm: prso p_k}, we conclude that $\Phi$ is $\mcl{P}_k$-phase retrievable.
    \\
    
    (\ref{thm: prso and span complex b})$\iff$(\ref{thm: prso and span complex c}): Assume that (\ref{thm: prso and span complex b}) is true and $X \in \Li(\mbb{C}^k, \mcl{X}) \setminus \{0\}$ with $\rank(X) = r$. Now the set $\{Y \in \Li(\mbb{C}^k, \mcl{X}) : XY^* + YX^* = 0\}$ is a subspace of $\Li(\mbb{C}^k , \mcl{X})$ with real dimension $r^2 + 2d(k-r)$. Since it always holds that
    \[
    (\spa_\mbb{R} \{(A_a^*(\mathbbm{1}_{\mcl{Z}}\otimes C)A_a - B_a^*(\mathbbm{1}_{\mcl{Z}}\otimes C)B_a) X\}_{a\in\Sigma, C \in \Herm(\mbb{C}^s)}\big)^{\perp_\mbb{R}} \supseteq \{ Y \in \L(\mathbb{C}^k,\mcl{X}) \; :\; XY^*+YX^*=0 \},
    \]
    it follows that
    \[
    \spa_\mbb{R} \{(A_a^*(\mathbbm{1}_{\mcl{Z}}\otimes C)A_a - B_a^*(\mathbbm{1}_{\mcl{Z}}\otimes C)B_a) X\}_{a\in\Sigma, C \in \Herm(\mbb{C}^s)} \subseteq \{ Y \in \L(\mathbb{C}^k,\mcl{X}) \; :\; XY^*+YX^*=0 \}^{\perp_\mbb{R}}.
    \]
    Now we have
    \[
    \dim_\mbb{R}(\{ Y \in \L(\mathbb{C}^k,\mcl{X}) \; :\; XY^*+YX^*=0 \}^{\perp_\mbb{R}}) = 2dk - (r^2 + 2d(k-r)) = r(2d-r).
    \]
    Therefore, it follows that
    \[
    \dim_\mbb{R}(\spa_\mbb{R} \{(A_a^*(\mathbbm{1}_{\mcl{Z}}\otimes C)A_a - B_a^*(\mathbbm{1}_{\mcl{Z}}\otimes C)B_a) X\}_{a\in\Sigma, C \in \Herm(\mbb{C}^s)}) \le r(2d-r).
    \]
    From (\ref{thm: prso and span complex b}), we conclude that 
    \[
    \dim_\mbb{R}(\spa_\mbb{R} \{(A_a^*(\mathbbm{1}_{\mcl{Z}}\otimes C)A_a - B_a^*(\mathbbm{1}_{\mcl{Z}}\otimes C)B_a) X\}_{a\in\Sigma, C \in \Herm(\mbb{C}^s)}) \ge r(2d-r).
    \]
    Conversely if, the statement in (\ref{thm: prso and span complex c}) holds, then the statement in (\ref{thm: prso and span complex b}) follows since
    \[
    (\spa_\mbb{R} \{(A_a^*(\mathbbm{1}_{\mcl{Z}}\otimes C)A_a - B_a^*(\mathbbm{1}_{\mcl{Z}}\otimes C)B_a) X\}_{a\in\Sigma, C \in \Herm(\mbb{C}^s)}\big)^{\perp_\mbb{R}} \supseteq \{ Y \in \L(\mathbb{C}^k,\mcl{X}) \; :\; XY^*+YX^*=0 \}
    \]
    always holds as pointed out above.
    \\
    
    (\ref{thm: prso and span complex c})$\iff$(\ref{thm: prso and span complex d}): The statements in (\ref{thm: prso and span complex c}) and (\ref{thm: prso and span complex d}) are equivalent because for every $X \in \L(\mathbb{C}^k,\mcl{X}) \setminus \{0\}$ and for every basis $\{C_1, C_2, \dots, C_m \}$ of $\Herm(\mbb{C}^s)$, where $m = s^2$, the $\mbb{R}$-linear map $T_{X} : \mathbb{R}^{\Sigma \times \{1,2, \dots, m\}} \to \L(\mathbb{C}^k,\mcl{X})$ given by 
        \[
        T_{X}(c_{aj})_{a\in\Sigma, j \in \{1,2,\dots, m\}} = \sum_{a\in\Sigma} \sum_{j=1}^m c_{aj}(A_a^* (\mathbbm{1}_{\mcl{Z}} \otimes C_j)A_a - B_a^* (\mathbbm{1}_{\mcl{Z}} \otimes C_j)B_a)X
        \]
        satisfies 
        \[
        \range(T_{X}) = \spa_\mbb{R} \{(A_a^*(\mathbbm{1}_{\mcl{Z}}\otimes C)A_a - B_a^*(\mathbbm{1}_{\mcl{Z}}\otimes C)B_a) X\}_{a\in\Sigma, C \in \Herm(\mbb{C}^s)}.
        \]
\end{proof}

Specializing Theorem \ref{thm: prso and span complex} to a sum with just one term, we immediately obtain the following characterization of completely positive $\mcl{P}_k$-phase retrievable super operators in terms of their Stinespring representations.

\begin{corollary}\label{cor: spanning property for CP}
Let $\mcl{X},\mcl{Z}$ be $\mbb{C}$-Euclidean spaces and assume that $\mcl{X}$ has dimension $d$. Let $A \in \L(\mcl{X}, \mcl{Z}\otimes \mbb{C}^s)$. Let $\mcl{P}_k = \{\rho \in \Pos(\mcl{X}): \rank(\rho) \le k\}$.  The following are equivalent:

\begin{enumerate}[(a)]
        \item 
        The completely positive super operator $\Phi : \L(\mcl{X}) \to \L(\mathbb{C}^s)$  defined by the Stinespring representation
\[
\Phi(\rho) = \Tr_{\mcl{Z}}(A\rho A^* ) 
\]
is $\mcl{P}_k$-phase retrievable.

        \item For every $X \in \L(\mathbb{C}^k,\mcl{X}) \setminus \{0\}$ with $\rank(X) = r$, we have
        \[
        \big(\big\{A^*(\mathbbm{1}_{\mcl{Z}}\otimes C)A X\big\}_{C \in \Herm(\mbb{C}^s)}\big)^{\perp_{\mbb{R}}} = \{ Y \in \L(\mathbb{C}^k,\mcl{X}) \; :\; XY^*+YX^*=0 \}.
        \]
        \item For every $X \in \L(\mathbb{C}^k,\mcl{X}) \setminus \{0\}$, we have 
        \[
        \dim_\mbb{R} \big( \spa_\mbb{R} \big\{A^*(\mathbbm{1}_{\mcl{Z}}\otimes C)A X\big\}_{C \in \Herm(\mbb{C}^s)}\big) \ge r(2d-r),
        \]
        where $r= \rank(X)$.
        \item For every $X \in \L(\mathbb{C}^k,\mcl{X}) \setminus \{0\}$  the linear map $T_{X} : \Herm(\mbb{C}^s) \to \L(\mathbb{C}^k,\mcl{X})$ given by 
        \[
        T_{X} C=  A^* (\mathbbm{1}_{\mcl{Z}} \otimes C)AX
        \]
        satisfies $\rank(T_{X}) \ge r(2d-r)$, where $r= \rank(X)$.
    \end{enumerate}    
\end{corollary}

The statement of Theorem \ref{thm: prso and span complex} has a corresponding version in the real case, which we now state without proof, as its argument parallels that of the preceding theorem. The proof repeatedly uses the linearity, the cyclicity, and the invariance of the trace under transposition on real vector spaces.
Naturally, there is also a real case version of Corollary \ref{cor: spanning property for CP} obtained by specializing to a sum with just one term.

\theorem \label{thm: prso and span}
Let $\mcl{X},\mcl{Z}$ be $\mbb{R}$-Euclidean spaces and assume that $\mcl{X}$ has dimension $d$. Let $\Sigma$ be a finite set. For every $a \in \Sigma$, let $A_a, B_a \in \L(\mcl{X}, \mcl{Z}\otimes \mbb{R}^s)$. Let $\mcl{P}_k = \{\rho \in \Pos(\mcl{X}): \rank(\rho) \le k\}$.  The following are equivalent:

\begin{enumerate}[(a)]
        \item \label{thm: prso and span a} The super operator $\Phi : \L(\mcl{X}) \to \L(\mathbb{R}^s \otimes \mathbb{R}^\Sigma)$  defined by
\[
\Phi(\rho) = \sum_{a\in \Sigma} \Tr_{\mcl{Z}}(A_a\rho A_a^T - B_a \rho B_a^T)  \otimes \E_{a,a}
\]
is $\mcl{P}_k$-phase retrievable.

        \item \label{thm: prso and span b} For every $X \in \L(\mathbb{R}^k,\mcl{X}) \setminus \{0\}$ with $\rank(X) = r$, we have
        \[
        (\spa \{(A_a^T(\mathbbm{1}_{\mcl{Z}} \otimes C)A_a - B_a^T(\mathbbm{1}_{\mcl{Z}} \otimes C)B_a) X\}_{a\in\Sigma, C \in \Herm(\mbb{R}^s)}\big)^\perp = \{ Y \in \L(\mathbb{R}^k,\mcl{X}) \; :\; XY^T+YX^T=0 \}.
        \]
        \item \label{thm: prso and span c} For every $X \in \L(\mathbb{R}^k,\mcl{X}) \setminus \{0\}$, we have 
        \[
        \dim \spa \{(A_a^T(\mathbbm{1}_{\mcl{Z}}\otimes C)A_a - B_a^T(\mathbbm{1}_{\mcl{Z}}\otimes C)B_a) X\}_{a\in\Sigma, C \in \Herm(\mbb{R}^s)} \ge r (d - \tfrac{r-1}{2}),
        \]
        where $r= \rank(X)$.
        \item \label{thm: prso and span d} For every $X \in \L(\mathbb{R}^k,\mcl{X}) \setminus \{0\}$ with $\rank(X) = r$, and for every basis $\{C_1, C_2, \dots, C_m \}$ of $\Herm(\mbb{R}^s)$, where $m = \tfrac{s(s+1)}{2}$, the linear map $T_{X} : \mathbb{R}^{\Sigma \times \{1,2, \dots, m\}} \to \L(\mathbb{R}^k,\mcl{X})$ given by 
        \[
        T_{X}(c_{aj})_{a\in\Sigma, j \in \{1,2,\dots, m\}} = \sum_{a\in\Sigma} \sum_{j=1}^m c_{aj}(A_a^T (\mathbbm{1}_{\mcl{Z}} \otimes C_j)A_a - B_a^T (\mathbbm{1}_{\mcl{Z}} \otimes C_j)B_a)X
        \]
        satisfies $\rank(T_{X}) \ge r(d -\tfrac{r-1}{2})$.
    \end{enumerate}

    \remark \label{rem: Herm measurement} If $s=1$ in Theorem \ref{thm: prso and span complex} or Theorem \ref{thm: prso and span}, then the partial trace with respect to $\mcl{Z}$ reduces to a trace and  by the linearity and cyclicity of the trace operator it follows that
\[
\Phi(\rho) = \sum_{a \in \Sigma} \pair{A_a^*A_a - B_a^* B_a}{\rho} \E_{a,a}
\]
for all $\rho \in \Li(\mcl{X})$.
This is the super operator corresponding to the Hermitian measurement $\mu : \Sigma \to \Li(\mcl{X})$ defined by $\mu(a) = A_a^* A_a - B_a^* B_a$ for all $a \in \Sigma$. The Jordan-Hahn decomposition for Hermitian operators provides that every Hermitian measurement of the form $\Sigma \to \Li(\mcl{X})$ can be expressed this way. If $B_a = 0$ for all $a \in \Sigma$, then $\mu$ is a positive measurement with $\mu(a) = A^*_a A_a$ for all $a \in \Sigma$. Every positive measurement of the form $\Sigma \to \Li(\mcl{X})$ can be expressed this way and by Proposition \ref{prop: positive meas and so}, $\Phi$ is a positive super operator.  If in addition $\sum_{a \in \Sigma}A_a^*A_a = \mathbbm{1}_{\mcl{X}}$, then $\Phi$ is the quantum-to-classical channel corresponding to the quantum measurement $\mu : \Sigma \to \Li(\mcl{X})$ defined by $\mu(a) = A_a^*A_a$.\\

Using Remark \ref{rem: Herm measurement}, we obtain several convenient characterizations of $\mcl{P}_k$-phase retrievable measurements and $\mcl{P}_k$-phase retrievable Hermitian measurements. From Theorem \ref{thm: equivalence of P-phase retrievable measurements and channels} and Theorem \ref{thm: prso p_k}, we obtain the following measurement analogue of Theorem \ref{thm: prso p_k}.

\corollary \label{cor: KPRM} Let $k \in \mbb{N}$, $\mcl{X}$ an $\mbb{F}$-Euclidean space, $\Sigma$ a finite set, and $\mu : \Sigma \to \Li(\mcl{X})$ a measurement. The following are equivalent:
\begin{enumerate}[(a)]
    \item \label{KPRM a} The measurement $\mu$ is $\mcl{P}_k$-phase retrievable.

    \item \label{KPRM b} For any $X, Y \in \Li(\mbb{F}^k, \mcl{X})$, having $\pair{\mu(a)}{XY^* + YX^*} = 0$ for all $a \in \Sigma$ implies $XY^* + YX^* = 0$.
\end{enumerate}

From Proposition \ref{prop: Herm meas and so}, Theorem \ref{thm: equivalence of P-phase retrievable measurements and channels}, and Theorem \ref{thm: Herm-pres} we have the following corollary which is a Hermitian measurement version of Theorem \ref{thm: Herm-pres}.

\corollary \label{cor: KPR Herm M} Let $k \in \mbb{N}$, $\mcl{X}$ an $\mbb{F}$-Euclidean space, $\Sigma$ a finite set, and $\mu : \Sigma \to \Herm(\mcl{X})$ a Hermitian measurement. The following are equivalent:

\begin{enumerate}[(a)]
\item\label{k-PRM-a} The Hermitian measurement $\mu$ is $\mcl{P}_k$-phase retrievable.

    \item \label{k-PRM-c} For any $X,Y\in \L(\mbb{F}^k,\mcl{X})$, having $\Re \pair{\mu(a)}{X Y ^ *}=0$ for all $a \in \Sigma$ implies $XY^* +YX^*=0$.

    \item \label{k-PRM-d} For any $U,V\in \L(\mbb{F}^k,\mcl{X})$ having $\Re \pair{\mu(a)}{(U+V)(U-V)^*} = 0$ for all $a \in \Sigma$ implies  $UU^* = VV^*$.
\end{enumerate}

In the corollary that follows, we use Proposition \ref{prop: composed meas Herm} and Corollary \ref{cor: KPR Herm M} to characterize $\mcl{P}_k$-phase retrievable composed measurements.
\begin{corollary}\label{cor:k-phase retrievable s-operators} Let $k \in \mbb{N}$, let $\mcl{X}$ and $\mcl{Y}$ be $\mbb{F}$-Euclidean spaces, and let $\Sigma$ be a finite set. Let $\Phi : \Li(\mcl{X}) \to \Li(\mcl{Y})$ be a super operator that commutes with the adjoint, and $\mu:\Sigma \to \Herm(\mcl{Y})$ a Hermitian measurement. The following are equivalent:

\begin{enumerate}[(a)]
\item \label{k-PRs-operator-a} The composed measurement $\Phi^* \circ \mu$ is $\mcl{P}_k$-phase retrievable.

    \item \label{k-PRs-operator-c} For any $X,Y\in \L(\mbb{F}^k,\mcl{X})$, having $\Re \pair{\mu(a)}{\Phi(YX ^ {*})}=0$ for all $a \in \Sigma$ implies $XY^* +YX^*=0$.

    \item \label{k-PRs-operator-d} For any $U,V\in \L(\mbb{F}^k,\mcl{X})$ having $\Re \pair{\mu(a)}{\Phi((U-V)(U ^ * + V^*))} = 0$ for all $a \in \Sigma$ implies  $U U^* = VV^*$.
\end{enumerate}
\end{corollary}
\remark When $k=1$, the operators $U,V,X$ and $Y$ that appear in Corollaries \ref{cor: KPRM}, \ref{cor: KPR Herm M}, and \ref{cor:k-phase retrievable s-operators} reduce to column vectors.

\section{Existence of nontrivial $\mcl{P}_k$-Phase Retrievable  Quantum Channels} \label{Sec: PKPRQC}
Every injective quantum channel defined is obviously $\mcl{P}_k$-phase retrievable. Therefore, it is important we find those that are $\mcl{P}_k$-phase retrievable but not injective on $\Li(\mcl{X})$. This is the goal we will accomplish in this section. In \cite[Example 1.1]{liu2023phase} it was shown that there exist non‑injective quantum channels that are nevertheless injective on pure states, i.e. they are $\D_1$-phase retrievable. The argument relies on tools from Frame Theory, but those methods do not appear to extend naturally to states of higher rank. In contrast, we show that the phenomenon does generalize once one employs techniques from Quantum Information Theory. In particular, \cite[Prop. 1, Prop. 2, and Thm. 1]{heinosaari2013quantum} provide the crucial ingredients that enable this extension.
We remark that this is one of the few sections where the choice of scalars will make a difference.

 A key ingredient in this construction is the notion of an \emph{operator system}. Operator systems form a fundamental class of Hermitian subspaces of matrix algebras, extensively studied in Functional Analysis and Operator Theory. They serve as the natural setting for analyzing the adjoints of quantum channels and their ranges. 
\begin{definition} \label{def: operator system} 

Given an  $\mbb{F}$-Euclidean space $\mcl{X}$, a subspace $\mcl{S} \subseteq \Li(\mcl{X})$ is called an \emph{operator system} if $\mathbbm{1}_{\mcl{X}} \in \mcl{S}$ and $X ^ * \in S$ for all $X \in \mcl{S}$.
\\

For example, given a quantum channel $\Phi: \Li(\mcl{X}) \to \Li(\mcl{Y})$, its confusability graph $\mcl{S}_\Phi \subseteq \Li(\mcl{X})$ is an operator system.
Furthermore, it is known that every operator system is of this form \cite[Prop. IV.1]{MR3860769}.
    
\end{definition}

\begin{lemma} \label{lem: operator system and measurement}
    Let $\mcl{X}$ be an $\mbb{F}$-Euclidean space. Let $\mcl{P}_k = \{X \in \Pos(\mcl{X}) : \rank(X) \le k\}$ and $\D_k = \{\rho \in \mcl{P}_k : \Tr(\rho) = 1\}$.
    \begin{enumerate}[(a)]
        \item \label{operator system and measurement a} Let $\mcl{S} \subseteq \Li(\mcl{X})$ be an operator system. The following are equivalent:
        \begin{enumerate}[(i)]
            \item \label{operator system and measurement a(i)} $(\mcl{P}_k- \mcl{P}_k) \, \cap \, \mcl{S}^\perp = \{0\}$ 

            \item \label{operator system and measurement a(ii)} $(\D_k- \D_k)\, \cap \, \mcl{S}^\perp = \{0\}$. 
        \end{enumerate}
    
      \item \label{operator system and measurement b} If $\mbb{F} = \mbb{C}$ and $\mu: \Sigma \to \Pos(\mcl{X})$ is a quantum measurement, then the following are equivalent:
        \begin{enumerate}[(i)]
           \item \label{operator system and measurement b(i)} $\mu$ is $\mcl{P}_k$-phase retrievable. 
           \item \label{operator system and measurement b(ii)} $\mu$ is $\D_k$-phase retrievable.
           
           \item \label{operator system and measurement b(iii)} $(\mcl{P}_k- \mcl{P}_k) \, \cap \, \{\mu(a)\}_{a \in \Sigma}^\perp = \{0\}$
    
           \item \label{operator system and measurement b(iv)} $(\D_k- \D_k)\, \cap \, \{\mu(a)\}_{a \in \Sigma}^\perp = \{0\}$.
       \end{enumerate}
   \end{enumerate}
\end{lemma}
\begin{proof}
    \eqref{operator system and measurement a}: Since $\D_k \subseteq \mcl{P}_k$, $(\mcl{P}_k- \mcl{P}_k) \, \cap \, \mcl{S}^\perp = \{0\}$  implies $(\D_k- \D_k)\, \cap \, \mcl{S}^\perp = \{0\}$. Now assume that the latter holds, and let $\rho, \sigma \in \mcl{P}_k$ be such that $\rho-\sigma \in (\mcl{P}_k- \mcl{P}_k) \, \cap \, \mcl{S}^\perp$, then $\mathbbm{1}_\mcl{X} \in \mcl{S}$ implies $\Tr(\rho-\sigma) = \pair{\rho-\sigma}{\mathbbm{1}_\mcl{X}} = 0$. Therefore $\Tr(\rho) = \Tr(\sigma)$. 
    If $\Tr(\rho) = \Tr(\sigma)=0$, then $\rho=\sigma=0$ and thus $\rho-\sigma=0$.
    Otherwise $\Tr(\rho) = \Tr(\sigma)>0$,
  and therefore, $\tfrac{1}{\Tr(\rho)}(\rho -\sigma) \in (\D_k- \D_k)\, \cap \, \mcl{S}^\perp$, so $\rho-\sigma = 0$. 
    \\
    
    \eqref{operator system and measurement b}: In the complex case, operator systems coincide with spans of quantum measurements.
    Indeed, given a $\mbb{C}$-Euclidean space $\mcl{X}$, a finite set $\Sigma$, and a quantum measurement $\mu: \Sigma \to \Pos(\mcl{X})$, the subspace $\mcl{S} = \spa\{\mu(a)\}_{a \in \Sigma} \subseteq \Li(\mcl{X})$ is clearly an operator system. From \cite[Prop. 1]{heinosaari2013quantum}, for every operator system $\mcl{S}\subseteq \Li(\mcl{X})$, we can find a finite set $\Sigma$ and a measurement $\mu: \Sigma \to \Pos(\mcl{X})$ such that $\mcl{S} = \spa\{\mu(a)\}_{a \in \Sigma}$. Therefore, combining Corollary \eqref{cor: Omega PR measurement} and part \eqref{operator system and measurement a}, we have the equivalencies \eqref{operator system and measurement b(i)}$\iff$\eqref{operator system and measurement b(iii)}$\iff$\eqref{operator system and measurement b(iv)}$\iff$\eqref{operator system and measurement b(ii)}.
\end{proof}

\begin{remark} \label{rem: operator system and quantum measurement equivalence}

    The real case analogous correspondence for part \eqref{operator system and measurement b} of Lemma \ref{lem: operator system and measurement} does not hold because the real span of a collection of symmetric matrices consists of only symmetric matrices but clearly there exist operator systems containing non-symmetric matrices (e.g. $\mcl{S} = \Li(\mcl{X})$).
\end{remark}

\begin{corollary} \label{cor: DKPRQC and PKQRC Quantum Channel}
    Let $\Sigma$ be a finite set and suppose $\mcl{X}$ and $\mcl{Y}$ are $\mbb{C}$-Euclidean spaces. Let $\mcl{P}_k = \{X \in \Pos(\mcl{X}) : \rank(X) \le k\}$ and $\D_k = \{\rho \in \mcl{P}_k : \Tr(\rho) = 1\}$.  A quantum channel $\Phi : \Li(\mcl{X}) \to \Li(\mcl{Y})$ is $\mcl{P}_k$-phase retrievable if and only if it is $\D_k$-phase retrievable.
\end{corollary}

The following lemma is one of the two crucial tools in this section. Combining it with the constructions of certain operator systems from \cite{heinosaari2013quantum} will yield the quantum channels we are interested in.

\begin{lemma} \label{lem: non-quantum-to-classical}
    Let $n \ge 2$ be a positive integer and $\mcl{X}$ an $n$-dimensional $\mbb{F}$-Euclidean space. For every operator system $\mcl{S} \subseteq \Li(\mcl{X})$ with $\dim(\mcl{S}) \ge 4$ and positive integer $q \in \mbb{N}$ satisfying both $\binom{q}{2} \ge \dim(\mcl{S}) - 1$ and $\dim(\mcl{S}) > q$, there exists a non-quantum-to-classical channel $\Phi: \Li(\mcl{X}) \to \Li(\mbb{F}^q)$ such that
    \[
    \mcl{S} = \range(\Phi^*).
    \]
\end{lemma}
\begin{proof}

By \cite[Prop. IV.1]{MR3860769} (which is stated only in the complex case, but the same argument works in the real one), the assumption $\binom{q}{2} \ge \dim(\mcl{S})-1$ implies we can find a quantum channel $\Psi: \Li(\mcl{X}) \to \Li(\mbb{F}^{nq})$ such that $\mcl{S}$ is the confusability graph of $\Psi$. Therefore we have $\mcl{S} = S_\Psi = \spa\{A_k^*A_l : 1\le k,l \le q\}$, where $\{A_k\}_{k=1}^q$ are Kraus operators of $\Psi$. It follows from \cite[Lem. 1]{duan2009super} that $\mcl{S}_\Psi = \range((\Psi^c)^*)$, where $\Psi^c$ is the canonical complementary channel associated with $\{A_k\}_{k=1}^q$. We now show that under the assumption that $\dim(\mcl{S}) > q$, it holds that $\Psi^c$ is a non-quantum-to-classical channel. Since the highest rank of a quantum-to-classical channel from $\Li(\mcl{X})$ to $\Li(\mbb{F}^q)$ is $q$, it suffices to show that $\rank ((\Psi^c)) > q$. Using the Rank-Nullity Theorem and the fact that $\range((\Psi^c)^*)^\perp = \ker(\Psi^c)$ we have,
\[
n^2 = \dim(\Li(\mcl{X})) = \rank(\Psi^c) + \dim(\ker(\Psi^c)) = \rank(\Psi^c) + \dim((\range(\Psi^c)^*)^\perp).
\]

Now $\dim(\range((\Psi^c)^*)^\perp) = \dim(\mcl{S}^\perp)= n^2 - \dim(\mcl{S})$. Hence, $\rank(\Psi^c) = \dim(\mcl{S})$. Since $\dim(\mcl{S}) > q$ we conclude that $\Psi^c$ is not quantum-to-classical. 
\end{proof}

\begin{remark}
    In Lemma \ref{lem: non-quantum-to-classical}, the assumption $\dim(\mcl{S}) \ge 4$ guarantees the existence of a positive integer $q \in \mbb{N}$ such that $\binom{q}{2} \ge \dim(\mcl{S}) - 1$, and $\dim(\mcl{S}) > q$: it is easy to see that in this case $q = \dim(\mcl{S}) - 1$ is a valid choice.
\end{remark}

Using Lemma \ref{lem: non-quantum-to-classical} we state the following result, which establishes the existence of $\mcl{P}_k$-phase retrievable,  non-injective quantum channels.

\begin{theorem} \label{thm: Non-injective PKPR QC}
     Let $\mcl{X}$ be a $\mbb{F}$-Euclidean space and $m, n,k \in \mbb{N}$ with $\dim(\mcl{X}) = n \ge 3$, $1 \le k < \frac{n}{2}$, and $m=4k(n-k)$. Let $\mcl{P}_k = \{\rho \in \Pos(\mcl{X}) : \rank(\rho) \le k\}$.

     \begin{enumerate}[(a)]
         \item \label{thm: Non-injective PKPR QC a} If $\mbb{F}= \mbb{C}$, there exists a $\mcl{P}_k$-phase retrievable and non-injective quantum-to-classical-channel of the form $\Phi:\Li(\mcl{X}) \to \Li(\mbb{C}^m)$.

         \item  \label{thm: Non-injective PKPR QC b} If $q \in \mbb{N}$ is a positive integer such that $\binom{q}{2} \ge m -1$ and $m > q$, then there exists a $\mcl{P}_k$-phase retrievable non-injective, not quantum-to-classical-channel of the form $\Phi:\Li(\mcl{X}) \to \Li(\mbb{F}^q)$.
     \end{enumerate}
\end{theorem}
\begin{proof}
    For a fixed $\mbb{F}$-Euclidean space $\mcl{X}$, and positive integers $n,k\in \mbb{N}$ satisfying $1 \le k < n/2$, an analysis of the proof of \cite[Thm. 1]{heinosaari2013quantum} shows that there exists an operator system $\mcl{S} \subseteq \L(\mcl{X})$ such that
    $(\D_k-\D_k) \cap \mcl{S}^\perp = \{0\}$ and $\dim (\mcl{S}) = 4k(n-k)$. From Lemma \ref{lem: operator system and measurement} we have $(\mcl{P}_k-\mcl{P}_k) \cap \mcl{S}^\perp = \{0\}$.
    While \cite{heinosaari2013quantum} works with complex scalars, it is easy to see that the argument works equally well in the real case.
    \\

     \eqref{thm: Non-injective PKPR QC a} By \cite[Prop. 1]{heinosaari2013quantum} there exists a $\D_k$-phase retrievable quantum measurement $\mu: \{1, \dots, m\} \to \Pos(\mcl{X})$ such that $\spa\{\mu(a)\}_{a \in \Sigma} = \mcl{S}$. Since $\spa\{\mu(a)\}_{a \in \Sigma} ^ \perp = \{\mu(a)\}_{a \in \Sigma}^\perp$, it follows from Lemma \ref{lem: operator system and measurement} that $\mu$ is $\mcl{P}_k$-phase retrievable. Applying Corollary \ref{cor: equivalence of k-phase retrievable measurements and channels}, the quantum-to-classical channel $\Psi_\mu :\ \Li(\mcl{X}) \to \Li(\mbb{C} ^ m)$ defined as
    
    \[
    \Psi_\mu (X) = \sum_{a = 1} ^ m \pair{\mu(a)}{X} E _{a,a}
    \]
    
    for all $X \in \Li(\mcl{X})$, is $\mcl{P}_k$-phase retrievable, and from Theorem \ref{thm: injectivity and phase retrievability of a quantum channel} it is injective on $\mcl{P}_k$.
    
    It remains to show that $\Psi_\mu$ is not an injective super operator on $\Li(\mcl{X})$. To this end, observe that $\range (\Psi_\mu) \subseteq \Di(\mbb{C}^m)$, where $\Di(\mbb{C}^m)$ denotes the space of diagonal operators on $\mbb{C}^m$. Since $\Di(\mbb{C}^m)$ has dimension $m$, it follows that $\rank(\Psi_\mu) \le m$. Now $m<n^2 = \dim(\Li(\mcl{X}))$. To see this, fix $n$ and observe that $4k(n-k) \le n^2$, with equality if only if $k=\frac{n}{2}$. Since $k <\frac{n}{2}$ by assumption, it follows that $m < n^2$. Hence, $\rank(\Psi_\mu) \leq m < n^2$, so that $\Psi_\mu$ is not injective.
    \\
    
    \eqref{thm: Non-injective PKPR QC b} The assumptions on $n$ and $k$ imply $m \ge 8$. Since $q \in \mbb{N}$ satisfies $\binom{q}{2} \ge \dim(\mcl{S}) - 1$ and $\dim(\mcl{S}) > q$  it follows from Lemma \ref{lem: non-quantum-to-classical} that there exists a not-quantum-to-classical channel $\Phi: \Li(\mcl{X}) \to \Li(\mbb{F}^{q})$ such that $\mcl{S} = \range(\Phi^*)$.
     We have
    \[
    \rank(\Phi) = \rank(\Phi^*) = \dim(\mcl{S}) = m < n^2.
    \]
Therefore, $\Phi$ is not injective. Since $(\mcl{P}_k - \mcl{P}_k) \, \cap \, \range(\Phi^*)^\perp = \{0\}$, it follows from Lemma \ref{lem: injectivity} and Theorem \ref{thm: injectivity and phase retrievability of a quantum channel} that $\Phi$ is $\mcl{P}_k$-phase retrievable.
\end{proof}

\begin{remark}
Regarding Theorem \ref{thm: Non-injective PKPR QC}:
\begin{enumerate}[(i)]
    \item In part \eqref{thm: Non-injective PKPR QC a}, we expect that arguments in the style of  \cite{MR3671475} can show that a ``generic'' choice of measurement with $|\Sigma|$ large enough should be $\mcl{P}_k$-phase retrievable. We do not attempt to address this in the present work.

    \item In part \eqref{thm: Non-injective PKPR QC b}, one would like to have the dimension of the codomain of the channel to be as small as possible, ideally comparable to the dimension of the domain.
Note that for large $n = \dim(\mcl{X})$, $q$ can be taken to be essentially $\sqrt{8k(n-k)}$. For a fixed $k$ this is of the order of $\sqrt{n}$, a marked improvement. When $k$ is proportional to $n$, $q$ will also be proportional to $n$.

\item Using our stability results from later in the paper, starting from the quantum-to-classical channel in part \eqref{thm: Non-injective PKPR QC a} by perturbation one can obtain another channel that is not quantum-to-classical and still satisfies both of the desired conditions.
However, this is very unsatisfactory because the dimension of the codomain is much bigger than what we obtained in part \eqref{thm: Non-injective PKPR QC b}.
\end{enumerate}
\end{remark}

The next result should be understood as a refinement of Theorem \ref{thm: Non-injective PKPR QC}, and also of \cite[Thm. 1]{heinosaari2013quantum}.
Our proof is essentially the same as that of the latter, adding some careful extra choices.

\begin{theorem} \label{thm: PKPR not PK+1PR QC}
Let $k\in\mbb{N}$. Let $n=2(k+1)$ and $m=4k(n-k)=4k(k+2)$.
    Let $\mcl{X}$ be a $\mbb{F}$-Euclidean space with $\dim(\mcl{X}) = n$. Let $\mcl{P}_k = \{\rho \in \Pos(\mcl{X}) : \rank(\rho) \le k\}$.

     \begin{enumerate}[(a)]
         \item \label{thm: PKPR not PK+1PR QC a} If $\mbb{F}=\mbb{C}$, there exists a quantum-to-classical quantum channel of the form $\Phi: \Li(\mcl{X}) \to \Li(\mbb{C}^m)$, which is $\mcl{P}_k$-phase retrievable but not $\mcl{P}_{k+1}$-phase retrievable. Equivalently, there exists a quantum measurement $\mu : \{1,2,\dotsc, m\} \to \Pos(\mcl{X})$ which is  $\mcl{P}_k$-phase retrievable but not $\mcl{P}_{k+1}$-phase retrievable.

         \item  \label{thm: PKPR not PK+1PR QC b} If $q \in \mbb{N}$ is a positive integer such that $\binom{q}{2} \ge m -1$ and $m > q$, then there exists a not quantum channel of the form $\Phi: \Li(\mcl{X}) \to \Li(\mbb{F}^q)$, which is $\mcl{P}_k$-phase retrievable but not $\mcl{P}_{k+1}$-phase retrievable.

     \end{enumerate}
     
\end{theorem}

\begin{proof}
    The proof is the same as that of Theorem \ref{thm: Non-injective PKPR QC}, with a slightly more particular choice of operator system $\mcl{S}$. Note that in addition to $(\mcl{P}_k-\mcl{P}_k) \cap \mcl{S}^\perp = \{0\}$, all we need to achieve is to also have $(\mcl{P}_{k+1}-\mcl{P}_{k+1}) \cap \mcl{S}^\perp \not= \{0\}$. 
    Our argument below follows the notation and terminology of the proof of \cite[Thm. 1]{heinosaari2013quantum}, just making some more specific choices along the way.
    \\
    
    For simplicity, we take $\mcl{X}=\mbb{F}^n$.
    Since $n = 2(k+1)$, then the only integer $r$ satisfying $2k + 1 \le r \le n - 1$ is $r = 2k + 1$. Let $V_r$ be the Vandermonde matrix generated by $(1, 2, ... , r)$, that is,
    \[ V_r = 
    \begin{pmatrix}
        1 & 1 & 1 & \dots & 1\\
        1 & 2 & 4 & \dots & 2^{r-1}\\
        \vdots & \vdots & \vdots & \ddots & \vdots\\
        1 & r & r^2 & \dots & r^{r-1}
    \end{pmatrix}.
    \]
    Then $V_r$ is \emph{totally nonsingular}, \footnote{A square matrix is totally nonsingular if all its minors are nonzero.} and in particular invertible. Let $v$ be the vector of the first column of $V_r$. We put $v$ on the \emph{$r$th diagonal}\footnote{For an $n \times n$ matrix $M$, the $r$th diagonal is the set of entries $M_{ij}$ with $j - i = n - r$.} of the $n \times n$ matrix $B_1$ with zeros elsewhere. Hence,
    \[
    B_1 =  \begin{pmatrix}
        0 & 1 & 0 & \dots & 0\\
        0 & 0 & 1 & \dots & 0\\
        \vdots & \vdots & \vdots & \ddots & \vdots\\
        0 & 0 & 0 & \dots & 1\\
        0 & 0 & 0 & \dots & 0
    \end{pmatrix}.
    \]
    Let $B_2$ be the transpose of $B_1$. Then $\rank(B_1) = \rank(B_2) = 2k+1$ and $\Tr(B_1) = \Tr(B_2) = 0$. Let $w = (w_j)_{j=1}^n$  represent the second column of the Vandermonde matrix $V_n$ generated by $(1,2, ...,n)$, so $w_j = j$. Let $u = (u_j)_{j = 1} ^ n$ be the transpose of the last row vector of $V_n ^ {-1}$, then $u$ is orthogonal to $w$; that is, $\pair{u}{w} = \sum_{j=1}^n u_j w_j = 0$. Let $S_k^j$ be the Stirling number of the first kind: the  coefficient of $x^j$ in $x(x-1)\cdots (x-k+1)$, for $k,j \ge 1$. Define \[\lambda_{ij} = \sum_{k=j}^n (-1)^{k+j} \binom{j}{i}S_k^j, \text{ } 1 \le i,j \le n.
    \]
    Then, by \cite{macon1958inverses}, the entries of $V_n$ are given by $(V_n)_{i,j}^{-1} = i\lambda_{ij}$. Furthermore, since $n>1$, except for the first row of $V_n^{-1}$, the entries in each row sum to zero. Therefore, for $u_j = (V_n^{-1})_{n,j} = n \lambda_{nj}$ it then follows that $\sum_{j=1}^n u_j =0$. 
    \\
    
    Let $\Tilde{w} = (\Tilde{w}_j)_{j = 1} ^n$ be a column vector with $\Tilde{w}_j = u_j w_{j}$. Using $u$ and $\Tilde{w}$ we construct $n \times n$ diagonal matrices $B_3$ and $B_4$ with $(B_3)_{j,j} = u_j$ and $(B_4)_{j,j} = \Tilde{w}_j$. Then $\rank(B_3) = \rank(B_4) = n$ and $\Tr(B_3) = \Tr(B_4) = 0$. By construction the matrices $B_1, B_2, B_3 \text{ and } B_4$ are linearly independent. Let $\mcl{B} = \spa\{B_1, B_2, B_3, B_4\}$, then $\mcl{B}$ satisfies conditions (a) to (d) listed in the proof of \cite[Thm. 1]{heinosaari2013quantum}, namely:
    \begin{itemize}
        \item[(i)] $\mcl{B}^ \dagger = \mcl{B}$, where $\mcl{B} ^ \dagger = \{B^* : B \in \mcl{B}\}$,
        \item[(ii)] $\Tr(B) = 0$ for all $B \in \mcl{B}$,
        \item[(iii)] $\dim(\mcl{B}) =4 = (n-2k)^2$, and
        \item[(iv)] $\rank(B) \ge 2k + 1$ for every nonzero $B \in \mcl{B}$.
    \end{itemize} 
    
    From here, taking $\mcl{S} = \mcl{B} ^ \perp$ yields an operator system with $\dim(\mcl{S}) = n ^ 2 - 4 = 4k(n-k) = m$.  
    \\

    Finally, let $P - Q$ be the Jordan-Hahn decomposition of $B_3$. Then $P, Q \in \mcl{P}_{k + 1}$ since all the entries in the diagonal of $B_3$ are nonzero, and the number of positive entries is $k+1$, which is also the number of negative entries. Therefore, $\rank(P) = \rank(Q) = k + 1$, and $ B_3 = P - Q \in (\mcl{P}_{k + 1} - \mcl{P}_{k + 1}) \cap \mcl{S}^\perp$. It follows that $(\mcl{P}_{k + 1} - \mcl{P}_{k + 1}) \cap \mcl{S}^\perp \ne \{0\}$.
\end{proof}

\section{Stability of $\mcl{P}_k$-phase retrievable super operators with respect to norms} \label{Sec: Stability of PKPRMSO}

In this section, we investigate the stability of $\mcl{P}_k$-phase retrievable super operators, establishing that they are stable. This means that $\mcl{P}_k$-phase retrievable super operators remain robust under small perturbations. By Definition \ref{def: Omega-PRs-operator} and Remark \ref{rem: Special case of P-phase retrievable quantum channels}, it follows that $\mcl{P}_k$-phase retrievable super operators are precisely super operators that are injective on positive semidefinite operators of rank at most $k$. Such super operators are directly connected to $\mcl{P}_k$-phase retrievable measurements. In real life, measurements are never accurate and since our mathematical models need to withstand some error, it will be satisfactory if super operators that arise naturally from $\mcl{P}_k$-phase retrievable measurements are able to withstand some error. The advantage is that if we replace a $\mcl{P}_k$-phase retrievable super operator (with respect to a measurement) with a possibly different $\mcl{P}_k$-phase retrievable super operator up to some error estimate, then the new super operator is bound to be $\mcl{P}_k$-phase retrievable (with respect to the same measurement). Before presenting the main theorems, we record some foundational results. The first, stated below, provides a bound on the ranks of the positive and negative parts in the Jordan-Hahn decomposition of a Hermitian operator. We suspect it may be well-known, but we are not aware of a convenient reference and thus  include an argument here for completeness.

\lemma \label{Hahn-Jordan bound 2} Let $A, B \ge 0$ be such that $\rank(A)\le k$ and $\rank(B) \le k$. If $P - Q$ is the Jordan-Hahn Decomposition of $A - B$, then $\rank(P) \le k$ and $\rank(Q) \le k$.
\begin{proof}
    We know that every eigenvalue of a Hermitian operator is real. Therefore, for every Hermitian operator $X$, we may order its eigenvalues as $\lambda_1 \ge \dots \ge \lambda_n$. The $j$th eigenvalue of $X$ is $\lambda_j(X)$. We denote the number of its positive and negative eigenvalues by $n_+(X)$ and $n_-(X)$, respectively. By \cite[Prop. 2.2(i)]{zheng2022inertia}, $n_+(-B) \le 0$ implies $\lambda_j(A-B) \le \lambda_j(A)$. Let $\Theta_1 = \{j \mid \lambda_j(A-B) >0\}$ and $\Theta_2 = \{j \mid \lambda_j(A) > 0\}$. If $j \in \Theta_1$, then $0 < \lambda_j(A-B) \le \lambda_j(A)$. Hence, $\Theta_1 \subseteq \Theta_2$. This implies $n_+(A-B) \le n_+(A)$. Since $n_+(A) \le \rank(A)$, we have $\rank(P) = n_+(A-B) \le k$. Using a similar argument as before for $B-A$, we have $n_+(B-A) \le n_+(B)$ and we have $\rank(Q) \le k$.
\end{proof}

The second result is a characterization of a $\mcl{P}_k$-phase retrievable super operators, just as Corollary \ref{cor: KPRM} characterizes $\mcl{P}_k$-phase retrievable measurements.

\begin{theorem} \label{thm: stability of Pk-PRSO}
    Let $k \in \mbb{N}$, $\mcl{X}$ and $\mcl{Y}$ $\mbb{F}$-Euclidean spaces, $\mcl{P}_k =\{\rho \in \Pos(\mcl{X}): \rank(\rho) \le k\}$, and $\Phi : \Li(\mcl{X})\longrightarrow \Li(\mcl{Y})$ a super operator. The following are equivalent:
    \begin{enumerate}[(a)]
        \item \label{stability of Pk-PRSO-a} 
        The super operator $\Phi$ is $\mcl{P}_k$-phase retrievable.
        \item \label{stability of Pk-PRSO-b} 
        For every $\rho, \sigma \in \mcl{P}_k$, not both zero, such that $\rho\sigma=0$, we have $\Phi(\rho) \ne \Phi(\sigma)$.

    \end{enumerate}
\end{theorem}
   \begin{proof} (\ref{stability of Pk-PRSO-a})$\implies$(\ref{stability of Pk-PRSO-b}) : Assume that (\ref{stability of Pk-PRSO-a}) is true. Let $\rho, \sigma \in \mcl{P}_k$, not both zero, such that $\rho\sigma=0$.  If $\Phi(\rho) = \Phi(\sigma)$, then by (\ref{stability of Pk-PRSO-a}) $\rho = \sigma$. It follows that $\rho ^ 2 = \sigma ^ 2 = 0$. Hence, $\rho = \sigma = 0$ and this gives a contradiction.

   (\ref{stability of Pk-PRSO-b})$\implies$(\ref{stability of Pk-PRSO-a}) : Assume that (\ref{stability of Pk-PRSO-b}) is true. Let $\rho, \sigma \in \mcl{P}_k$ be such that  $\Phi(\rho) = \Phi(\sigma)$. Assume by contradiction  that $\rho \ne \sigma$. Let $P - Q$ be the Jordan-Hahn decomposition of $\rho - \sigma$. Then $0 \ne \rho - \sigma = P - Q$. By Lemma \ref{Hahn-Jordan bound 2}, $P, Q \in \mcl{P}_k$. They are not both zero, and $PQ = 0$. Therefore, (\ref{stability of Pk-PRSO-b}) shows that $\Phi(P) \ne \Phi(Q)$. It follows that $\Phi(\rho - \sigma) \ne 0$, which is a contradiction to our assumption.
   \end{proof}
   
\begin{remark} \label{stability of Pk-PRSO-c}
    We note that part (\ref{stability of Pk-PRSO-b}) of Theorem (\ref{thm: stability of Pk-PRSO}) can be equivalently stated as:
    \\ 
    
 For every $\rho, \sigma \in \mcl{P}_k$ with $\rho \sigma = 0$ and $\n{\rho - \sigma}_1 = 1$, we have $\Phi(\rho) \not=\Phi(\sigma)$.
 
\end{remark}

  \begin{remark}
      
      In simple terms, Theorem \ref{thm: stability of Pk-PRSO} says that a super operator $\Phi$ is $\mcl{P}_k$-phase retrievable if and only if it distinguishes between any two positive semidefinite operators, not both zero, of rank at most $k$, with orthogonal images.
  \end{remark}

  Before we finally state the main results in this section, we give a useful characterization of $\mcl{P}_k$-phase retrievable super operators in terms of a simple quantity.
  It is analogous to the fact that a linear operator on a finite-dimensional normed space is injective if and only if the image of the unit sphere is bounded away from zero. Closely related arguments for stability notions appear in \cite[Sec. V]{MR3671475}, though only for measurements and not for arbitrary super operators.

\begin{lemma} \label{lem: Pk-PRSO equiv}
    Let $\mcl{X}$ and $\mcl{Y}$ be $\mbb{F}$-Euclidean spaces, and $\mcl{P}_k = \{\rho \in \Pos(\mcl{X}): \rank(\rho) \le k\}$. A super operator $\Phi: \Li(\mcl{X})\longrightarrow \Li(\mcl{Y})$ is $\mcl{P}_k$-phase retrievable if and only if
    \[
    \min\big\{ \n{\Phi(\rho - \sigma)}_1 \; :\; \rho, \sigma \in \mcl{P}_k, \rho \sigma = 0, \n{\rho}_1 +\n{\sigma}_1 = 1 \big\} > 0.
    \]
\end{lemma}

\begin{proof}
    Let $D = \{(\rho, \sigma): \rho, \sigma \in \mcl{P}_k, \rho \sigma = 0, \n{\rho - \sigma}_1 =  \n{\rho}_1 +\n{\sigma}_1 = 1\}$. Then $D$ is closed and bounded, so it is compact.
 
    Define a function $f : D \longrightarrow \mbb{R}_{\ge 0}$ by 
    $
    f(\rho, \sigma) =  \n{\Phi(\rho - \sigma)}_1,
    $
    then $f$ is clearly a continuous function.
    
    Suppose $\Phi$ is $\mcl{P}_k$-phase retrievable.  For all $(\rho, \sigma) \in D$, $\rho \ne \sigma$. Therefore, for all $(\rho, \sigma) \in D$, $\Phi(\rho - \sigma) \ne 0$. Hence, for every $(\rho, \sigma) \in D$, $f(\rho, \sigma) > 0$. By the continuity of $f$ and the compactness of $D$, it follows that $\underset{(\rho, \sigma) \in D}{\min} \n{\Phi(\rho - \sigma)}_1 > 0$.\\
    Conversely, if $\underset{(\rho, \sigma) \in D}{\min} \n{\Phi(\rho - \sigma)}_1 > 0$, then $f(\rho, \sigma) > 0$ for all $(\rho, \sigma) \in D$. Hence, for all $(\rho, \sigma) \in D$, we have $\Phi(\rho -\sigma) \ne 0$. By Theorem \ref{thm: stability of Pk-PRSO}, $\Phi$ is $\mcl{P}_k$-phase retrievable.
\end{proof}

\begin{remark}\label{rem: other norms}
    We note that Lemma \ref{lem: Pk-PRSO equiv} remains valid (with exactly the same proof) whenever the space $\Li(\mcl{X})$ is equipped with any norm, and $\Li(\mcl{Y})$ is replaced by any normed space. Furthermore, in the case when $\Phi$ is a quantum channel we will have
 \[
    \min\big\{ \n{\Phi(\rho - \sigma)}_1 \; :\; \rho, \sigma \in \mcl{P}_k, \rho \sigma = 0, \n{\rho}_1 +\n{\sigma}_1 = 1 \big\} \le 1.
    \]
    This follows from the observation that, by \cite[Cor.3.40]{watrous2018theory}, we have $\n{\Phi}_{1\to 1} = 1$. Moreover, by the definition of the operator norm, it holds for every $\rho, \sigma \in D$
    \[
    \n{\Phi(\rho-\sigma)}_1 \le \n{\Phi}_{1\to 1}\n{\rho -\sigma}_1 = \n{\rho}_1 + \n{\sigma}_1 = 1.
    \]
\end{remark}
   We now come to the main theorems of this section---the stability of $\mcl{P}_k$-phase retrievable super operators. The first one shows that if a super operator is $\mcl{P}_k$-phase retrievable, then up to some error, any close super operator will also be $\mcl{P}_k$-phase retrievable.

    \begin{theorem} \label{thm: stability for PkPR s-operators}
        Let $\mcl{X}$ and $\mcl{Y}$ be $\mbb{F}$-Euclidean spaces, $\mcl{P}_k = \{\rho \in \Pos(\mcl{X}) : \rank(\rho) \le k\}$ and $\Phi: \Li(\mcl{X}) \longrightarrow \Li(\mcl{Y})$ a $\mcl{P}_k$-phase retrievable super operator. Then there exists $\delta > 0$ such that for any super operator $\Psi : \Li(\mcl{X}) \longrightarrow \Li(\mcl{Y})$ satisfying $\n{\Phi -\Psi}_{1 \rightarrow 1} < \delta$, it follows that $\Psi$ is $\mcl{P}_k$-phase retrievable. 
    \end{theorem}
      \begin{proof}
        Let $D = \{(\rho, \sigma): \rho, \sigma \in \mcl{P}_k, \rho \sigma = 0, \n{\rho - \sigma}_1 =  \n{\rho}_1 +\n{\sigma}_1 = 1\}$ and $f$ the function defined in the proof of Lemma \ref{lem: Pk-PRSO equiv}.  Let $\epsilon = \underset{(\rho, \sigma) \in D}{\min}f(\rho, \sigma)$ and $\delta =\epsilon /2$. Then from Lemma \ref{lem: Pk-PRSO equiv}, $\epsilon > 0$, so that $\delta > 0$. Suppose that $\Psi: \Li(\mcl{X}) \longrightarrow \Li(\mcl{Y})$ is a super operator satisfying $\n{\Phi -\Psi}_{1 \rightarrow 1} < \delta$. Then for every $(\rho, \sigma) \in D$ we have :
        \begin{align*}
        \n{\Phi(\rho - \sigma)}_1 & \le  \n{(\Phi-\Psi)(\rho - \sigma)}_1 + \n{\Psi(\rho - \sigma)}_1\\
        & \le \n{\Phi - \Psi}_{1\rightarrow 1}\n{\rho - \sigma}_1 +\n{\Psi(\rho - \sigma)}_1\\
        & < \frac{\epsilon}{2} + \n{\Psi(\rho - \sigma)}_1.
        \end{align*}
        The first inequality follows from the triangle inequality, the second holds by the definition of the operator norm of $\Phi - \Psi$, and the third  follows because $\delta = \epsilon/2$, $\n{\Phi - \Psi}_{1 \rightarrow 1} < \delta$, and $\n{\rho - \sigma}_1 = 1$. 
       Taking the minimum over $(\rho, \sigma) \in D$, it follows that
       \[
       \underset{(\rho, \sigma) \in D}{\min}\n{\Psi(\rho - \sigma)}_1 > \frac{\epsilon}{2}.
       \]
       By Lemma \ref{lem: Pk-PRSO equiv}, $\Psi$ is $\mcl{P}_k$-phase retrievable.
    \end{proof}

    \begin{corollary} \label{cor: set PkPRSO is open}
      Let $\mcl{X}$ be an $\mbb{F}$-Euclidean space and $\mcl{P}_k = \{\rho \in \Pos(\mcl{X}) : \rank(\rho) \le k\}$. The set of super operators that are injective on $\mcl{P}_k$ is open.
    \end{corollary}

\begin{proof}
    Suppose $\Phi:\Li(\mcl{X}) \longrightarrow \Li(\mcl{Y})$ is a $\mcl{P}_k$ injective super operator. By Theorem \ref{thm: stability for PkPR s-operators} we can find $\delta > 0$ such that if a super operator $\Psi:\Li(\mcl{X}) \longrightarrow \Li(\mcl{Y})$ satisfies $\n{\Phi -\Psi}_{1 \rightarrow 1} < \delta $, then it is $\mcl{P}_k$-phase retrievable.
\end{proof}

The second theorem shows the stability of $\mcl{P}_k$-phase retrievable super operators with respect to perturbations of the inputs. It shows that if a super operator is $\mcl{P}_k$-phase retrievable then we can control the distance between the inputs (in $\mcl{P}_k$) of the super operator using the distance between their respective images, up to some constant.

\begin{theorem} \label{thm: stability Pk-PRSO-input}
    Let $k \in \mbb{N}$, let $\mcl{X}$ and $\mcl{Y}$ be $\mbb{F}$-Euclidean spaces, and $\mcl{P}_k =\{\rho \in \Pos(\mcl{X}) : \rank(\rho) \le k\}$. For every $\mcl{P}_k$-phase retrievable super operator $\Phi : \Li(\mcl{X}) \longrightarrow \Li(\mcl{Y})$, there exists $K_{\Phi} > 0$ such that for every $\rho, \sigma \in \mcl{P}_k$
    \begin{equation}\label{eqn: stability Pk-PRSO-input}    
    \n{\rho - \sigma}_1 \le K_{\Phi} \n{\Phi(\rho - \sigma)}_1.
    \end{equation}
    Therefore, a super operator $\Phi : \Li(\mcl{X}) \longrightarrow \Li(\mcl{Y})$ is $\mcl{P}_k$-phase retrievable if and only if its restriction to $\mcl{P}_k$ is a bi-Lipschitz map.
\end{theorem}

\begin{proof}   
   Assume $\Phi$ is a $\mcl{P}_k$-phase retrievable super operator. Let $D = \{(\omega,\vartheta): \omega,\vartheta \in \mcl{P}_k, \omega \vartheta = 0, \n{\omega - \vartheta}_1 = 1\}$. Let $\rho, \sigma \in \mcl{P}_k$. If $\rho = \sigma$ then the conclusion holds trivially. Now we consider the case where $\rho \ne \sigma$. Let $\varphi - \vartheta$ be the Jordan-Hahn decomposition of $\rho - \sigma$. Hence, $\varphi \vartheta = 0$. Lemma \ref{Hahn-Jordan bound 2} tells us that $\varphi, \vartheta \in \mcl{P}_k$. Let $\epsilon = \underset{(\omega, \vartheta) \in D}{\min}\n{(\Phi(\omega - \vartheta))}_1$. Since $(\varphi/\n{\varphi - \vartheta}_1, \vartheta /\n{\varphi - \vartheta}_1) \in D$ and $\n{\rho - \sigma}_1 = \n{\varphi - \vartheta}_1$ we have
   \[
   0 < \epsilon \le \n{\Phi\Bigg(\frac{\varphi}{\n{\varphi - \vartheta}_1} - \frac{\vartheta}{\n{\varphi - \vartheta}_1} \Bigg)}_1= \frac{\n{\Phi(\varphi - \vartheta)}_1}{\n{\varphi - \vartheta}_1}  = \frac{\n{\Phi(\rho - \sigma)}_1}{\n{\rho - \sigma}_1}.
   \]
   Taking $K_\Phi = 1/\epsilon$, we have $\n{\rho - \sigma}_1 \le K_\Phi \n{\Phi(\rho - \sigma)}_1$ as required.
\end{proof}

\begin{remark}
    We note that the constant $K_{\Phi}$ produced in Theorem~\ref{thm: stability Pk-PRSO-input} is optimal. 
That is, if a constant $C>0$ satisfies the condition, then necessarily $K_{\Phi} \le C$. This holds because for such $C$, given $\rho, \sigma \in D$ it follows that
\[
\frac{\n{\rho - \sigma}_1}{\n{\Phi(\rho-
\sigma)}_1} \le C.
\]
Choosing any $(\rho_0, \sigma_0) \in D$ that satisfies $\n{\Phi(\rho_0 - \sigma_0)}_1 = \epsilon$, we obtain 
\[
K_\Phi =\frac{1}{\epsilon} = \frac{\n{\rho_0 - \sigma_0}_1}{\n{\Phi(\rho_0-
\sigma_0)}_1} \le C.
\]

\end{remark}

\begin{remark}
Analogously to Remark \ref{rem: other norms},
note that versions of Lemma \ref{lem: Pk-PRSO equiv}, and Theorems \ref{thm: stability for PkPR s-operators} and \ref{thm: stability Pk-PRSO-input} are also valid, with the same proofs, with any other norms in place of the trace norm (but of course the values of the constants will change).
Similar Lipschitz stability results for $\mcl{P}_k$-phase retrieval, with respect to the Frobenius norm, were obtained in \cite[Prop. 4.1]{MR4477807} for the particular case of super operators associated to a Hermitian measurement (that is, of the form $\Psi_\mu$ for a measurement $\mu : \Sigma \to \Herm(\mcl{X})$).
\end{remark}

Worthy of note is the fact that the same approach of proof for Theorem \ref{thm: stability Pk-PRSO-input} works for any set $\mcl{P}$ which is ``\emph{closed under Jordan-Hahn decompositions}'' in the sense that for $\rho,\sigma \in \mcl{P}$ if we write the Jordan-Hahn decomposition $\rho-\sigma = \xi^+ - \xi^-$, then $\xi^+, \xi^- \in \mcl{P}$. As a specific example, if $G$ is a group and for some $\mbb{F}$-Euclidean space $\mcl{X}$, $\pi: G \to \U(\mcl{X})$ is a unitary representation of $G$, we define $\mcl{P}_k^\pi = \{\rho \in \Pos(\mcl{X}) : \rank(\rho)\le k \text{ and }  \pi(g) \rho \pi(g)^{-1} = \rho \text{ for all } g \in G\}$,
that is, the positive semidefinite operators of rank at most $k$ which commute with the range of $\pi$.
It holds that $\mcl{P}^\pi_k$ is closed under Jordan-Hahn decompositions. The theorem that follows gives a stability result for $\mcl{P}_k^\pi$-phase retrievable super operators.

\theorem \label{thm: PK-PRSO group rep}
Let $G$ be a group, $\mcl{X}$ an $\mbb{F}$-Euclidean space, $\U(\mcl{X})$ the set of unitary operators on $\mcl{X}$, and $\pi: G \to \U(\mcl{X})$ a group representation. Let $\mcl{P}_k^\pi = \{\rho \in \Pos(\mcl{X}) : \rank(\rho)\le k \text{ and }  \pi(g) \rho \pi(g)^{-1} = \rho \text{ for all } g \in G\}$. For every $\mbb{F}$-Euclidean space $\mcl{Y}$ and $\mcl{P}_k ^ \pi$-phase retrievable super operator $\Phi : \Li(\mcl{X}) \to \Li(\mcl{Y})$, there exists $K_\pi > 0$ such that for every $\rho, \sigma \in \mcl{P}_k^\pi$ we have
\[
\n{\rho - \sigma}_1 \le K_\pi \n{\Phi(\rho - \sigma)}_1.
\]

\begin{proof}
    Let $\rho, \sigma \in \mcl{P}_k^\pi$. If $\rho = \sigma$, the conclusion clearly holds. We assume that $\rho \ne \sigma$ and let $\rho - \sigma = \xi^+ - \xi^-$ be the Jordan-Hahn decomposition of $\rho - \sigma$. Then $\xi^+ \ne \xi^-$. Now $\xi^+ = \tfrac{(\rho - \sigma) + |\rho - \sigma|}{2}$ and $\xi^- = \tfrac{|\rho - \sigma|- (\rho - \sigma)}{2}$, where $|\rho-\sigma| = \sqrt{(\rho-\sigma)^*(\rho-\sigma)} = \sqrt{(\xi^+ - \xi^-)^2} = \xi^+ + \xi^-$. Let $\Gamma \subset \mbb{R}$ be the spectrum of $\rho - \sigma$. Define $\mbb{F}$-valued functions $h_1, h_2 : \Gamma \to \mbb{F}$ by $h_1(\lambda) = \tfrac{|\lambda|+\lambda}{2}$ and $h_2(\lambda) = \tfrac{|\lambda| - \lambda}{2}$. Then $h_1, h_2 \in C(\Gamma)$, where $C(\Gamma)$ is the set of complex-valued continuous functions on $\Gamma$. By \cite[Thm. 5.4.7]{buhler2018functional}, if $B\in \Li(\mcl{X})$ satisfies $B (\rho-\sigma) = (\rho-\sigma)B$, then $B f(\rho-\sigma) = f(\rho-\sigma) B$ for all $f \in C(\Gamma)$. Now $\xi^+ = h_1(\rho-\sigma), \xi^- = h_2(\rho-\sigma)$. Since $\rho - \sigma = \pi(g) (\rho-\sigma) \pi(g)^{-1} \iff (\rho-\sigma)\pi(g) = \pi(g) (\rho-\sigma)$ for all $g \in G$, it follows that $\xi^+ = \pi(g) \xi^+ \pi(g)^{-1}$ and $\xi^- = \pi(g) \xi^- \pi(g)^{-1}$ for all $g \in G$. From Lemma \ref{Hahn-Jordan bound 2}, it is clear that $\xi^+, \xi^- \in \mcl{P}_k$. Hence, we conclude that $\xi^+, \xi^- \in \mcl{P}_k^\pi$. \\
    If $\rho = \sigma$, the conclusion clearly holds. We assume that $\rho \ne \sigma$ and let $\rho - \sigma = \xi^+ - \xi^-$ be the Jordan-Hahn decomposition of $\rho - \sigma$. Then $\xi^+ \ne \xi^-$. Now $\xi^+ = \tfrac{(\rho - \sigma) + |\rho - \sigma|}{2}$ and $\xi^- = \tfrac{|\rho - \sigma|- (\rho - \sigma)}{2}$. Let $\Gamma \subset \mbb{R}$ be the spectrum of $\rho - \sigma$. Define $\mbb{F}$-valued functions $h_1, h_2 : \Gamma \to \mbb{F}$ by $h_1(\lambda) = \tfrac{|\lambda|+\lambda}{2}$ and $h_2(\lambda) = \tfrac{|\lambda| - \lambda}{2}$. Then $h_1, h_2 \in C(\Gamma)$, where $C(\Gamma)$ is the set of complex-valued continuous functions on $\Gamma$. By \cite[Thm. 5.4.7]{buhler2018functional}, if $B\in \Li(\mcl{X})$ satisfies $B (\rho-\sigma) = (\rho-\sigma)B$, then $B f(\rho-\sigma) = f(\rho-\sigma) B$ for all $f \in C(\Gamma)$. Now $\xi^+ = h_1(\rho-\sigma), \xi^- = h_2(\rho-\sigma)$. Since $\rho - \sigma = \pi(g) (\rho-\sigma) \pi(g)^{-1} \iff (\rho-\sigma)\pi(g) = \pi(g) (\rho-\sigma)$ for all $g \in G$, it follows that $\xi^+ = \pi(g) \xi^+ \pi(g)^{-1}$ and $\xi^- = \pi(g) \xi^- \pi(g)^{-1}$ for all $g \in G$. From Lemma \ref{Hahn-Jordan bound 2}, it is clear that $\xi^+, \xi^- \in \mcl{P}_k$. Hence, we conclude that $\xi^+, \xi^- \in \mcl{P}_k^\pi$. \\
    
    If $\mcl{Y}$ is any $\mbb{F}$-Euclidean space and $\Phi : \Li(\mcl{X}) \to \Li(\mcl{Y})$ is any $\mcl{P}_k^\pi$-phase retrievable super operator, then applying the same ideas used in the proof of Theorem \ref{thm: stability Pk-PRSO-input} it follows that we can find $K_{\pi}$ such that $\n{\rho - \sigma}_1 \le K_{\pi} \n{\Phi(\rho - \sigma)}_1.$
\end{proof}

\section{Stability of $\mcl{P}_k$-Phase Retrievable Measurements} \label{Sec: Stability of PKPRM}

Since $\mcl{P}_k$-phase retrievable super operators are stable, it is natural to also investigate the stability of measurements. For a set $\Omega \subseteq \Li(\mcl{X})$, Theorem \ref{thm: equivalence of P-phase retrievable measurements and channels} establishes that a measurement is $\Omega$-phase retrievable precisely when an associated super operator defined in (\ref{eqn: measurement SO}) is $\Omega$-phase retrievable. Hence for a finite set $\Sigma$ and an $\mbb{F}$-Euclidean space $\mcl{X}$, a measurement $\mu : \Sigma \longrightarrow \Li(\mcl{X})$ is not $\mcl{P}_k$-phase retrievable if and only if the associated super operator $\Psi_\mu$ defined in (\ref{eqn: measurement SO}) is not injective on $\mcl{P}_k$. In what follows, we use this correspondence to translate the results of Section \ref{Sec: Stability of PKPRMSO} into parallel statements for measurements.

\begin{theorem} 
    Let $k \in \mbb{N}$, $\mcl{X}$ an $\mbb{F}$-Euclidean space, $\Sigma$ a finite set, $\mcl{P}_k \subset \Pos(\mcl{X})$ be the set of positive semidefinite operators of rank at most $k$, and $\mu : \Sigma \longrightarrow \Li(\mcl{X})$ a measurement. The following are equivalent:
    \begin{enumerate}[(a)]
        \item \label{stability of Pk-PRM-a} 
        The measurement $\mu$ is $\mcl{P}_k$-phase retrievable.
        
        \item \label{stability of Pk-PRM-b} 
        For every $\rho, \sigma \in \mcl{P}_k$, not both zero, such that $\rho\sigma=0$, there exists $a \in \Sigma$ such that $\pair{\mu(a)}{\rho} \ne \pair{\mu(a)}{\sigma}$.
    \end{enumerate}
\end{theorem}
The next result is an immediate consequence of the previous one, in the same way that Corollary \ref{cor:k-phase retrievable s-operators} follows from Proposition \ref{prop: composed meas Herm} and Corollary \ref{cor: KPR Herm M}.

  \begin{corollary} \label{cor: Pk-phase retrievable ptps operator}
       Let $\mcl{X}$ and $\mcl{Y}$ be $\mbb{F}$-Euclidean spaces, $\Sigma$ a finite set, and $\mcl{P}_k = \{\rho \in \Pos(\mcl{X}): \rank(\rho) \le k \}$. Let $\mu : \Sigma \longrightarrow \Li(\mcl{Y})$ a measurement and $\Phi : \Li(\mcl{X}) \longrightarrow \Li(\mcl{Y})$ be a super operator. The following are equivalent:
    \begin{enumerate}[(a)]
        \item The measurement $\Phi^* \circ \mu$ is $\mcl{P}_k$-phase retrievable.
        \item For every $\rho, \sigma \in \mcl{P}_k$ not both zero, with $\rho \sigma = 0$, there exists $a \in \Sigma$ such that $\pair{\mu(a)}{\Phi(\rho)} \ne \pair{\mu(a)}{\Phi(\sigma)}$.   
    \end{enumerate}
  \end{corollary}

  \begin{remark}
      In simple terms, Corollary \ref{cor: Pk-phase retrievable ptps operator} says that the composed measurement $\Phi^* \circ \mu$ is $\mcl{P}_k$-phase retrievable if and only if there exists a $\mu$ such that the measurement $\Phi ^ {*} \circ \mu$ is $\mcl{P}_k$-phase retrievable if and only if distinguishes between any two positive semidefinite operators of rank at most $k$, not both zero, with orthogonal images.
  \end{remark}

For the next result, we note that given a finite set $\Sigma$, an $\mbb{F}$-Euclidean space $\mcl{X}$, and a measurement $\mu : \Sigma \longrightarrow \Li(\mcl{X})$, we can deduce from the definition of the corresponding super operator $\Psi_\mu : \Li(\mcl{X}) \longrightarrow \Li(\mbb{C}^ \Sigma)$ that $\n{\Psi_\mu(X - Y)}_1 = \n{(\pair{\mu(a)}{X-Y})_{a \in \Sigma}}_{\ell_1}$ for all $X,Y \in \Li(\mcl{X})$. We use this idea to relate Lemma \ref{lem: Pk-PRSO equiv} and the lemma that follows.

\begin{lemma} \label{lem: Pk-PRM equiv}
    Let $\Sigma$ be a finite set, $\mcl{X}$ an $\mbb{F}$-Euclidean space, and $\mcl{P}_k = \{\rho \in \Pos(\mcl{X}): \rank(\rho) \le k\}$. A measurement $\mu : \Sigma \longrightarrow \Li(\mcl{X})$ is $\mcl{P}_k$-phase retrievable if and only if
    \[
 \min\{{\n{(\pair{\mu(a)}{\rho - \sigma})_{a \in \Sigma}}_{\ell_1}}: \rho, \sigma \in \mcl{P}_k, \rho\sigma = 0, \n{\rho}_1+\n{\sigma}_1 = 1\}>0.
    \]
\end{lemma}

We now state the two main theorems of the section, which follow in correspondence from the two main theorems in Section \ref{Sec: Stability of PKPRMSO}. In the first, we show stability in the sense that if we replace a $\mcl{P}_k$-phase retrievable measurement with another one \emph{close enough}, then the new measurement will also be $\mcl{P}$-phase retrievable. But we need to discuss what it means to be ``close enough."
For a fixed finite set $\Sigma$, $\mbb{F}$-Euclidean space $\mcl{X}$ and measurements $\mu, \nu : \Sigma \longrightarrow \Li(\mcl{X})$ , a natural way to quantify the distance between $\mu$ and $\nu$ is through a suitable norm. Below, we make a convenient choice of norm that will readily relate to the results in Section \ref{Sec: Stability of PKPRMSO}.
\\

Let $\Sigma$ be a finite set, and $\mcl{X}$ an $\mbb{F}$-Euclidean space. Let $\mcl{M}_{\Sigma}$ represent the set of all measurements of the form $\mu : \Sigma \longrightarrow \Li(\mcl{X})$. Define a function $d: \mcl{M}_{\Sigma} \times \mcl{M}_{\Sigma} \longrightarrow \mbb{R}_{\ge 0}$ by 

\[
    d(\mu, \nu) = \n{\Psi_\mu - \Psi_\nu}_{1 \to 1}.
\]

It is readily verified that $d$ is a metric on $\mcl{M}_{\Sigma}$. Consequently, ($\mcl{M}_{\Sigma},d$) forms a metric space, allowing us to quantify the distance between any two measurements defined on $\Sigma$ with values in $\Li(\mcl{X})$. We use the metric $d$ to measure the closeness of two measurements $\mu, \nu \in \mcl{M}_\Sigma$. The following theorem says that if we have a $\mcl{P}_k$-phase retrievable measurement $\mu$, then up to some error $\delta > 0$, any measurement $\nu$ whose distance from $\mu$ is less than $\delta$ is also $\mcl{P}_k$-phase retrievable. 
\\
\begin{theorem} \label{thm: stability of Pk-PRM}
    Let $\mcl{X}$ be an $\mbb{F}$-Euclidean space, $\Sigma$ a finite set and $\mcl{M}_\Sigma = \{\mu: \Sigma \longrightarrow \Li(\mcl{X})\}$ the set of measurements on $\Sigma$ with image in $\Li(\mcl{X})$. For every $\mcl{P}_k$-phase retrievable measurement $\mu \in \mcl{M}_{\Sigma}$ there exists $\delta > 0$ such that if $\nu \in \mcl{M}_{\Sigma}$ satisfies $d(\mu, \nu) < \delta$, then $\nu$ is also $\mcl{P}_k$-phase retrievable.
\end{theorem}
   \begin{proof}
    By Theorem \ref{thm: equivalence of P-phase retrievable measurements and channels}, the assumption that $\mu \in \mcl{M}_\Sigma$ is $\mcl{P}_k$-phase retrievable implies that the corresponding super operator $\Psi_\mu$ defined in \eqref{eqn: measurement SO} is $\mcl{P}_k$-phase retrievable. By Theorem \ref{thm: stability for PkPR s-operators}, there exists $\delta > 0$ such that any super operator $\Phi$ that satisfies $\n{\Psi_\mu - \Phi}_{1 \rightarrow 1} < \delta$ is $\mcl{P}_k$-phase retrievable. If $\nu \in \mcl{M}_\Sigma$ satisfies $d(\mu,\nu)<\delta$, then 
    $\n{\Psi_\mu - \Psi_\nu}_{1\to 1} < \delta$, so $\Psi_\nu$ is $\mcl{P}_k$-phase retrievable. Theorem~\ref{thm: equivalence of P-phase retrievable measurements and channels} then yields that $\nu$ is $\mcl{P}_k$-phase retrievable.

\end{proof}

\begin{corollary} \label{cor: PKPRM set open}
    Let $\Sigma$ be a finite set, and $\mcl{X}$ an $\mbb{F}$-Euclidean space. The set of $\mcl{P}_k$-phase retrievable measurements of the form $\mu : \Sigma \longrightarrow \Li(\mcl{X})$ is open with respect to the metric $d$.
\end{corollary}

\begin{proof}
    Let $\mu: \Sigma \longrightarrow \Li(\mcl{X})$ be a $\mcl{P}_k$-phase retrievable measurement. By Theorem \ref{thm: stability of Pk-PRM}, there exists $\delta > 0$ such that if $\nu : \Sigma \longrightarrow \Li(\mcl{X})$ satisfies $d(\mu, \nu) < \delta$, then $\nu$ is $\mcl{P}_k$-phase retrievable.
\end{proof}

The second main result shows the stability of a $\mcl{P}_k$-phase retrievable measurement with respect to elements of $\mcl{P}_k$ (regarded as inputs). 

\begin{theorem} \label{thm: PK-PRSO measurement}
    Let $k \in \mbb{N}$, $\mcl{X}$ an $\mbb{F}$-Euclidean space, $\Sigma$ a finite set and $\mcl{P}_k = \{\rho \in \Pos(\mcl{X}): \rank(\rho) \le k\}$. For every $\mcl{P}_k$-phase retrievable measurement $\mu : \Sigma \longrightarrow \Li(\mcl{X})$, there exists $K_\mu > 0$ such that for all $\rho, \sigma \in \mcl{P}_k$ we have
\[
\n{\rho -\sigma}_1 \le K_\mu \n{(\pair{\mu(a)}{\rho - \sigma})_{a \in \Sigma}}_{\ell_1}.
\]
\end{theorem}

\begin{proof}
    Assume $\mu$ is $\mcl{P}_k$-phase retrievable. Then the super operator $\Psi_\mu$ corresponding to $\mu$ is $\mcl{P}_k$-phase retrievable, and we have that $\n{\Psi_\mu(X)}_1 = \n{(\pair{\mu(a)}{X})_{a \in \Sigma}}_{\ell_1}$ for all $X \in \Li(\mcl{X})$. By Theorem \ref{thm: stability Pk-PRSO-input}, there exists $K_\mu > 0$ such that for all $\rho, \sigma \in \mcl{P}_k$, we have $\n{(\rho-\sigma)}_1 \le K_\mu \n{(\pair{\mu(a)}{\rho-\sigma})_{a \in \Sigma}}_1$.
\end{proof}

The last result in this section is the measurement version of Theorem \ref{thm: PK-PRSO group rep} and its proof is similar to that of Theorem \ref{thm: PK-PRSO measurement}.

\begin{theorem} \label{thm:PK-PRSO group rep measurement}
    Let $G$ be a group, $\mcl{X}$ an $\mbb{F}$-Euclidean space, $\Sigma$ a finite set, $\pi:G\to \U(\mcl{X})$ a unitary representation, and $\mcl{P}_k^\pi = \{\rho \in \Pos(\mcl{X}): \rank(\rho)\le k, \text{ and } \pi(g) \rho \pi(g) ^ {-1} = \rho \text{ for all } g \in G\}$. For every $\mcl{P}_k^\pi$-phase retrievable measurement $\mu:\Sigma \to \Li(\mcl{X})$, there exists $C_\mu > 0$ such that for all $\rho, \sigma \in \mcl{P}_k^\pi$, we have
    \[
    \n{\rho-\sigma}_1 \le C_\mu \n{(\pair{\mu(a)}{\rho-\sigma})_{a \in \Sigma}}_{\ell_1}.
    \]
\end{theorem}

\begin{remark}
    Before concluding this section, it is worth emphasizing that the results established here, together with their counterparts in the previous section, reveal the close correspondence between the stability properties of $\mcl{P}_k$-phase retrievable of measurements and those of $\mcl{P}_k$-phase retrievable super operators. In this sense, the two theories are essentially two sides of the same coin.
\end{remark}

\section{Stability of $\mcl{P}_k$-Phase Retrievability with respect to the Bures--Wasserstein distance} \label{Sec: Stability BW Distance}

In Section \ref{Sec: Xter PRMSO & PRQM}, we briefly introduced the Bhattacharyya coefficient and described it as a classical analogue of quantum fidelity. The quantum fidelity between two positive semidefinite operators is a non‑negative quantity that quantifies their degree of similarity. In particular, it reflects the extent to which the two operators overlap: see below for a precise statement. 
For any pair of positive semidefinite operators $\rho$ and $\sigma$ defined on some $\mbb{F}$-Euclidean space $\mcl{X}$, the quantum fidelity of $\rho$ and $\sigma$, denoted $F(\rho, \sigma)$, is defined as
\[
F(\rho, \sigma) = \n{\sqrt{\rho} \sqrt{\sigma}}_1.
\]
Using the definition of trace norm, it follows that an equivalent definition for quantum fidelity is
\[
F(\rho, \sigma) = \Tr\Big (\sqrt{\sqrt{\rho}\sigma \sqrt{\rho}} \Big).
\]
For any positive semidefinite operators $\rho$ and $\sigma$, it is always the case that $0 \le F(\rho, \sigma) \le \sqrt{\Tr(\rho) \Tr(\sigma)}$. We have $F(\rho, \sigma) = 0$ if and only if $\rho \sigma = 0$, and $F(\rho, \sigma) = \sqrt{\Tr(\rho) \Tr(\sigma)}$ if and only if $\rho$ and $\sigma$ are linearly dependent. It follows that if $\rho, \sigma$ are states, we have $0 \le F(\rho, \sigma)\le 1$, and furthermore the extreme cases are nicely characterized: $F(\rho, \sigma) = 1$ if and only if $\rho = \sigma$, and $F(\rho, \sigma) = 0$ if and only if their ranges are orthogonal. For more information on quantum fidelity, see \cite[Chapter 3]{watrous2018theory}.

In spite of its many nice qualities, it is easy to see that the quantum fidelity function does not define a metric on $\Pos(\mcl{X})$. However, the function $d_{BW} : \Pos(\mcl{X}) \times \Pos(\mcl{X}) \to \mbb{R}_{\ge 0}$ given by
\[
d_{BW}(\rho, \sigma) = \sqrt{\Tr(\rho) + \Tr(\sigma) - 2 F(\rho, \sigma)},
\]
does define a metric on $\Pos(\mcl{X})$. This function is called the \emph{Bures--Wasserstein distance} between the positive semidefinite operators $\rho$ and $\sigma$. 
For more details on this distance, see \cite{MR3992484}.
Thus, the quantum fidelity function can be used to construct an actual metric on the set of positive semidefinite operators. In the setting of states, the standard inequalities relating quantum fidelity and the Bures--Wasserstein distance are the Fuchs–van de Graaf inequalities, which, as stated in \cite[Thm. 3.33]{watrous2018theory}, assert that
\[
2 - 2F(\rho, \sigma) \le \n{\rho - \sigma}_1 \le 2 \sqrt{1- F(\rho, \sigma)^2},
\]
for any states $\rho, \sigma \in \D(\mcl{X})$.

In this section, we investigate phase retrievability with respect to the Bures--Wasserstein distance between positive semidefinite operators. The key reason for being interested in this particular form of stability is the fact that the known results for stability of phase retrieval in Frame Theory can be understood in terms of this distance. 
To wit, as mentioned in the introduction, the following stability result has been proved in a variety of ways:
if $\{y_a\}_{a\in\Sigma}$ is a collection of vectors in the Euclidean space $\mcl{X}$ which does phase retrieval (in the sense that the map $x \mapsto \big( |\pair{y_a}{x}| \big)_{a \in \Sigma}$ is injective up to a phase), then there exists a constant $C > 0$ such that for every $x,z \in \mcl{X}$ we have
\[
\inf_{|\lambda|=1} \n{x-\lambda z} \le C \n{\big( |\pair{y_a}{x}|-|\pair{y_a}{z}|  \big)_{a \in \Sigma}}_{\ell_2}.
\]
Straightforward calculations show that the inequality above is equivalent to
\[
d_{BW}(xx^*,zz^*) \le C d_{BW}(\Psi_\mu(xx^*),\Psi_\mu(zz^*))
\]
where $\Psi_\mu : \L(\mcl{X}) \to \L(\mathbb{F}^\Sigma)$ is the super operator associated, in the sense of \eqref{eqn 2.4}, to the measurement $\mu : \Sigma \to \L(\mcl{X})$ given by $\mu(a) = y_ay_a^*$; this makes it abundantly clear that the aforementioned stability of phase retrieval in Frame Theory can be understood as Lipschitz stability of retrieval with respect to the Bures--Wasserstein distance.

The main ingredients for this section are inequalities connecting the Bures--Wasserstein distance and the quantum fidelity for arbitrary pairs of positive semidefinite operators, which are just a non-normalized version of the Fuchs--van de Graaf inequalities.
Before we delve into the main results we state an auxiliary lemma.
\begin{lemma} \label{lem: Trace-fidelity}
    Let $\mcl{X}$ be an $\mbb{F}$-Euclidean space. Then $\Tr(\sqrt{\rho}\sqrt{\sigma}) \le F(\rho,\sigma)$ for all $\rho, \sigma \in \Pos(\mcl{X})$. 
\end{lemma}
\begin{proof}
    First, denoting the set of unitary operators on $\mcl{X}$ by $\U(\mcl{X})$, it follows from \cite[Equation 1.182]{watrous2018theory} that $\n{X}_1 = \max \{|\Tr(UX)|: U \in \U(\mcl{X})\}$  for all $X \in \Li(\mcl{X})$. Therefore, given $\rho, \sigma \in \Pos(\mcl{X})$, it follows that $\Tr(\sqrt{\rho}\sqrt{\sigma}) = \Tr(\mathbbm{1}_\mcl{X}\sqrt{\rho}\sqrt{\sigma} ) \le  |\Tr(\mathbbm{1}_\mcl{X}\sqrt{\rho}\sqrt{\sigma} )|\le \n{\sqrt{\rho}\sqrt{\sigma}}_1$.
\end{proof}

\begin{theorem}[Generalized Fuchs--van de Graaf]  \label{thm: BW Fvd Graaf}
    Let $\mcl{X}$ be an $\mbb{F}$-Euclidean space. If $\rho, \sigma \in \Pos(\mcl{X})$ then
    \[
    \n{\rho - \sigma}_1 \le d_{BW}(\rho,\sigma) \sqrt{\Tr(\rho) +\Tr(\sigma)+2F(\rho, \sigma)} \quad\text{ and }\quad d_{BW}(\rho, \sigma)^2 \le \n{\rho - \sigma}_1.
    \]
\end{theorem}

\begin{proof}
The first inequality follows from simple algebraic manipulations of \cite[Lemma S5]{coles2014equivalence}, see also \cite[Lemma 46]{MR4126881} for a different proof. 

    For the second one, we first note from \cite[Lem. 3.34]{watrous2018theory} that $\n{\rho-\sigma}_1 \ge \n{\sqrt{\rho}-\sqrt{\sigma}}_2^2$. From Lemma \ref{lem: Trace-fidelity} and the definition of the Frobenius norm it follows that
    \[
    \begin{aligned}
        \n{\sqrt{\rho}-\sqrt{\sigma}}_2^2 & = \Tr((\sqrt{\rho}-\sqrt{\sigma})^2) = \Tr(\rho) + \Tr(\sigma) - 2 \Tr(\sqrt{\rho}\sqrt{\sigma}) \\ & \ge \Tr(\rho) + \Tr(\sigma) - 2 F(\rho,\sigma) = d_{BW}(\rho,\sigma)^2
    \end{aligned}
    \]
    
    Therefore, $d_{BW}(\rho,\sigma)^2 \le \n{\rho - \sigma}_1$.
\end{proof}

Using Theorem \ref{thm: BW Fvd Graaf} and Theorem \ref{thm: stability Pk-PRSO-input} we immediately translate our stability result with respect to norms into one involving Bures--Wasserstein distances.

\begin{corollary} \label{cor: holder stability}
    \label{cor: norm stability plus F-vdG}
Let $\mcl{X}$ and $\mcl{Y}$ be $\mbb{F}$-Euclidean spaces, and $\Phi : \L(\mcl{X}) \to \L(\mcl{Y})$ a positive $\mcl{P}_k$-phase retrievable super operator.
Then there exists a constant $C > 0$ such that for every $\rho,\sigma \in \mcl{P}_k$ we have
\begin{equation}\label{eqn:stability Pk norm-to-BW cheap}
    d_{BW}(\rho,\sigma)^2\le \n{\rho-\sigma}_1 \le C d_{BW}\big( \Phi(\rho), \Phi(\sigma) \big) \big(\n{\Phi(\rho)}^{1/2}_1+\n{\Phi(\sigma)}^{1/2}_1 \big). 
\end{equation}
Thus, $\mcl{P}_k$-phase retrieval is $\frac{1}{2}$-H\"older stable with respect to the Bures--Wasserstein distance on bounded subsets of $\mcl{P}_k$, and in particular it is so for $\D_k$-retrieval.
\end{corollary}

\begin{proof}
    The first inequality already holds from Theorem \ref{thm: BW Fvd Graaf}, so we focus on the second. From Theorem \ref{thm: stability Pk-PRSO-input} there exists a constant $C>0$ with $\n{\rho - \sigma}_1 \le C \n{\Phi(\rho)-\Phi(\sigma)}_1$ for every $\rho,\sigma \in \mcl{P}_k$. From Theorem \ref{thm: BW Fvd Graaf} it follows that 
    \[
    C \n{\Phi(\rho)-\Phi(\sigma)}_1 \le C d_{BW}(\Phi(\rho),\Phi(\sigma))\sqrt{\Tr(\Phi(\rho)) +\Tr(\Phi(\sigma))+2F(\Phi(\rho), \Phi(\sigma))}.
    \] 
    Now $\Phi$ is a positive super operator, $\n{\omega}_1 = \Tr(\omega)$ for every $\omega \in \Pos(\mcl{Y})$, and quantum fidelity satisfies $F(\rho, \sigma)^2 \le \Tr(\rho) \Tr(\sigma)$. Using these facts we obtain
    \[
    \sqrt{\Tr(\Phi(\rho)) +\Tr(\Phi(\sigma))+2F(\Phi(\rho), \Phi(\sigma))} \le \n{\Phi(\rho)}_1^{1/2} + \n{\Phi(\sigma)}_1^{1/2}.
    \]
    Therefore, $\n{\rho-\sigma}_1 \le C d_{BW}\big( \Phi(\rho), \Phi(\sigma) \big) \big(\n{\Phi(\rho)}^{1/2}_1+\n{\Phi(\sigma)}^{1/2}_1 \big)$, and this completes the proof of \eqref{eqn:stability Pk norm-to-BW cheap}.

For the second part of the conclusion, if we restrict to a subset of $\Pos(\mcl{X})$ which is bounded in trace norm, say by a number $L$, then from the definition of the operator norm and the inequality on the right in \eqref{eqn:stability Pk norm-to-BW cheap} we deduce that
\[
d_{BW}(\rho, \sigma)^2 \le C d_{BW}(\Phi(\rho),\Phi(\sigma)) (\n{\Phi}_{1\to 1}^{1/2} L^{1/2} + \n{\Phi}_{1 \to 1}^{1/2}L^{1/2}) = 2C\sqrt{L\n{\Phi}_{1\to 1} } d_{BW}(\Phi(\rho),\Phi(\sigma))
\]
The desired statement about stability follows immediately.
\end{proof}

Other stability results for $\mcl{P}_k$-phase retrieval with respect to the Bures--Wasserstein distance were obtained in \cite{MR4477807}, where a very thorough analysis of said stability was performed. However, the results of \cite{MR4477807} are restricted to the particular case of super operators associated to a Hermitian measurement (that is, of the form $\Psi_\mu$ for a measurement $\mu : \Sigma \to \Herm(\mcl{X})$).
Previous works in a similar vein (i.e. dealing with Lipschitz stability of phase retrieval) also include 
\cite{MR3464067,MR3902102,MR4062062} and the very recent \cite{MR5003156}.
See \cite{MR3526436} for a nice survey of results for phase retrieval in Frame Theory including ones about stability.

Note that for $k>1$ it follows from \cite[Thm. 4.9]{MR4477807} that there are super operators for which we cannot have Lipschitz stability in Corollary \ref{cor: norm stability plus F-vdG}; thus, it is not surprising that we have only obtained H\"older stability.
In the case $k=1$, below we slightly refine the argument to get general Lipschitz stability with respect to the Bures--Wasserstein distance (which, as explained above, generalizes the well-known stability of phase retrieval in Frame Theory).

\begin{theorem}\label{thm: stability BW}
Let $\mcl{X}$ and $\mcl{Y}$ be $\mbb{F}$-Euclidean spaces, and  $\Phi : \L(\mcl{X}) \to \L(\mcl{Y})$ a positive $\mcl{P}_1$-phase retrievable super operator.
Then there exists a constant $C > 0$ such that for every $\rho,\sigma \in \mcl{P}_1$ we have
\[
d_{BW}(\rho,\sigma) \le C d_{BW}\big( \Phi(\rho), \Phi(\sigma) \big).
\]
\end{theorem}

\begin{proof}
Assume $\rho, \sigma \in \mcl{P}_1$. Clearly, the inequality holds if $\rho = \sigma$, so we assume $\rho \ne \sigma$. 

If $\rho$ and $\sigma$ are both rank at most one, it is well-known (see \cite[Eqn. (1.184)]{watrous2018theory}) that
\[
\n{\rho-\sigma}_1 = \sqrt{\big(\Tr(\rho)+\Tr(\sigma) \big)^2 - 4F(\rho,\sigma)^2} = d_{BW}(\rho,\sigma) \sqrt{ \n{\rho}_1+\n{\sigma}_1 + 2F(\rho,\sigma) }.
\]
It follows from the above equality and Corollary \ref{cor: norm stability plus F-vdG} that there exists a constant $C>0$ such that for all $\rho,\sigma \in \mcl{P}_1$ with $\rho\not= \sigma$ we have,
\begin{align*}
    d_{BW}(\rho,\sigma) & \le C d_{BW}\big( \Phi(\rho), \Phi(\sigma) \big) \frac{ (\n{\Phi(\rho)}^{1/2}_1+\n{\Phi(\sigma)}^{1/2}_1)}{\sqrt{ \n{\rho}_1+\n{\sigma}_1 + 2 F(\rho, \sigma)}}\\
    &\le C d_{BW}\big( \Phi(\rho), \Phi(\sigma) \big) \frac{ (\n{\Phi(\rho)}^{1/2}_1+\n{\Phi(\sigma)}^{1/2}_1)}{\sqrt{ \n{\rho}_1+\n{\sigma}_1}} \\
    &\le  \sqrt{2}C d_{BW}\big( \Phi(\rho), \Phi(\sigma) \big)  \frac{ \sqrt{\n{\Phi(\rho)}_1+\n{\Phi(\sigma)}_1}}{\sqrt{ \n{\rho}_1+\n{\sigma}_1}} \\
    &\le  \sqrt{2}C \n{\Phi}_{1\to 1}^{1/2} d_{BW}\big( \Phi(\rho), \Phi(\sigma) \big).
\end{align*}
The second inequality holds because $\sqrt{ \n{\rho}_1+\n{\sigma}_1 + 2 F(\rho, \sigma)} \ge \sqrt{ \n{\rho}_1+\n{\sigma}_1}$. The third follows from the Cauchy-Schwarz inequality. The last is true because $\tfrac{\n{\Phi(\rho)}_1 + \n{\Phi(\sigma)}_1}{\n{\rho}_1 + \n{\sigma}_1} \le  \tfrac{\n{\Phi}_{1 \to 1}(\n{\rho}_1 + \n{\sigma}_1)}{\n{\rho}_1 + \n{\sigma}_1} =\n{\Phi}_{1 \to 1}$.
\end{proof}

\section*{Acknowledgement}

The authors thank Drs. Deguang Han and Kai Liu for their feedback on an earlier version of this paper, which allowed us to greatly improve some results (particularly Theorems \ref{thm: Non-injective PKPR QC} and \ref{thm: PKPR not PK+1PR QC}).

\section*{Generative AI disclosure}

The authors used ChatGPT by OpenAI to perform some calculations, to search for references, and to check the manuscript for typos.
AI was neither used to develop the ideas, nor to generate proofs in this project.

\bibliography{references}
\bibliographystyle{amsalpha} 
\end{document}